%% file: arxiv_surface_tension.tex
\documentclass[11pt,a4paper,twoside]{article}

\usepackage[T1]{fontenc}
\usepackage[utf8]{inputenc}

\usepackage{amsmath}
\usepackage{mathtools}
\usepackage{bm}
\usepackage{amsthm}

\usepackage{newtxtext}
\usepackage{newtxmath}

\usepackage[
  a4paper,
  margin=28mm,
  headheight=14pt
]{geometry}

\usepackage{microtype}
\usepackage{setspace}
\usepackage{graphicx}
\usepackage{booktabs}
\usepackage{siunitx}
\usepackage{float}
\usepackage{tabularx}
\usepackage{array}
\usepackage{adjustbox}
\usepackage{subcaption}
\usepackage[ruled,vlined]{algorithm2e}

\usepackage{tikz}
\usetikzlibrary{3d}
\usetikzlibrary{arrows.meta}

\usepackage[
  font=small,
  labelfont=bf,
  labelsep=period
]{caption}

\usepackage[round,authoryear]{natbib}
\newtheorem{theorem}{Theorem}[section]
\newtheorem{lemma}[theorem]{Lemma}
\newtheorem{proposition}[theorem]{Proposition} 
\newtheorem{corollary}[theorem]{Corollary}

\newtheorem{remark}[theorem]{Remark}
\newtheorem{definition}[theorem]{Definition}

\usepackage{titlesec}

\titleformat{\section}
  {\normalfont\large\bfseries}
  {\thesection.}
  {0.6em}
  {}

\titleformat{\subsection}
  {\normalfont\normalsize\bfseries}
  {\thesubsection.}
  {0.6em}
  {}

\titleformat{\subsubsection}
  {\normalfont\normalsize\itshape}
  {\thesubsubsection.}
  {0.6em}
  {}

\titleformat{\paragraph}
  {\normalfont\normalsize\bfseries}
  {\theparagraph}
  {1em}
  {}

\titlespacing*{\paragraph}
  {0pt}
  {3.25ex plus 1ex minus .2ex}
  {1.5ex plus .2ex}
  
\usepackage{fancyhdr}

\newcommand{\shorttitle}{Surface-tension calibration for N-phase mixtures}
\newcommand{\shortauthors}{ten Eikelder and Brunk}

\fancypagestyle{plain}{%
  \fancyhf{}%
  \fancyfoot[C]{\thepage}%
}
\usepackage{authblk}

\usepackage{xcolor}
\definecolor{darkgreen}{rgb}{0, 0.7, 0}
\definecolor{orange}{rgb}{0.98, 0.6, 0.01}
\definecolor{napiergreen}{rgb}{0.16, 0.5, 0.0}

\usepackage{pifont}

\usepackage{accents}
\newlength{\dhatheight}

\numberwithin{equation}{section}
\allowdisplaybreaks

\input{declarations}

\usepackage[
  colorlinks=true,
  linkcolor=blue,
  citecolor=blue,
  urlcolor=blue
]{hyperref}

\usepackage[nameinlink,capitalise,noabbrev]{cleveref}

\newcommand{\email}[1]{\href{mailto:#1}{\texttt{#1}}}

\makeatletter
\newcommand{\myref}[1]{\cref{#1}\mynameref{#1}{\csname r@#1\endcsname}}
\newcommand{\Myref}[1]{\Cref{#1}\mynameref{#1}{\csname r@#1\endcsname}}

\newcolumntype{R}[2]{%
    >{\adjustbox{angle=#1,lap=\width-(#2)}\bgroup}%
    l%
    <{\egroup}%
}

\newcommand{\thickhline}{%
    \noalign {\ifnum 0=`}\fi \hrule height 1pt
    \futurelet \reserved@a \@xhline
}

\newcolumntype{"}{@{\hskip\tabcolsep\vrule width 1pt\hskip\tabcolsep}}
\makeatother

\title{\vspace{-1.0cm}
\bfseries
Surface-tension calibration for N-phase mixtures
}

\author[1]{Marco F.P. ten Eikelder\thanks{\email{marco.eikelder@tu-darmstadt.de}}}
\author[2]{Aaron Brunk\thanks{\email{abrunk@uni-mainz.de}}}

\affil[1]{Institute for Mechanics, Computational Mechanics Group, Technical University of Darmstadt, Germany}
\affil[2]{Institute of Mathematics, Johannes Gutenberg-University Mainz, Germany}

\date{}

\begin{document}

\maketitle
\thispagestyle{plain}

\begin{abstract}
Diffuse-interface (phase-field) models are a widely used framework for interfacial dynamics in complex fluids, in which sharp interfaces are replaced by smooth transition layers and interfacial forces follow from a free-energy functional. In these models, surface tensions and diffuse thicknesses are not prescribed directly but are encoded implicitly by the bulk multiwell potential and the gradient-energy term through one-dimensional equilibrium profiles. While this link is classical in the binary Cahn–Hilliard setting, calibrating multiphase models is substantially more delicate because multiple pairwise surface tensions must be matched simultaneously and the relevant equilibrium paths are constrained by the Gibbs simplex. The practical problem is therefore: given a chosen bulk potential and a set of target pairwise surface tensions, determine gradient-energy coefficients that reproduce these targets in the full multiphase model.

Here we present a thermodynamically consistent calibration procedure for N-phase diffuse-interface free energies of Cahn–Hilliard type. The method determines a symmetric capillary matrix that matches prescribed pairwise surface tensions through the model’s equilibrium profiles. We further introduce a rescaling strategy that adjusts diffuse interface widths to mesh-resolvable values while preserving the calibrated surface tensions. The resulting calibrated free-energy closure can be incorporated directly into N-phase mixture simulations, and we demonstrate this by applying it to N-phase Navier--Stokes--Cahn--Hilliard flows.
\end{abstract}

\section{Introduction}\label{sec: Introduction}

Diffuse-interface (phase-field) models provide a versatile framework for interfacial dynamics in complex fluids, replacing sharp material interfaces by smooth transition layers of finite thickness \citep{anderson1998diffuse,gomez2018computational,saurel2018diffuse}. Their central appeal is that geometric effects (e.g. curvature-driven motion), topological changes, and contact-line dynamics can be represented without explicit interface tracking, while thermodynamic consistency can be built in through an energy-dissipation structure.
In diffuse-interface models, interfacial properties are not specified directly, but emerge from the interplay between a bulk potential and a gradient-energy term. In particular, pairwise surface tensions and diffuse thicknesses are encoded implicitly through the Helmholtz free energy and its associated equilibrium profiles. Prototypical realizations of this paradigm can be traced back to the Cahn–Hilliard theory of phase separation \citep{cahn1958free,elliott1989cahn,elliott1996cahn,novick2008cahn,brunk2026review}; coupling to fluid flow yields Navier-–Stokes–-Cahn-–Hilliard (NSCH) models, originating from model H \citep{hohenberg1977theory} and subsequently developed in many variants \citep{gurtinmodel,lowengrub1998quasi,ding2007diffuse,abels2012thermodynamically,eikelder2023unified}.

Despite the maturity of diffuse-interface modelling, a persistent practical challenge is calibration of interfacial physics. In many applications, surface tensions are measurable inputs (or are known from sharp-interface models \cite{Garcke2006,Dreyer}), whereas diffuse-interface parameters are not. For two phases\footnote{In this paper we use the terminology commonly used in the fluid-dynamics and diffuse-interface literature. Namely, the term \emph{phase} denotes a fluid constituent, such as water or oil. This differs from the (stricter) terminology in physics, where a phase typically refers to the thermodynamic state of a material, such as liquid, gas, or solid. We assume the phases to be incompressible meaning that each constituent has a fixed thermodynamic state and a constant specific density.}, calibration is straightforward. The one-dimensional heteroclinic equilibrium (together with equipartition of the free energy) yields a closed-form relation between the surface tension, the bulk potential, and the gradient coefficient, and the interface thickness can be tuned by rescaling the relative weight of bulk and gradient contributions while preserving surface tension. For multiphase systems, however, calibration becomes substantially more subtle. Pairwise surface tensions must be matched simultaneously, equilibrium profiles can involve higher-dimensional composition space, and the mapping from a capillary matrix to the set of pairwise surface tensions is generally nonlinear and constrained by the simplex structure.

Many multiphase diffuse-interface models have been developed over the last decades, spanning multiphase-field formulations and multiphase Cahn--Hilliard/NSCH variants \cite{steinbach1996phase,garcke1999,elliott1997diffusional,garcke2000singular,nestler2008phase,garcke2008modeling,kim2005phase,boyer2006study,boyer2014hierarchy,toth2016phase,dong2018multiphase,wu2017multiphase,ten2024thermodynamically,ten2025unified,ten2026compressible}. While these frameworks provide flexible model classes for $N$-phase systems, they leave substantial freedom in how the gradient-energy coefficients (capillary parameters) are chosen. In practice, these parameters are often set by heuristic rules or pairwise substitutions. The connection between such choices and a prescribed set of sharp-interface surface tensions is not transparent in general, particularly for multi-junction configurations and asymmetric surface-tension data. Most of the existing models have a clear connection between parameters and surface tensions for $N=2,3$. Beyond that some approaches do not have this clear connection or are even over-determined, \cite{Toth2016,Plapp2005,boyer2014hierarchy}.

We discuss two influential approaches in the present context. First, the multiphase-field energy of \cite{garcke1999} is designed so that binary interfacial profiles are confined to simplex faces. This feature simplifies the interpretation of pairwise interfacial contributions and facilitates prescribing pairwise parameters. Second, the hierarchical constructions of Boyer and co-authors (\cite{boyer2006study,boyer2014hierarchy}) enforce a reduction property when phases are absent and are typically parameterized by binary limits. In particular, the pairwise parameters are chosen so that, when the model is restricted to a binary subsystem (i.e. on a simplex face where all other phases are absent), the resulting diffuse-interface surface tension matches that of the underlying two-phase model. As such, both of the above two approaches do not consider the general case where equilibrium interfacial profiles leave the face of the Gibbs simplex.

This leaves open a practical inverse problem that is central in applications. Given a chosen bulk potential and target pairwise surface tensions, one seeks capillary parameters such that the diffuse-interface model reproduces these targets through its equilibrium profiles in the full multiphase setting. 
The goal of the present work is to provide a practical, thermodynamically consistent calibration procedure for $N$-phase diffuse-interface models that connects measurable interfacial data to the constitutive parameters of the free energy. We focus on the classical free-energy structure in which the gradient contribution is determined by a symmetric capillary matrix and the bulk contribution by a multiwell potential (i.e. a Cahn--Hilliard structure). We reconstruct the capillary matrix from the target surface tensions via one-dimensional equilibrium profiles, combining a linear initialization based on reference paths in composition space with a nonlinear refinement that iteratively recomputes equilibrium profiles and updates the matrix until the tensions are matched. In addition, we address the interface-width selection required for mesh-resolved computations by a rescaling that preserves the calibrated surface tensions while adjusting diffuse thicknesses to prescribed (numerical) targets. The outcome is a calibrated free-energy closure for $N$-phase mixture simulations. We note that selecting an admissible structure for the multiphase free energy and the associated mobility is a separate modelling problem; this is addressed in \cite{mixtureaware2026}, where we derive mixture-aware thermodynamic closures (free energy and mobility) within an $N$-phase NSCH mixture-theory setting. Furthermore, the design of numerical methods that preserve the structural properties of these models is addressed in \cite{structureNphase2026}; related fully discrete energy-stable schemes for two-phase NSCH models with arbitrary density ratios were developed in \cite{brunk2025simple}. Using these results as a modelling and computational baseline, we demonstrate the present calibration procedure both for multiphase Cahn--Hilliard dynamics and for $N$-phase Navier--Stokes--Cahn--Hilliard flows (using the framework presented in \cite{ten2025unified}).

The paper is organized as follows. In \cref{sec:thermodynamic-setting}, we introduce the thermodynamic setting, including the Gibbs simplex, coexistence points, grand potentials, and the diffuse-interface free-energy densities used throughout the paper. Next, in \cref{sec:surface_tension_calibration}, we derive the pairwise surface tensions associated with planar equilibrium profiles and develop the calibration procedure for determining the capillary matrix from prescribed target surface tensions. In \cref{sec:interface_width_calibration}, we define operational pairwise interface widths from the one-dimensional equilibrium profiles and introduce a global thickness scaling that makes the interfaces mesh-resolvable while preserving the calibrated surface tensions. In \cref{sec:num_exp}, we test the surface-tension calibration algorithm in
two-, three-, and four-phase examples and compare the straight-line
initialization with the reconstructed capillary matrices and resulting
interface widths. In \cref{sec:Applications}, we apply the methodology to a three-phase lens test and to an axisymmetric rising-bubble problem, thereby testing the prescribed surface tensions in both equilibrium and flow settings. Conclusions and outlook are given in \cref{sec: Conclusions}.

\section{Thermodynamic setting}
\label{sec:thermodynamic-setting}

This section provides the thermodynamic setting. In \cref{subsec:gibbs}, we introduce the volume fractions as constrained variables on the Gibbs simplex and identify the tangent space for admissible variations. In \cref{subsec:coexistence}, we define coexistence points using constrained bulk chemical potentials, tangent hyperplanes, and grand potentials. Finally, in \cref{subsec:free_energy}, we introduce the class of diffuse-interface free energies considered in this work, including the bulk free-energy densities and the capillary contribution.

\subsection{Gibbs simplex and admissible variables}\label{subsec:gibbs}
We consider the vector of volume fractions $\boldsymbol{\phi}=(\phi_1,\ldots,\phi_N)$ under the usual saturation constraint that $\sum_\mA \phi_\mA = 1$. Hence, the vector of volume fractions contains non-negative contributions satisfying the saturation constraint
\begin{align}\label{eq:sat_constraint}
  \sum_\mA \phi_\mA = 1.
\end{align}
Hence $\boldsymbol{\phi}$ lies in the Gibbs simplex $\mathcal{G}$ which is given by
\begin{align}
    \mathcal{G}:=\{\boldsymbol{\phi}\in\mathbb{R}^N: 0 \leq \phi_\mA \leq 1,\sum_\mA \phi_\mA = 1\}. \label{eq:defgibbs}
\end{align}
Correspondingly, for fixed $i$, the vector $\partial_{x_i}\boldsymbol{\phi}$ lies in the tangent space to the saturation hyperplane given by 
\begin{equation}\label{eq:deftanspace}
 T(\mathcal{G}):=\{\boldsymbol{c}\in\mathbb{R}^N: \sum_\mA c_\mA = 0\}.
\end{equation}
This reflects the fact that only $N-1$ volume fractions are independent.  An example Gibbs triangle with the associated coordinate axes and simplex boundaries is shown in Figure~\ref{fig:gibbs-triangle}. Furthermore, we refer to \cite{mixtureaware2026} for more details on the Gibbs simplex construction.

\begin{figure}
  \centering
  \includegraphics[width=0.7\linewidth]{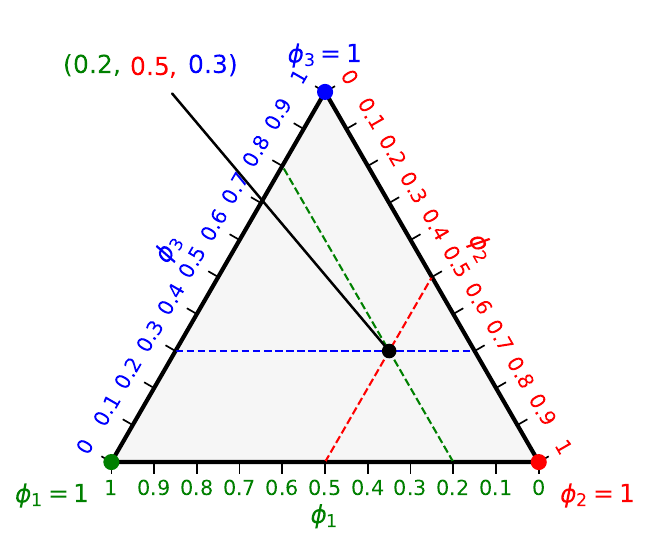}
  \caption{Gibbs triangle (ternary diagram) for $N=3$. Vertices correspond to pure states ($\phi_{\alpha}=1$), faces to binary mixtures (one phase identically zero), and the interior to fully ternary states. The simplex boundaries $\phi_{\alpha}=0$ are indicated and the point $(\phi_1,\phi_2,\phi_3)=(0.2,0.5,0.3)$ is shown as an example.}
  \label{fig:gibbs-triangle}
\end{figure}

\subsection{Coexistence points and grand potential}\label{subsec:coexistence}

This subsection introduces coexistence points for the bulk free-energy density $\Psi_0$ on the Gibbs simplex. To account for the saturation constraint, we use a constrained bulk potential. This provides a convenient way to define the tangent hyperplanes and grand potentials associated with coexisting bulk states.

Since the volume fractions satisfy the saturation constraint \eqref{eq:sat_constraint}, derivatives of $\Psi_0$ have to be interpreted on the saturation hyperplane. We account for this by introducing
the constrained bulk potential
\begin{align}
    \Psi_{0,\lambda}(\boldsymbol{\phi}) := \Psi_0(\boldsymbol{\phi}) +  \lambda \left(\sum_{\mA=1}^N \phi_\mA - 1
    \right),
\end{align}
where the scalar $\lambda$ is the Lagrange multiplier of the saturation constraint \eqref{eq:sat_constraint}. On the Gibbs simplex $\mathcal{G}$ we have
\begin{align}
     \Psi_{0,\lambda}(\boldsymbol{\phi}) = \Psi_0(\boldsymbol{\phi}),
\end{align}
but its unconstrained derivative is
\begin{align}
    \boldsymbol{\mu}_{0,\lambda}=\boldsymbol{\mu}_{0,\lambda}(\boldsymbol{\phi}):=\frac{\partial \Psi_{0,\lambda}}{\partial \boldsymbol{\phi}} =  \frac{\partial \Psi_{0}}{\partial \boldsymbol{\phi}} + \lambda \mathbf{1}.
\end{align}
Therefore the bulk derivative is only determined up to an additive multiple of $\boldsymbol{1}=(1,\ldots,1)^\top$. Different choices of $\lambda$ give the same action on admissible variations, since for all
$\delta\boldsymbol{\phi}\in T\mathcal G$ one has
\begin{align}
      \boldsymbol{\mu}_{0,\lambda} \cdot\delta\boldsymbol{\phi} = \frac{\partial \Psi_0}{\partial \boldsymbol{\phi}}\cdot\delta\boldsymbol{\phi}.
\end{align}

We now consider the existence and description of coexistence points. To this end we first introduce the tangent plane and the grand potential. Given a point $\mathbf b\in\mathcal G$, we define the tangent hyperplane to $\Psi_0$ at $\mathbf b$, restricted to the saturation hyperplane, by
\begin{align}\label{eq:tangent-hyperplane}
    T_{\mathbf b}(\boldsymbol{\phi}) := \Psi_0(\mathbf b) + \boldsymbol{\mu}_{0,\lambda}(\mathbf b)  \cdot \left( \boldsymbol{\phi}-\mathbf b \right). 
\end{align}
The value of $T_{\mathbf b}$ on the saturation hyperplane is independent of the particular choice of $\lambda$, since
$\boldsymbol{\phi}-\mathbf b\in T(\mathcal G)$ for
$\boldsymbol{\phi},\mathbf b\in\mathcal G$. The grand potential relative to $\mathbf b$ is the excess of the bulk free energy above this tangent hyperplane:
\begin{align}\label{eq:grand-potential}
    \Omega_{\mathbf b}(\boldsymbol{\phi})
    :=
    \Psi_0(\boldsymbol{\phi})
    -
    T_{\mathbf b}(\boldsymbol{\phi}).    
\end{align}
By construction $\Omega_{\mathbf b}(\mathbf b)=0$. Furthermore, if $\boldsymbol{\phi}\in\mathcal G$, then $\Omega_{\mathbf b}(\boldsymbol{\phi})$ measures the excess of $\Psi_0(\boldsymbol{\phi})$
above the tangent hyperplane at $\mathbf b$. Thus $\Omega_{\mathbf b}(\boldsymbol{\phi})=0$ means that $\boldsymbol{\phi}$ lies on the tangent hyperplane defined at $\mathbf{b}$.

\begin{definition}[Coexistence points]\label{def:coexistence-points}
 We call an $M$-tuple of vectors $\{\mathbf{b}^{(1)},\ldots,\mathbf{b}^{(M)}\}$ coexistence points of $M$ phases if they share a common tangent hyperplane which is supporting on $\mathcal G$. That is, the tangent hyperplanes
$T_{\mathbf b^{(q)}}$, $q=1,\ldots,M$, coincide on the saturation hyperplane, and the common tangent hyperplane satisfies
\begin{align}
\Psi_0(\boldsymbol{\phi}) \geq T_{\mathbf b^{(q)}}(\boldsymbol{\phi})
\qquad \text{for all } \boldsymbol{\phi}\in\mathcal G,
     \text{ and } q=1,\ldots,M,
\end{align}
with equality if and only if
$\boldsymbol{\phi}\in\{\mathbf{b}^{(1)},\ldots,\mathbf{b}^{(M)}\}$.
\end{definition}
Geometrically, this hyperplane lies below the free energy landscape and touches the free energy landscape only at the coexistence points. 

\begin{lemma}[Characterization of coexistence points]
The $M$-tuple of vectors $\{\mathbf{b}^{(1)},\ldots,\mathbf{b}^{(M)}\}$ are coexistence points of $M$ phases if and only if there exist scalars
$\lambda_1,\ldots,\lambda_M$ and a vector $\boldsymbol{\mu}^\star\in\mathbb R^N$ such that they have equal chemical potentials:
\begin{align}\label{eq:equal_mu}
    \boldsymbol{\mu}_{0,\lambda_q}(\mathbf{b}^{(q)})
        =\boldsymbol{\mu}_{0,\lambda_r}(\mathbf{b}^{(r)}) \qquad q,r=1,\ldots,M,
\end{align}
and satisfy the common-tangent condition:
\begin{align}\label{eq:common-tangent}
        \Omega_{\mathbf b^{(q)}}(\mathbf b^{(r)})=0,  \qquad q,r=1,\ldots,M.
\end{align}
\end{lemma}
\begin{proof}
If the points share a common tangent hyperplane, then the tangent hyperplanes at the points $\mathbf b^{(q)}, q=1,...,M$ have the same slope on the saturation hyperplane. This gives the equal-chemical-potential condition \eqref{eq:equal_mu}. Moreover, each
point $\mathbf b^{(r)}$ lies on the tangent hyperplane constructed at
$\mathbf b^{(q)}$, which is precisely the common-tangent construction \eqref{eq:common-tangent}.

Conversely, the equal-chemical-potential
condition \eqref{eq:equal_mu} gives a common slope, say $\boldsymbol{\mu}^\star$.  Hence each tangent hyperplane can be written as
\begin{align}
    T_{\mathbf b^{(r)}}(\boldsymbol{\phi})
    = \Psi_0(\mathbf b^{(r)}) +
    \boldsymbol{\mu}^\star\cdot(\boldsymbol{\phi}-\mathbf b^{(r)}) .
\end{align}
The common-tangent condition with indices $q$ and $r$ gives
\begin{align}
    0
    =
    \Omega_{\mathbf b^{(q)}}(\mathbf b^{(r)})
    =&~
    \Psi_0(\mathbf b^{(r)})
    -  T_{\mathbf b^{(q)}}(\mathbf b^{(r)})\nn\\
    =&~ \left(\Psi_0(\mathbf b^{(r)})  -  \boldsymbol{\mu}^\star\cdot\mathbf b^{(r)}\right) -\left(\Psi_0(\mathbf b^{(q)}) - \boldsymbol{\mu}^\star\cdot\mathbf b^{(q)}\right) .
\end{align}
All tangent hyperplanes thus have both the same slope and the same offset, and hence they coincide.
\end{proof}

In general, coexistence points are characterized by the common-tangent
construction rather than by the values of $\Psi_0$ alone. In certain situations, however, this condition simplifies. If several constrained minima of $\Psi_0$ have the same bulk free-energy value, then their tangent hyperplanes coincide and the minima are coexistence points. This gives the following sufficient condition.

\begin{lemma}[Equal-depth constrained minima]
\label{lem:equal-depth-constrained-minima}
Let $\mathbf b^{(1)},\ldots,\mathbf b^{(M)}\in\mathcal G$ be distinct minima of
$\Psi_0$ subject to the saturation constraint. If these constrained minima have
equal depth,
\begin{align}
    \Psi_0(\mathbf b^{(r)})=c,
    \qquad r=1,\ldots,M,
\end{align}
then $\mathbf b^{(1)},\ldots,\mathbf b^{(M)}$ are coexistence points.
\end{lemma}

\begin{proof}
Since $\mathbf b^{(r)}$ is a constrained minimum of $\Psi_0$ on the saturation
hyperplane, the constrained first-order optimality condition implies that there
exists a scalar $\lambda_r$ such that
\begin{align}
    \boldsymbol{\mu}_{0,\lambda_r}(\mathbf b^{(r)})
    =
    \frac{\partial\Psi_0}{\partial\boldsymbol{\phi}}(\mathbf b^{(r)})
    +
    \lambda_r\mathbf 1
    =
    \boldsymbol{0},
    \qquad r=1,\ldots,M .
\end{align}
Thus the equal-chemical-potential condition in
\cref{def:coexistence-points} holds with common value
$\boldsymbol{\mu}^\star=\boldsymbol{0}$. Moreover, all tangent hyperplanes at $\mathbf b^{(r)}$ coincide and are equal to the constant
hyperplane $c$:
\begin{align}
    T_{\mathbf b^{(r)}}(\boldsymbol{\phi})
    &=
    \Psi_0(\mathbf b^{(r)})
    +
    \boldsymbol{\mu}_{0,\lambda_r}(\mathbf b^{(r)})
    \cdot
    \left(
    \boldsymbol{\phi}-\mathbf b^{(r)}
    \right)= c.
\end{align}
Consequently, for all $r,s=1,\ldots,M$,
\begin{align}
    \Omega_{\mathbf b^{(r)}}(\mathbf b^{(s)})
    =
    \Psi_0(\mathbf b^{(s)})
    -
    T_{\mathbf b^{(r)}}(\mathbf b^{(s)})
    =
    c-c
    =
    0 .
\end{align}
Hence the common-tangent condition in \cref{def:coexistence-points} is also
satisfied. Therefore
$\mathbf b^{(1)},\ldots,\mathbf b^{(M)}$ are coexistence points.
\end{proof}
In accordance with the Gibbs phase rule we find that $F=N-M$, where $F$ is the number of thermodynamic degrees of freedom, $N$ the total number of phases and $M$ the number of coexisting phases, cf. \cite{landau_lifschitz_statphys_german,stinner2004diffuse}.
This means that if we consider coexistence of $M\leq N$ phases we have $N-M\geq 0$ degrees of freedom.
Hence, for the maximal case $M=N$, the coexistence points, if they exist, are unique, while in the other cases one expects hypersurfaces of valid coexistence points.
This conclusion can also be derived from a simple counting argument of the coexistence conditions.
For $M$ coexistence points we have $M(N-1)$ degrees of freedom, due to the simplex constraint.
Equal chemical potentials impose $(N-1)(M-1)$ constraints, and equal grand potential impose $M-1$ constraints.
Subtraction of constraints from degrees of freedom yields
\begin{align}
    M(N-1) - N(M-1) = N-M=F.
\end{align}

\begin{figure}
\captionsetup[subfigure]{justification=centering}
\begin{subfigure}{0.48\textwidth}
\centering
\includegraphics[scale=0.60]{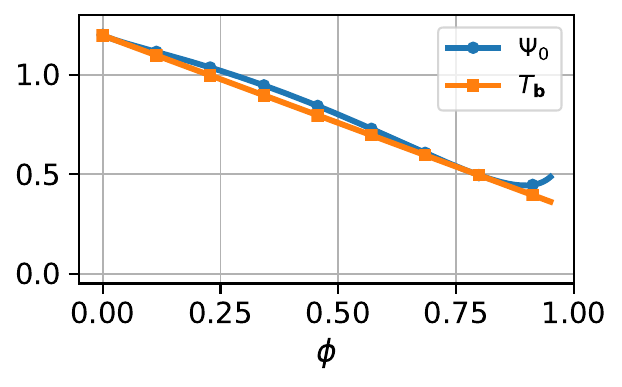}
\caption{Bulk free-energy density and common tangent hyperplane}

\end{subfigure}
\begin{subfigure}{0.48\textwidth}
\centering
\includegraphics[scale=0.60]{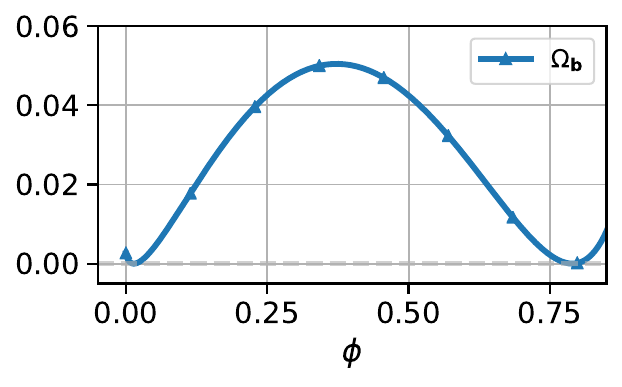}
\caption{Grand potential}

\end{subfigure}
\caption{Geometric interpretation of coexistence points. (a) A bulk free-energy density with a common tangent hyperplane, and (b) the associated grand potential.}
\label{fig:coexistence}
\end{figure}

\subsection{Free-energy density}\label{subsec:free_energy}
In the following we consider a free energy density of the following form
\begin{subequations}
\label{eq:Psi_final}
    \begin{align}
    \Psi(\boldsymbol{\phi},\nabla\boldsymbol{\phi})
    =&~ \Psi_0(\boldsymbol{\phi})+\frac12\sum_{\alpha,\beta=1}^{N}\kappa_{\alpha\beta}\nabla\phi_\alpha\cdot\nabla\phi_\beta.
\end{align}
\end{subequations}    
Here $\Psi_0$ denotes the bulk free-energy density,  whereas the gradient term represents an excess free energy in the interfacial region. Here $\kappa_{\mA\mB}=\kappa_{\mB\mA}$ defines the symmetric (volume-based) capillary matrix which is positive semi-definite on the tangent plane and includes both self- and cross-interaction terms among phases. The
matrix $\boldsymbol{\kappa}=(\kappa_{\mA\mB})_{\mA,\mB=1}^N$ is assumed to be
symmetric, $\kappa_{\mA\mB}=\kappa_{\mB\mA}$, and contains both self- and
cross-gradient contributions. This form is compatible with the free-energy
closure used in the mixture-aware formulation of \cite{mixtureaware2026} and
includes capillary contributions commonly used in multiphase phase-field models, see for example \cite{garcke1999,boyer2014hierarchy}.

The bulk part of the mixture-aware closure is given by
\begin{align}   \label{eq:FH}
\Psi_0(\boldsymbol{\phi}) = \sum_\mA W_\mA \phi_\mA\log(\phi_\mA) - \sum_{\mA\mB} \chi_{\mA\mB}\phi_\mA\phi_\mB,
\end{align} 
where $W_\mA>0$ are the weights of the entropic contributions, and $\chi_{\mA\mB}$ are pairwise interaction parameters. For the computations in this paper we use the polynomially regularized version of \eqref{eq:FH} proposed in \cite{mixtureaware2026}. The free energy \eqref{eq:FH} may be written in the Flory--Huggins form:
\begin{equation}
\Psi_0(\boldsymbol{\phi}) = \sum_\mA \frac{\phi_\mA}{N_\mA}\log(\phi_\mA) - \sum_{\mA\mB} \chi_{\mA\mB}\phi_\mA\phi_\mB,    
\end{equation} 
for constants polymerization degrees $N_\mA$, $\mA=1,...,N$. Another quite well-known example is the three-phase polynomial free energy of \cite{boyer2006study}. This model may be interpreted as a
particular multiphase analogue of the binary Ginzburg--Landau potential:
\begin{subequations}\label{eq:Boyerpot}
\begin{align}
\Psi_0&(\phi_1,\phi_2,\phi_3)
=\frac{6}{\varepsilon}\left(\Sigma_{1}\phi_1^2(1-\phi_1)^2+\Sigma_{2}\phi_2^2(1-\phi_2)^2 + \Sigma_{3}\phi_3^2(1-\phi_3)^2\right)\nn\\
&~\qquad\qquad\qquad+\frac{12}{\varepsilon}\Lambda \phi_1^2\phi_2^2\phi_3^2,\\
\Sigma_{1}&=\gamma_{12}+\gamma_{13}-\gamma_{23},\qquad
\Sigma_{2}=\gamma_{12}+\gamma_{23}-\gamma_{13},\qquad
\Sigma_{3}=\gamma_{13}+\gamma_{23}-\gamma_{12}. 
\end{align}
\end{subequations}  
Here $\gamma_{\mA\mB}$ denote prescribed pairwise parameters and $\Lambda\ge0$ penalizes fully ternary states. For this free energy, the relation between the parameters $\gamma_{\mA\mB}$ and the actual interfacial energies depends on the associated capillary contribution. We discuss this in \cref{thm:boyer-face-confinement}.

In the incompressible literature it is common to work with capillary matrices that exclude self-interactions, see e.g. \cite{boyer2006study}. Using the definition of the Gibbs simplex \eqref{eq:defgibbs} and its tangent space \eqref{eq:deftanspace} we reformulate the gradient contribution.
\begin{proposition}[Gradient contribution of the free energy]\label[proposition]{prop:rewriting gradient term}
    The gradient contribution of the free energy \eqref{eq:Psi_final} may be written as:
    \begin{subequations}
      \begin{align}
     \tfrac{1}{2}\sum_{\mA,\mB} \kappa_{\mA\mB} \nabla \phi_\mA \cdot \nabla \phi_\mB =&~ -\tfrac{1}{2}\sum_{\substack{\mA,\mB\\\mA < \mB}} \sigma_{\mA\mB} \nabla \phi_\mA \cdot \nabla \phi_\mB,\\
     \sigma_{\mA\mB} =&~ \kappa_{\mA\mA}+\kappa_{\mB\mB}-2\kappa_{\mA\mB} \geq 0.
     \end{align}
     \end{subequations}
\end{proposition}
We emphasize that in general neither $\kappa_{\mA\mB}$ nor $\sigma_{\mA\mB}$ represents a surface tension quantity. 
The (symmetric) capillary matrix $\sigma_{\mA\mB}=\sigma_{\mB\mA}$ captures only the interfacial contributions between different phases, excluding self-interactions: $\sigma_{\mA\mA} = 0$. In addition, \cref{prop:rewriting gradient term} reveals that the choice $\kappa_{\mA\mB}=\tfrac{1}{2}(\kappa_{\mA\mA}+\kappa_{\mB\mB})$ for all $\mA\neq\mB$ causes the gradient contribution in \eqref{eq:Psi_final} to vanish on the tangent plane.

\begin{proof}
    This follows from the sequence of identities:
    \begin{align}
        \sum_{\mA,\mB} \kappa_{\mA\mB} \nabla \phi_\mA \cdot \nabla \phi_\mB =&~ \sum_{\substack{\mA,\mB\\\mA < \mB}} \kappa_{\mA\mB} \nabla \phi_\mA \cdot \nabla \phi_\mB + \sum_{\substack{\mA,\mB\\\mB < \mA}} \kappa_{\mA\mB} \nabla \phi_\mA \cdot \nabla \phi_\mB + \sum_{\mA} \kappa_{\mA\mA} \nabla \phi_\mA \cdot \nabla \phi_\mA \nn\\
        =&~ \sum_{\substack{\mA,\mB\\\mA < \mB}} (\kappa_{\mA\mB}+\kappa_{\mB\mA}) \nabla \phi_\mA \cdot \nabla \phi_\mB - \sum_{\mA} \sum_{\substack{\mB\\\mB \neq \mA}} \kappa_{\mA\mA} \nabla \phi_\mA \cdot \nabla \phi_\mB \nn\\        
        =&~ \sum_{\substack{\mA,\mB\\\mA < \mB}}(\kappa_{\mA\mB}+\kappa_{\mB\mA}) \nabla \phi_\mA \cdot \nabla \phi_\mB - \sum_{\substack{\mA,\mB\\\mA < \mB}}  \kappa_{\mA\mA} \nabla \phi_\mA \cdot \nabla \phi_\mB \nn\\
        &~-  \sum_{\substack{\mA,\mB\\\mB < \mA}}  \kappa_{\mA\mA} \nabla \phi_\mA \cdot \nabla \phi_\mB \nn\\
        =&~\sum_{\substack{\mA,\mB\\\mA < \mB}}  ( \kappa_{\mA\mB}+\kappa_{\mB\mA} - \kappa_{\mA\mA}-\kappa_{\mB\mB}) \nabla \phi_\mA \cdot \nabla \phi_\mB,
    \end{align}
    where the first follows from splitting the double sum, the second from the tangent space \eqref{eq:deftanspace} and rearranging the summation, the third from again rearranging the terms, and the final identity from combining the terms.
\end{proof}

\section{Surface tension calibration}
\label{sec:surface_tension_calibration}

In this section we describe how the free energy \eqref{eq:Psi_final} determines pairwise surface tensions and how the capillary matrix $\boldsymbol{\kappa}$ can be calibrated from prescribed target surface tension values. In \cref{subsec:surface_tension}, we provide the definition of surface tension for a planar equilibrium interface and derive the corresponding path-based representation. Next, in \cref{subsec:mass_constraints} we describe the influence of mass constraints. In \cref{subsec:straight_line_approximation}, we introduce a straight-line approximation that provides an initial estimate for the surface tensions and effective pairwise capillary coefficients. In \cref{subsec:kappa_reconstruction}, we formulate an iterative reconstruction procedure for determining the capillary matrix from prescribed pairwise surface tensions. Finally, in \cref{subsec:face_confinement} we discuss the confinement of the minimizing path to the corresponding Gibbs simplex face for different free energies.

\subsection{Surface tension}\label{subsec:surface_tension}
Let $\mA,\mB\in\{1,\ldots,N\}$ with $\mA\neq\mB$, and let $\mathbf b^{(\mA)},\mathbf b^{(\mB)}\in\mathcal G$ denote the corresponding coexistence points. We consider a planar interface with unit normal $\boldsymbol{\nu}$ and introduce the normal coordinate $z=\mathbf x\cdot\boldsymbol{\nu}$. To describe the transition from $\mathbf b^{(\mA)}$ to $\mathbf b^{(\mB)}$, let $s:\mathbb R\to[0,1]$ be a strictly monotone parameter depending on the normal coordinate $z$, with 
\begin{align}
    \lim_{z\to-\infty}s(z)=0, \qquad \lim_{z\to+\infty}s(z)=1.
\end{align}
A planar one-dimensional interfacial profile is then a curve  $\mathbf{p}^{(\mA\mB)}=\mathbf{p}^{(\mA\mB)}(s)$ connecting two coexistence points:
\begin{align}\label{eq: boundary conditions phi}
  \mathbf{p}^{(\mA\mB)}: [0,1] \to \mathcal{G},\qquad \mathbf{p}^{(\mA\mB)}(0)=\mathbf{b}^{(\mA)},\quad \mathbf{p}^{(\mA\mB)}(1)=\mathbf{b}^{(\mB)}.
\end{align}
The corresponding physical profile, written as a function of the normal coordinate, is denoted by
$\hat{\mathbf p}^{(\mA\mB)}:\mathbb R\to\mathcal G$ and defined by
\begin{align}
    \hat{\mathbf p}^{(\mA\mB)}(z) := \mathbf p^{(\mA\mB)}(s(z)).
\end{align}
Hence, the phase-field profile depends only on the normal coordinate through
\begin{align}
    \boldsymbol{\phi}(\mathbf{x})=\mathbf{p}^{(\mA\mB)}(s)=\hat{\mathbf{p}}^{(\mA\mB)}(z).
\end{align}
With this notation, derivatives with respect to the normal coordinate and the path parameter are related by 
\begin{align}
    \frac{{\rm d}\hat{\mathbf p}^{(\mA\mB)}}{{\rm d}z}(z) = \mathbf p^{(\mA\mB)\prime}(s(z))\frac{{\rm d}s}{{\rm d}z}, \qquad \nabla\boldsymbol{\phi} = \frac{{\rm d}\hat{\mathbf p}^{(\mA\mB)}}{{\rm d}z}\otimes\boldsymbol{\nu}  = \mathbf p^{(\mA\mB)\prime}(s)\frac{{\rm d}s}{{\rm d}z}\otimes\boldsymbol{\nu}.
\end{align}

The free energy per unit interfacial area associated with the physical profile
$\hat{\mathbf p}^{(\mA\mB)}$ is
\begin{align}\label{eq: free energy interfacial area z}
  \gamma_{\mA\mB}[\hat{\mathbf p}^{(\mA\mB)}]
  =
  \int_{-\infty}^{\infty}
  \left(
      \Omega_{\mathbf{b}^{(\mA)}}(\hat{\mathbf p}^{(\mA\mB)}(z))
      + \frac{1}{2}\sum_{\gamma,\delta}
      \kappa_{\gamma\delta}
      \frac{{\rm d}\hat p_\gamma^{(\mA\mB)}}{{\rm d}z}(z)
      \frac{{\rm d}\hat p_\delta^{(\mA\mB)}}{{\rm d}z}(z)
  \right)
  {\rm d}z.
\end{align}
The surface tension between phases $\mA$ and $\mB$ is the minimal free energy per unit interfacial area among all profiles connecting the two coexistence points:
\begin{align}\label{eq: ST FH}
  \gamma_{\mA\mB}
  =
  \inf_{\hat{\mathbf p}^{(\mA\mB)}}\gamma_{\mA\mB}[\hat{\mathbf p}^{(\mA\mB)}].
\end{align}

The same interfacial energy can be written in terms of the path parameter $s$ as:
\begin{align}\label{eq: free energy interfacial area s}
  \gamma_{\mA\mB}[\mathbf p^{(\mA\mB)},z]
  =
  \int_0^1
  \left(
      \Omega_{\mathbf{b}^{(\mA)}}(\mathbf p^{(\mA\mB)}(s))
      \frac{{\rm d}z}{{\rm d}s}
      +
      \frac{1}{2}
      \sum_{\gamma,\delta}
    \kappa_{\gamma\delta}
    p_\gamma^{(\mA\mB)\prime}(s)
    p_\delta^{(\mA\mB)\prime}(s)
      \frac{{\rm d}s}{{\rm d}z}
  \right)
  {\rm d}s.
\end{align}
For a fixed path $\mathbf p^{(\mA\mB)}$, minimization with respect to the parameterization $z=z(s)$ gives
\begin{align}\label{eq: dzds_surface_tension}
    \frac{{\rm d}z}{{\rm d}s}
    =
    \sqrt{
    \frac{\sum_{\gamma,\delta}
    \kappa_{\gamma\delta}
    p_\gamma^{(\mA\mB)\prime}(s)
    p_\delta^{(\mA\mB)\prime}(s)}
    {2\Omega_{\mathbf{b}^{(\mA)}}(\mathbf p^{(\mA\mB)}(s))}
    }.
\end{align}
Substitution into \eqref{eq: free energy interfacial area s} yields the parameterization-invariant path functional
\begin{align}\label{eq: ST FH1}
  \gamma_{\mA\mB}
  =
  \inf_{\mathbf p^{(\mA\mB)}}
  \int_0^1
      \sqrt{
      2\Omega_{\mathbf{b}^{(\mA)}}(\mathbf p^{(\mA\mB)}(s))
      \sum_{\gamma,\delta}
      \kappa_{\gamma\delta}
      p_\gamma^{(\mA\mB)\prime}(s)
      p_\delta^{(\mA\mB)\prime}(s)
      }
   {\rm d}s.
\end{align}
Optimality conditions yield the well-known Euler-Lagrange equations associated with
\eqref{eq: free energy interfacial area z}:
\begin{align}\label{eq: EL}
  \dfrac{\partial\Omega_{\mathbf{b}^{(\mA)}}}{\partial \hat p_\gamma^{(\mA\mB)}}(\hat{\mathbf p}^{(\mA\mB)}(z)) - \sum_{\delta}\kappa_{\gamma\delta} \frac{{\rm d}^2\hat p_\delta^{(\mA\mB)}}{{\rm d}z^2}(z) = 0, \qquad \gamma=1,\ldots,N.
\end{align}
Taking the inner product of \eqref{eq: EL} with
${\rm d}\hat{\mathbf p}^{(\mA\mB)}/{\rm d}z$ and integrating with respect to $z$ gives
\begin{align}\label{eq: EL2}
    \Omega_{\mathbf{b}^{(\mA)}}(\hat{\mathbf p}^{(\mA\mB)}(z))
    -
    \frac{1}{2}\sum_{\gamma,\delta}
    \kappa_{\gamma\delta}
    \frac{{\rm d}\hat p_\gamma^{(\mA\mB)}}{{\rm d}z}(z)
    \frac{{\rm d}\hat p_\delta^{(\mA\mB)}}{{\rm d}z}(z)
    =
    C.
\end{align}
Using the far-field conditions at $z=\pm\infty$, together with
$\Omega_{\mathbf{b}^{(\mA)}}(\mathbf{b}^{(\mA)})=\Omega_{\mathbf{b}^{(\mA)}}(\mathbf{b}^{(\mB)})=0$ and
${\rm d}\hat{\mathbf p}^{(\mA\mB)}/{\rm d}z\to\mathbf 0$ as $z\to\pm\infty$, gives $C=0$.
Hence, the equilibrium profile satisfies the equipartition identity
\begin{align}\label{eq: EL3}
   \Omega_{\mathbf{b}^{(\mA)}}(\hat{\mathbf p}^{(\mA\mB)}(z))
    =
    \frac{1}{2}\sum_{\gamma,\delta}
    \kappa_{\gamma\delta}
    \frac{{\rm d}\hat p_\gamma^{(\mA\mB)}}{{\rm d}z}(z)
    \frac{{\rm d}\hat p_\delta^{(\mA\mB)}}{{\rm d}z}(z).
\end{align}
Substituting \eqref{eq: EL3} into the interfacial energy gives the equivalent
physical-coordinate forms
\begin{align}\label{eq: ST FH2 z}
  \gamma_{\mA\mB}
  =
  \inf_{\hat{\mathbf p}^{(\mA\mB)}} &~
  2\int_{-\infty}^{\infty}
      \Omega_{\mathbf{b}^{(\mA)}}(\hat{\mathbf p}^{(\mA\mB)}(z))
   {\rm d}z \nn\\
  =
  \inf_{\hat{\mathbf p}^{(\mA\mB)}} &~
  \int_{-\infty}^{\infty}
      \sum_{\gamma,\delta}
      \kappa_{\gamma\delta}
      \frac{{\rm d}\hat p_\gamma^{(\mA\mB)}}{{\rm d}z}(z)
      \frac{{\rm d}\hat p_\delta^{(\mA\mB)}}{{\rm d}z}(z)
   {\rm d}z .   
\end{align}
Equivalently, using \eqref{eq: dzds_surface_tension}, the surface tension can be written as the parameterization-invariant path functional
\begin{align}\label{eq: ST FH2}
  \gamma_{\mA\mB}
  =
  \inf_{\mathbf{p}^{(\mA\mB)}} 
  \int_{0}^{1} 
      \sqrt{
      2\Omega_{\mathbf{b}^{(\mA)}}(\mathbf{p}^{(\mA\mB)}(s))
      \sum_{\gamma,\delta}
      \kappa_{\gamma\delta}
      p_\gamma^{(\mA\mB)\prime}(s)
      p_\delta^{(\mA\mB)\prime}(s)}
   {\rm d}s.
\end{align}

\subsection{Effects of mass constraints}\label{subsec:mass_constraints}

In many applications the stationary solutions correspond to minimization of the free energy \eqref{eq:Psi_final}, but under additional constraints, most commonly a mass conservation constraint. This leads to
\begin{align}
 &\min_{\boldsymbol{\phi}} \int_\Omega  \Psi(\boldsymbol{\phi},\nabla \boldsymbol{\phi})   - \sum_\mA \hat{c}_\mA (\phi_\alpha-m_\alpha) + \lambda(\sum_\mA \phi_\mA - 1) \nn\\
 = &\min_{ \la \boldsymbol{\phi},\mathbf{e}_\mA \ra=m_\mA} \int_\Omega  \Psi(\boldsymbol{\phi},\nabla \boldsymbol{\phi})   + \lambda(\sum_\mA \phi_\mA - 1) \label{eq:1Dmin}
\end{align}

However, it is not at all clear whether the equilibrium profile related to the surface tension is also a minimizer of the constrained energy, i.e. we have to clarify the effects of $\hat{c}_\mA$ and $\lambda$ on equilibrium solutions. 

\begin{theorem}[Equipartition of the energy]\label{thm:equi_form}
Let $\boldsymbol{\phi}$ be a one-dimensional equilibrium profile connecting
$\mathbf{b}^{(\mA)}$ and $\mathbf{b}^{(\mB)}$, with
$\partial_z\boldsymbol{\phi}\to0$ at both endpoints. Then we have the equipartition:
\begin{align}
    \widetilde{\Psi}_0(\boldsymbol{\phi})
    -
    \widetilde{\Psi}_0(\mathbf b^{(\mA)})
    =
    \sum_{\mC,\mD}
    \frac{\kappa_{\mC\mD}}{2}
    \partial_z\phi_\mC \partial_z\phi_\mD,
\end{align}
where
\begin{align}
    \widetilde{\Psi}_0(\boldsymbol{\phi})
    =
    \Psi_0(\boldsymbol{\phi})
    -
    \sum_\mC \hat c_\mC\phi_\mC,
\end{align}
and whose endpoint values satisfy
$\widetilde{\Psi}_0(\mathbf b^{(\mA)})
 =
 \widetilde{\Psi}_0(\mathbf b^{(\mB)})$.
\end{theorem}

The result requires that the tilted potential $\tilde\Psi_0$ is the same for $\mathbf b^{(\mA)}$ and $\mathbf b^{(\mB)}$. This enforces algebraic conditions on $\mathbf b^{(\mA)}, \mathbf b^{(\mB)}$, which can directly be connected to the coexistence points of $\tilde\Psi_0$ and the grand potential.

\begin{corollary}[Equal-depth wells]
\label{cor:equal-depth-equipartition}
Assume the setting of \cref{thm:equi_form}. If the two endpoint
values of the bulk potential are equal, i.e. $\Psi_0(\mathbf{b}^{(\mA)})=\Psi_0(\mathbf{b}^{(\mB)})$, then the mass multiplier satisfies
\begin{align}
    \hat{\mathbf c}\cdot(\mathbf{b}^{(\mB)}-\mathbf{b}^{(\mA)})=0.
\end{align}
Moreover, if
\begin{align}
    \hat{\mathbf c}\cdot(\boldsymbol{\phi}(z)-\mathbf{b}^{(\mA)})=0
    \qquad \text{for all }z,
\end{align}
then
\begin{align}
    \Psi_0(\boldsymbol{\phi})
    -
    \Psi_0(\mathbf{b}^{(\mA)})
    =
    \sum_{\mC,\mD}
    \frac{\kappa_{\mC\mD}}{2}
    \partial_z\phi_\mC\,\partial_z\phi_\mD .
\end{align}
If, in addition, the common endpoint value is normalized to zero, then
\begin{align}
    \Psi_0(\boldsymbol{\phi})
    =
    \sum_{\mC,\mD}
    \frac{\kappa_{\mC\mD}}{2}
    \partial_z\phi_\mC\,\partial_z\phi_\mD .
\end{align}
\end{corollary}

\begin{corollary}[Equipartition for unequal well depths]
\label{cor:unequal-depth-tilt}
Assume the setting of \cref{thm:equi_form}.  If the two endpoint
values of the bulk potential are unequal, i.e. 
$\Psi_0(\mathbf{b}^{(\mA)})\neq \Psi_0(\mathbf{b}^{(\mB)})$ then 
\begin{align}
    \Psi_0(\boldsymbol{\phi}) - \Psi_0(\mathbf{b}^{(\mA)})
    - \hat{\mathbf c}\cdot \left(\boldsymbol{\phi}-\mathbf{b}^{(\mA)}\right)  = \sum_{\mC,\mD}
    \frac{\kappa_{\mC\mD}}{2} \partial_z\phi_\mC\,\partial_z\phi_\mD,
\end{align}
where one compatible choice of $\hat{\mathbf c}$ is
\begin{align}
    \hat{\mathbf c} = \frac{\Psi_0(\mathbf{b}^{(\mB)})-\Psi_0(\mathbf{b}^{(\mA)})}{\|\mathbf{b}^{(\mB)}-\mathbf{b}^{(\mA)}\|^2}(\mathbf{b}^{(\mB)}-\mathbf{b}^{(\mA)}).
\end{align}
\end{corollary}

Hence, we observe that, in general, the mass constraint leads to different equilibrium profiles. Note that surface tension does not depend on such constraints; therefore, it may happen that, for a given mass constraint, the equilibrium profile at which the surface tension value is attained is not realizable. The proofs of the above theorem and corollaries are given in Appendix \ref{appendix sec: minizers}.

\subsection{Straight-line approximation}
\label{subsec:straight_line_approximation}
The exact surface tension \eqref{eq: ST FH2} is obtained by minimizing over all admissible paths in the Gibbs simplex. In the calibration problem, however, we also need an explicit initial relation between the target surface tensions and the entries of $\boldsymbol{\kappa}$. For this purpose we first approximate the minimizing path by the straight segment connecting the coexistence points. The resulting expression is generally not exact, since the straight segment need not be an equilibrium path, but it provides a useful initialization for the reconstruction procedure below.

For each pair $\mA\neq\mB$, we define the straight segment connecting the coexistence points. Thus, for each pair $\mA\neq\mB$, we set
\begin{align}
    \tilde{\mathbf{p}}^{(\mA\mB)}=\mathbf{b}^{(\mA)}+s(\mathbf{b}^{(\mB)}-\mathbf{b}^{(\mA)}), \qquad s\in[0,1].
\end{align}
The corresponding approximate surface tension is
\begin{align}
    \tilde{\gamma}_{\mA\mB}:=\int_{0}^{1}       \sqrt{2\Omega_{\mathbf{b}^{(\mA)}}(\tilde{\mathbf{p}}^{(\mA\mB)}(s))\sum_{\gamma,\delta} \kappa_{\gamma\delta}\tilde{p}_\gamma^{(\mA\mB)\prime}(s) \tilde{p}_\delta^{(\mA\mB)\prime}(s)}
   {\rm d}s
\end{align}
Since $\tilde{\mathbf p}^{(\mA\mB)\prime}(s)=\mathbf b^{(\mB)}-\mathbf b^{(\mA)}$, this reduces to
\begin{align}\label{eq:straight_line_surface_tension}
  \tilde{\gamma}_{\mA\mB}= &~ \int_{0}^{1}       \sqrt{2\Omega_{\mathbf{b}^{(\mA)}}(\tilde{\mathbf{p}}^{(\mA\mB)}(s))\sum_{\gamma,\delta} \kappa_{\gamma\delta}(\mathbf{b}^{(\mA)}-\mathbf{b}^{(\mB)})_\gamma(\mathbf{b}^{(\mA)}-\mathbf{b}^{(\mB)})_\delta}   {\rm d}s\nn\\
   = &~ F_{\mA\mB}\sqrt{\tilde\sigma_{\mA\mB}},
\end{align}
where we defined 
\begin{subequations}
\begin{align}
    F_{\mA\mB} =&~ \int_{0}^{1}       \sqrt{2\Omega_{\mathbf{b}^{(\mA)}}(\tilde{\mathbf{p}}^{(\mA\mB)}(s)) }
   {\rm d}s, \\
   \tilde\sigma_{\mA\mB} =&~ \sum_{\gamma,\delta} \kappa_{\gamma\delta}(\mathbf{b}^{(\mA)}-\mathbf{b}^{(\mB)})_\gamma(\mathbf{b}^{(\mA)}-\mathbf{b}^{(\mB)})_\delta  .
\end{align}
\end{subequations}

\begin{remark}[Alternative forms of the approximate surface tension]
The straight-line approximation \eqref{eq:straight_line_surface_tension} separates the contribution of the bulk potential, through $F_{\mA\mB}$, from the contribution of the capillary metric, through $\tilde\sigma_{\mA\mB}$. If, in addition, one were to impose the equipartition relation within this restricted straight-line ansatz, then one would obtain the formal condition
\begin{align}
    F_{\mA\mB}=\sqrt{\tilde\sigma_{\mA\mB}},
\end{align}
and hence the equivalent identities
\begin{align}
\tilde{\gamma}_{\mA\mB}= ~ F_{\mA\mB}^2, \qquad \tilde{\gamma}_{\mA\mB}= ~ \tilde\sigma_{\mA\mB}. 
\end{align}
However, this equipartition condition is not generally satisfied by the prescribed straight-line path, since such a path is rarely the true minimizing interfacial profile. Consequently, the three quantities $F_{\mA\mB}\sqrt{\tilde\sigma_{\mA\mB}}, F_{\mA\mB}^2, \tilde\sigma_{\mA\mB}$ should be regarded as distinct in general; they coincide only when the straight-line ansatz is compatible with the equilibrium equipartition structure.  
\end{remark}

If the coexistence points are the unit vectors, i.e.
$\mathbf{b}^{(\mA)}=\mathbf{e}_\mA$ and
$\mathbf{b}^{(\mB)}=\mathbf{e}_\mB$, we recover the simple identity:
\begin{align}
    \sum_{\gamma,\delta} \kappa_{\gamma\delta}(\mathbf{e}_\mA-\mathbf{e}_\mB)_\gamma(\mathbf{e}_\mA-\mathbf{e}_\mB)_\delta = \kappa_{\mA\mA}-2\kappa_{\mA\mB}+\kappa_{\mB\mB}=\sigma_{\mA\mB}. \label{eq:defsigmaline}
\end{align}
Thus, for a given bulk potential and capillary matrix $\boldsymbol{\kappa}$, the straight-line approximation provides explicit estimates of the pairwise surface tensions. In the calibration problem we require the inverse direction: given a bulk potential and prescribed target values $\gamma_{\mA\mB}$, determine a capillary matrix $\boldsymbol{\kappa}$ whose equilibrium interfacial profiles reproduce these surface tensions.

\subsection{Reconstruction of the capillary matrix}
\label{subsec:kappa_reconstruction}

We now formulate the reconstruction problem for $\boldsymbol{\kappa}$. The coexistence points $\{\mathbf b^{(\mA)}\}$ are determined by the bulk potential $\Psi_0$ and are assumed to be available in what follows. In special cases, such as when the pure phases correspond to the unit vectors, these points are known explicitly. In general, however, their computation is a separate nonlinear problem and may be nontrivial.

A full symmetric $N\times N$ capillary matrix contains $N(N+1)/2$ entries, whereas the prescribed data consist of $N(N-1)/2$ pairwise surface tensions. However, due to the saturation constraint, only the restriction of $\boldsymbol{\kappa}$ to the tangent space \eqref{eq:deftanspace} contributes to the gradient energy. Equivalently, the phases of $\boldsymbol{\kappa}$ in the direction of $\bm 1$ are gauge degrees of freedom and do not affect the surface tensions. The effective number of capillary parameters is therefore $N(N-1)/2$, matching the number of pairwise target surface tensions.

If an equilibrium profile is fixed, then the equipartition form of the interfacial energy implies
a linear dependence on the capillary matrix. Indeed, for an exact equilibrium
profile one has
\begin{align}
    \gamma_{\mA\mB}
    =
    \int_{-\infty}^{\infty}
    \sum_{\gamma,\delta}
    \kappa_{\gamma\delta}
    \frac{{\rm d}\hat p_\gamma^{(\mA\mB)}}{{\rm d}z}
    \frac{{\rm d}\hat p_\delta^{(\mA\mB)}}{{\rm d}z}
    {\rm d}z .
\end{align}
Thus, if the profile is held fixed, the right-hand side is linear in the entries
of $\boldsymbol{\kappa}$:
\begin{subequations}\label{eq:fixed_profile_linear_system}
    \begin{align}
    \gamma_{\mA\mB}
    \approx&~
    \sum_{\gamma,\delta}
    A_{\mA\mB,\gamma\delta}\kappa_{\gamma\delta},
    \qquad \mA<\mB,\\
    A_{\mA\mB,\gamma\delta}
    :=&~
    \int_{-\infty}^{\infty}
    \frac{{\rm d}\hat p_\gamma^{(\mA\mB)}}{{\rm d}z}
    \frac{{\rm d}\hat p_\delta^{(\mA\mB)}}{{\rm d}z}
    {\rm d}z.
\end{align}
\end{subequations}
In the algorithm we use this idea in the path-parametrized form. For the current
collection of paths $\mathbf p^{(\mA\mB)}$, we define the path-metric coefficients
\begin{align}
    g^{\mA\mB}_{\gamma,\delta}
    :=
    \int_{0}^{1}
    p_\gamma^{(\mA\mB)\prime}(s)
    p_\delta^{(\mA\mB)\prime}(s)\,{\rm d}s .
\end{align}
For these fixed paths, the effective pairwise capillary quantities satisfy the
linear relation
\begin{align}
    \sigma_{\mA\mB}
    \approx
    \sum_{\gamma,\delta}
    g^{\mA\mB}_{\gamma,\delta}\kappa_{\gamma\delta}.
\end{align}
Collecting these equations for all pairs $\mA<\mB$ gives a linear system for the
entries of $\boldsymbol{\kappa}$. To assemble this system, we enumerate both the pairwise surface-tension data and the entries of the capillary matrix. The row index $\mu$ labels the unordered phase pair $(\mA,\mB)$ with $\mA<\mB$. We use the lexicographic ordering
\begin{align}
    (1,2),(1,3),\ldots,(1,N),(2,3),\ldots,(2,N),\ldots,(N-1,N).
\end{align}
Thus the pair $(\mA,\mB)$ is assigned the index
\begin{align}
    \eta(\mA,\mB)
    =
    \sum_{r=1}^{\mA-1}(N-r)+(\mB-\mA),
    \qquad \mA<\mB.
\end{align}
The first term counts the number of pairs whose first index is smaller than $\mA$, while the second term gives the position of $\mB$ among the pairs starting with $\mA$. Hence $\mu$ ranges from $1$ to $N(N-1)/2$.
Similarly, the column index $\nu$ labels the entries of $\boldsymbol{\kappa}$. In the algorithm we vectorize the full matrix column by column, namely
\begin{align}
    \bm M_\nu = \kappa_{\gamma\delta},
    \qquad
    \nu(\gamma,\delta)=\gamma+(\delta-1)N,
    \qquad
    \gamma,\delta=1,\ldots,N.
\end{align}
With this convention, the coefficient multiplying $\kappa_{\gamma\delta}$ in the equation for the pair $(\mA,\mB)$ is placed in the matrix entry $ G_{\eta(\mA,\mB),\nu(\gamma,\delta)} =  g^{\mA\mB}_{\gamma,\delta}$. Collecting all pairwise equations gives the matrix system
\begin{align}
    \bm{\sigma}\approx \bm G\bm M,
    \qquad
    \bm G\in\mathbb R^{N(N-1)/2\times N^2},
\end{align}
where $\bm{\sigma}_{\eta(\mA,\mB)}=\sigma_{\mA\mB}$. The least-squares problem
\begin{align}
    \bm M\in
    \arg\min_{\widehat{\bm M}}
    \|\bm G\widehat{\bm M}-\bm\sigma\|_2
\end{align}
therefore determines a vectorized capillary matrix that best matches the
prescribed effective pairwise coefficients for the current paths. The updated capillary matrix is then set to
\begin{align}
    \boldsymbol{\kappa}=\operatorname{mat}(\bm M),
\end{align}
where the inverse vectorization map $\operatorname{mat}$ is defined by
\begin{align}
    \operatorname{mat}(\bm M)_{\gamma\delta} = M_{\nu(\gamma,\delta)}.
\end{align}
Because only the restriction of the quadratic form associated with
$\boldsymbol{\kappa}$ to the tangent space
\begin{align}
    \bm 1^\perp
    :=
    \left\{
    \mathbf q\in\mathbb R^N:
    \mathbf q\cdot \bm 1=0
    \right\}
\end{align}
affects the gradient energy, the least-squares solution is not unique when the
full matrix $\boldsymbol{\kappa}$ is represented. phases in the direction of
$\bm 1$ are gauge degrees of freedom. We select the minimum-norm least-squares
representative, which fixes this gauge.

The updated capillary matrix is then used to recompute the one-dimensional
equilibrium profiles and the corresponding surface tensions. If the computed
surface tensions do not match the prescribed targets, the process is repeated.

This leads to the design of algorithm \ref{alg:reconstruct-K-2e}. It should be noted that energy equipartition \eqref{eq: EL3} can be checked by computing all three different equations in \cref{subsec:surface_tension}. In general, we do not expect that the solution is positive-definite on $\mathbf{1}^\perp$ cf. Remark \ref{rem:spreading}. 

\begin{figure}
\centering
\begin{algorithm}[H]
\caption{Surface-Tension–Consistent Reconstruction of $\boldsymbol{\kappa}$}
\label{alg:reconstruct-K-2e}
\DontPrintSemicolon
\SetKwInOut{Input}{Input}
\SetKwInOut{Output}{Output}
\Input{$\mathbf{b}^{(\mA)}$, targets $\gamma_{\mA\mB}$, tolerance $\mathrm{tol}>0$}
\Output{$\tilde\gamma_{\mA\mB}$ and $\boldsymbol{\kappa}$ with $\|\gamma_{\mA\mB}-\tilde\gamma_{\mA\mB}\|\le \mathrm{tol}$}
Set $\mathbf{p}^{(\mA\mB)}=\mathbf{b}^{(\mA)}+s(\mathbf{b}^{(\mB)}-\mathbf{b}^{(\mA)})$ for $\alpha<\beta$.
\While{true}{
  Compute $g^{\mA\mB}_{\gamma,\delta} = \int_{0}^{1} p_\gamma^{(\mA\mB)\prime}(s) p_\delta^{(\mA\mB)\prime}(s) \mathrm{d}s$.\;
  Define index maps $\eta = \sum_{r=1}^{\mA-1} (N-r) + (\mB-\mA)$, $\nu = \gamma + (\delta-1)N$.\;
  Set $G_{\eta,\nu} = g^{\mA\mB}_{\gamma,\delta}$ so that $\bm{G}\in\mathbb{R}^{N(N-1)/2\times N^2}$.\;
  Form $\bm{\sigma}$ with entries $\bm{\sigma}_\eta = \sigma_{\mA\mB}$ and the model $\bm{\sigma} = \bm{G}\bm{M}$.\;
  Compute $\bm{M} \in \arg\min_{\widehat{\bm{M}}} \|\bm{G}\widehat{\bm{M}} - \bm{\sigma}\|_2$.\;
  Set $\boldsymbol{\kappa} \leftarrow \operatorname{mat}(\bm{M})$.\;
  Solve for the equilibrium profiles (given $\mathbf{b}^{(\mA)}$ and $\boldsymbol{\kappa}$).\;
  Compute the surface tensions $\tilde\gamma_{\mA\mB}$.\;
  \If{$\|\gamma_{\mA\mB}-\tilde\gamma_{\mA\mB}\| \le \mathrm{tol}$}{
     \Return{$\tilde\gamma_{\mA\mB}, \boldsymbol{\kappa}$}
  }
}
\end{algorithm}
\end{figure}

The reconstruction described above assumes that the coexistence points associated
with the prescribed phases are distinct. If two coexistence points coincide, the
corresponding pair does not define an independent material interface and the
calibration reduces to a lower-dimensional problem. This reduction is described in Appendix \ref{sec:coinciding-coexistence-points}.

\begin{remark}[Straight-line approximation]
When adopting a straight-line approximation, one can resort to a least squares problem using distance matrices. In practice one has to include tolerances and positive definiteness checks of $\boldsymbol{\kappa}$ on the subspace $\bm{1}^\perp.$ A crucial part is to solve for the equilibrium profiles, which is presented in detail in Appendix \ref{app:equicomp}. 
\end{remark}
\begin{remark}[Total spreading]\label{rem:spreading}
In the case of total spreading, i.e. for $N=3$: $\gamma_{\mA\mB} > \gamma_{\mA\delta} + \gamma_{\delta\mB}$ and its higher-dimensional extensions, cf. \cite{boyer2014hierarchy}, a smooth minimizer does not exist. As a consequence, this case is not accounted for by the algorithm (it would break down).
\end{remark}

\subsection{Face confinement of minimizing profiles}\label{subsec:face_confinement}
The preceding surface tension formula involves a minimization over admissible paths in the Gibbs simplex. Whether the minimizing path between two coexistence points is confined to the corresponding simplex face depends on the choice of free energy. We make this distinction for the Flory--Huggins free energy and the Boyer free energy.

We show that for the mixture-aware free energy, an energy-minimizing interfacial profile cannot lie exactly on a Gibbs-simplex face.
\begin{theorem}[No face confinement for the mixture-aware free energy]\label{thm:no-face-confinement-ideal-mixing}
Let $N\geq 3$. A profile that is confined to a Gibbs-simplex face on a
nontrivial interfacial region cannot be a local minimizer of the total
mixture-aware energy with $W_\mA=W>0$ and finite coefficients
$\chi_{\alpha\beta}$ and $\kappa_{\alpha\beta}$.
\end{theorem}

For the free energy proposed by Boyer \eqref{eq:Boyerpot} we show that the minimizing profiles lie on the faces of the Gibbs simplex.
\begin{theorem}[Exact face confinement for the Boyer free energy]
\label{thm:boyer-face-confinement}
Let $N=3$. For the Boyer free energy defined in \cite{boyer2006study}, with
positive pair coefficients and nonnegative ternary penalty, every energy-minimizing one-dimensional profile connecting two pure states is confined
to the corresponding Gibbs-simplex face. Moreover, the minimum energy equals the prescribed surface tension of that pair.
\end{theorem}

The proofs of Theorems \ref{thm:no-face-confinement-ideal-mixing} and \ref{thm:boyer-face-confinement} are provided in the Appendix \ref{appendix sec: minizers}.

\section{Interface width calibration}\label{sec:interface_width_calibration}

In this section we define and calibrate the diffuse-interface widths associated with the pairwise equilibrium profiles. In  \cref{subsec:interface_width}, we introduce an operational definition of the $\mA$--$\mB$ interface width based on the normal-coordinate profile and a prescribed transition interval. In
\cref{subsec:interface_width_parameters}, we fix this interval by requiring that, for the standard binary double-well model, the resulting width coincides with the nominal interface-thickness parameter. Finally, in \cref{sec:width-rescaling}, we introduce a global thickness parameter and use it to set a mesh-resolvable representative interface width without changing the calibrated surface tensions.

\subsection{Interface width}\label{subsec:interface_width}
We define an operational interface width for the equilibrium profile $\hat{\mathbf p}^{(\mA\mB)}(z)$ introduced above. The interface width between $\mathbf b^{(\mA)}$ and $\mathbf b^{(\mB)}$ is determined by fixing constants $a$ and $b$ with $0<a<b<1$, and measuring
the extent in $z$ (the physical coordinate normal to the interface) over which the profile transitions between $a$ and $b$. By construction, the width $\varepsilon_{\mA\mB}(a,b)$ depends on $(a,b)$: moving $a$ and $b$ closer to the pure phases increases the portion of the diffuse tails included in the measurement. For typical diffuse--interface models the equilibrium profile approaches the pure phases only
only asymptotically as $|z|\to\infty$, so that the tails extend to $|z|\to\infty$ and $\varepsilon_{\mA\mB}(0,1)$ would diverge. We
therefore restrict to $0<a<b<1$ and regard $\varepsilon_{\mA\mB}(a,b)$ as an operational definition that must be fixed
by convention.

To avoid a definition that depends on a distinguished phase label, we use the antisymmetric difference
\begin{align}\label{eq:dAB_def}
  d_{\mA\mB}(z) := \hat p_\mA(z)-\hat p_\mB(z),
\end{align}
and choose the symmetric face-projection parameter
\begin{align}\label{eq:s_symmetric_width}
  s(z):=\frac{1+\hat p_\mB(z)-\hat p_\mA(z)}{2}=\frac{1-d_{\mA\mB}(z)}{2}.
\end{align}
Given a bracket $0<a<b<1$, we set $\ell_a:=1-2a$ and $\ell_b:=1-2b$, so that $s(z_a)=a$ and $s(z_b)=b$ are equivalent to $d_{\mA\mB}(z_a)=\ell_a$ and $d_{\mA\mB}(z_b)=\ell_b$. We then define $z_a$ and $z_b$ as the (unique) normal coordinates where 
\begin{align}\label{eq:d_crossings_width}
  d_{\mA\mB}(z_a)=\ell_a,\qquad d_{\mA\mB}(z_b)=\ell_b,
\end{align}
(assuming $d_{\mA\mB}$ is monotone across the interface) and set the profile-based interface width as
\begin{align}\label{eq:width_profile}
  \varepsilon_{\mA\mB}(a,b):=\bigl|z_{b}-z_{a}\bigr|.
\end{align}
This definition depends only on the pair $(\mA,\mB)$ through the antisymmetric difference $d_{\mA\mB}$ and is therefore
invariant under relabeling of phases. The construction is illustrated in \cref{fig:interface_widths} for the three pairwise interfaces in an $N=3$ mixture. The upper panels show the equilibrium volume-fraction profiles, while the lower panels show the difference fields $d_{ij}(z)$ used to determine the crossing points $z_a$ and $z_b$.

\begin{figure}
\captionsetup[subfigure]{justification=centering}
\begin{subfigure}{0.48\textwidth}
\centering
\includegraphics[scale=0.35]{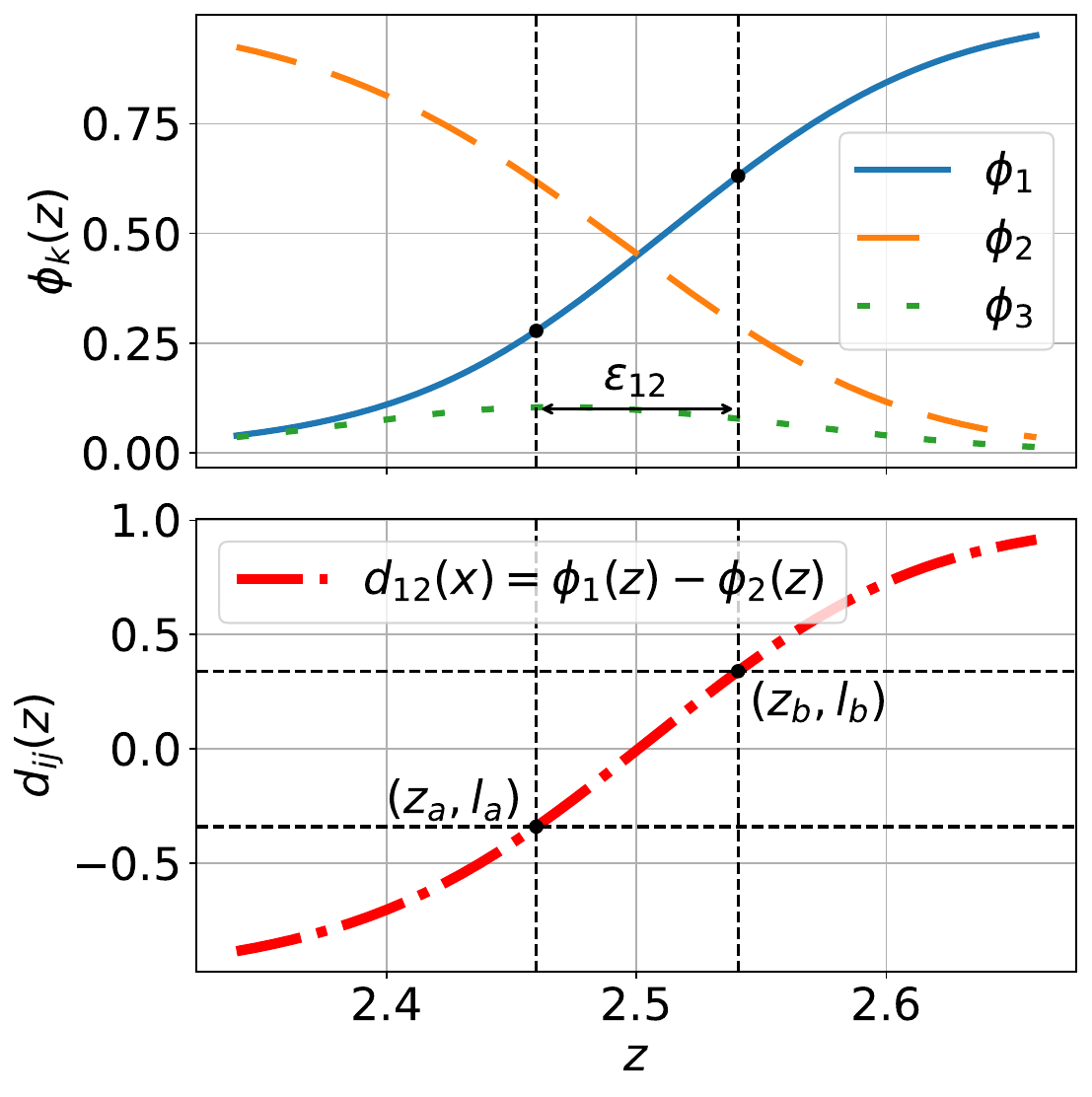}
\caption{Interface $1-2$}
\end{subfigure}
\begin{subfigure}{0.48\textwidth}
\centering
\includegraphics[scale=0.35]{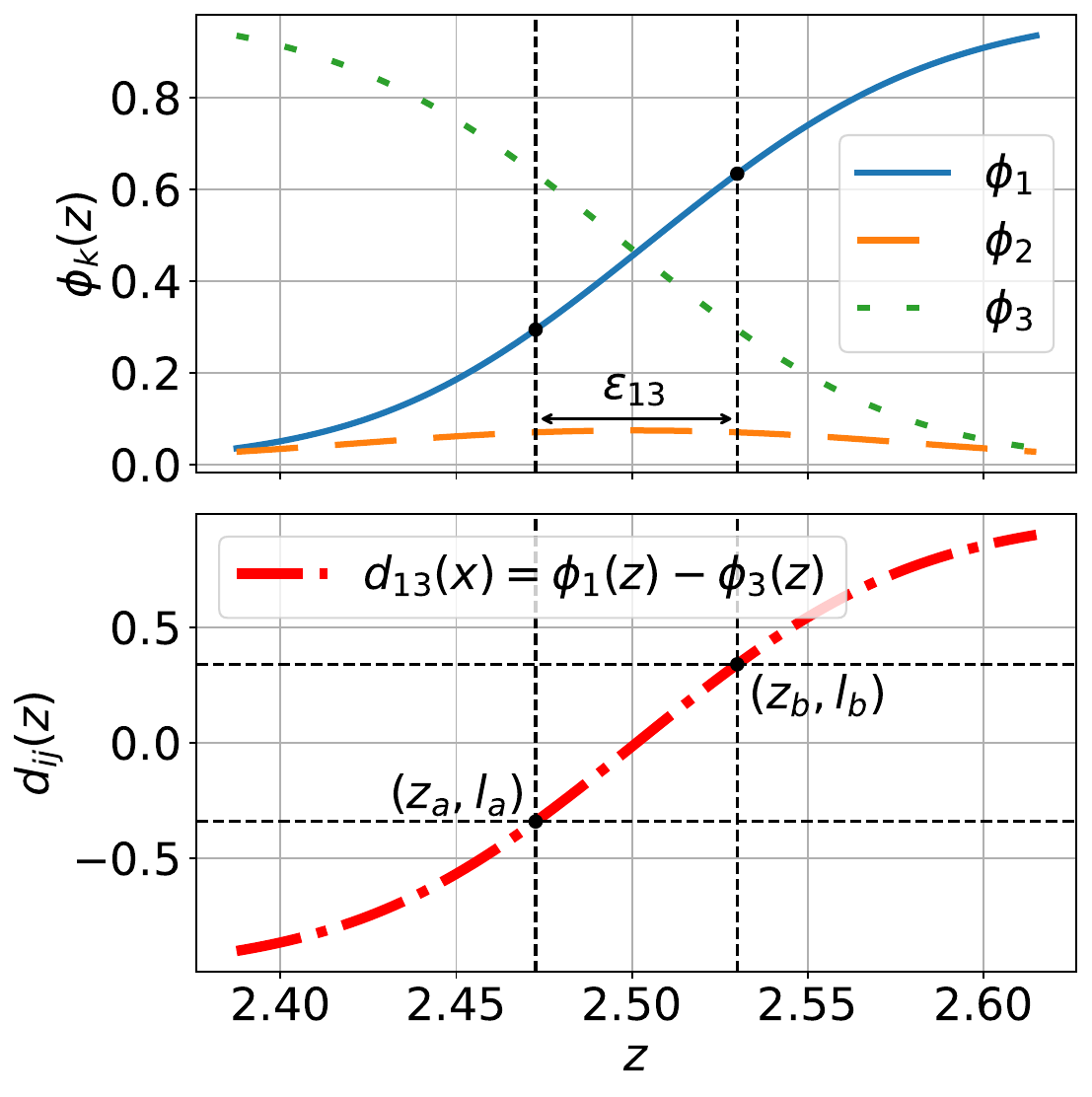}
\caption{Interface $1-3$}
\end{subfigure}
\par\centering
\begin{subfigure}{0.48\textwidth}
\centering
\includegraphics[scale=0.35]{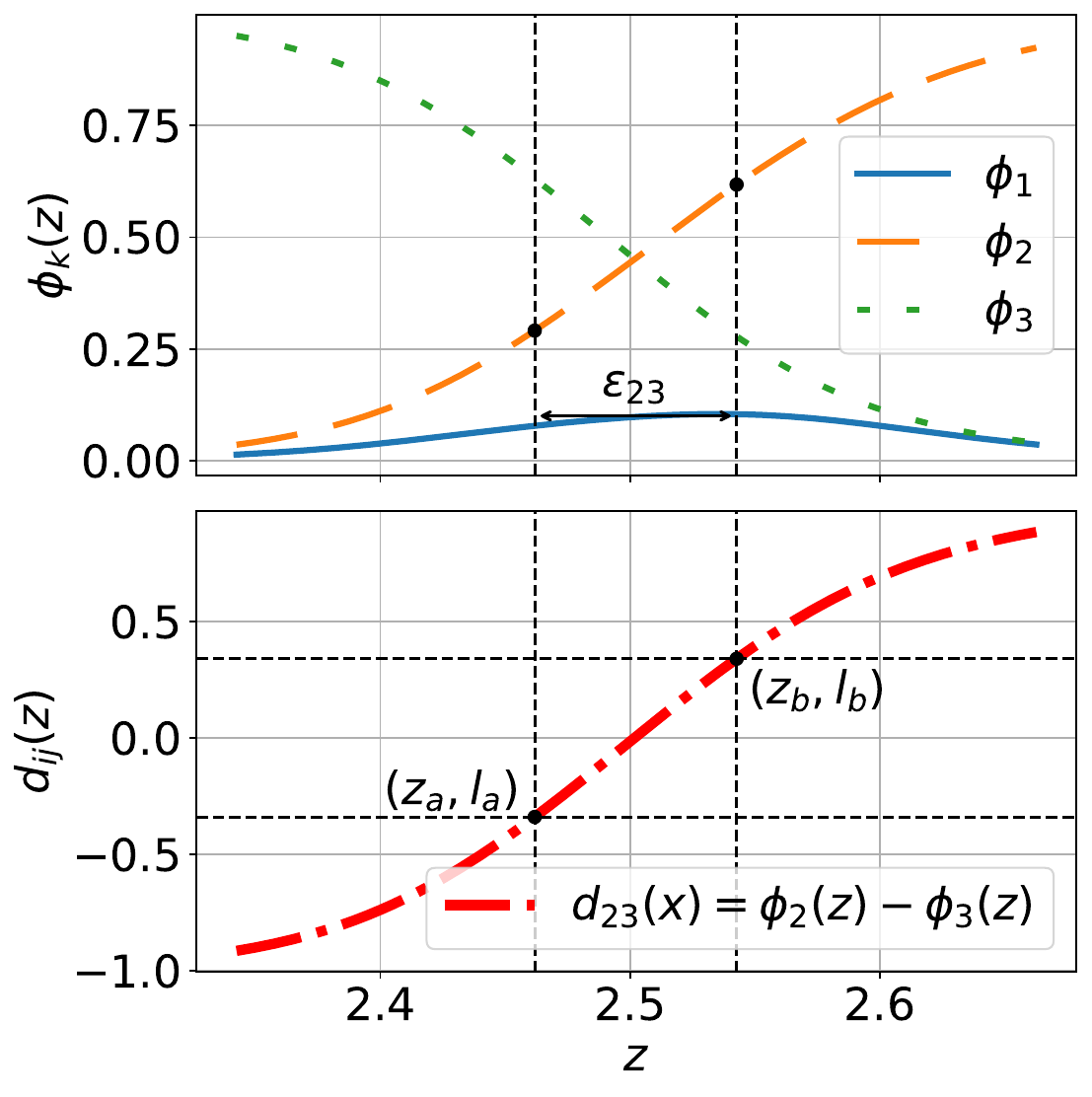}
\caption{Interface $2-3$}
\end{subfigure}
\caption{Interface-width construction for the three pairwise interfaces in an $N=3$ mixture. In each subfigure, the upper panel shows the equilibrium volume-fraction profiles $\phi_k(z)$ across the indicated interface. The dashed vertical lines mark the normal coordinates $z_a$ and $z_b$, and the horizontal arrow indicates the measured width $\varepsilon_{ij}(a,b)=|z_b-z_a|$. The lower panel shows the corresponding antisymmetric difference field $d_{ij}(z)=\phi_i(z)-\phi_j(z)$; the coordinates $z_a$ and $z_b$ are determined by the crossings $d_{ij}(z_a)=\ell_a$ and $d_{ij}(z_b)=\ell_b$.}
\label{fig:interface_widths}
\end{figure}

The same width can be evaluated from the free-energy representation of the equilibrium path. First we note the identity 
\begin{align}\label{eq:width_def_energy}
  \varepsilon_{\mA\mB}(a,b)=\bigl|z_{s=b}-z_{s=a}\bigr|
  =\left|\int_a^b \frac{dz}{ds}\,ds\right|,
\end{align}
which holds for the equilibrium path parametrized by a strictly monotone scalar parameter $s:\mathbb{R}\to[0,1]$. Inserting \eqref{eq: dzds_surface_tension} into \eqref{eq:width_def_energy} yields the free-energy-based expression
\begin{align}\label{eq:width_energy}
  \varepsilon_{\mA\mB}(a,b)
  =\int_a^b
  \sqrt{\frac{\mathbf p'(s)^{\!\top}\boldsymbol{\kappa}\,\mathbf p'(s)}
  {2\,\Omega_{\mathbf{b}^{(\mA)}}(\mathbf p^{(\mA\mB)}(s))}}\,ds.
\end{align}

\subsection{Specification of the interface width parameters}\label{subsec:interface_width_parameters}
We specify $a$ and $b$ so that the standard interface width associated with the standard double well potential for 2 phases is returned. As such, we consider a two–fluid system for which we introduce a single order parameter
$\phi:=\phi_1$ so that $\phi_2=1-\phi_1$. The free–energy density
can be written as
\begin{align}\label{eq:Cahn-Hilliard-free-energy-2-fluid}
    \Psi= \breve{\Psi}_0(\phi)
      + \frac{1}{2}\sigma_{12}\nabla\phi\cdot\nabla\phi,
\end{align}
where $\breve{\Psi}_0(\phi) := \Psi_0(\phi,1-\phi)$ and where
$\sigma_{12} = \kappa_{11}-2\kappa_{12}+\kappa_{22}$. We take the standard double well potential $\Psi_0 = 4 W \phi^2(1-\phi)^2$ (which may be written as $W(1-\varphi^2)^2/4$ with $\varphi=\phi_1-\phi_2$) for which we find:
\begin{subequations}
  \begin{align}
    \gamma_{12} =&~ \frac{\sqrt{2}}{3}\sqrt{W \sigma_{12}},\qquad \varepsilon_{12}(a,b) = ~ \sqrt{\frac{\sigma_{12}}{4W}}\frac{1}{\sqrt{2}} \left(\log \frac{b}{1-b}-\log \frac{a}{1-a}\right).
  \end{align}
\end{subequations}

We proceed in three steps:

\begin{enumerate}
\item \emph{Symmetric bracket.}
We choose $a$ and $b$ symmetrically around $\phi=1/2$, i.e.
\begin{align}
a = 1-b < \frac{1}{2} < b = 1-a.
\end{align}
Substituting this into the expression for $\varepsilon_{12}(a,b)$
yields
\begin{align}
\varepsilon_{12}(a,b)
= \sqrt{\frac{\sigma_{12}}{4W}}
\frac{1}{\sqrt{2}}
\log\frac{b^2}{(1-b)^2}.
\end{align}

\item \emph{Standard scaling.}
A common parametrization in phase–field models is
\begin{align}
\sigma_{12} = 4\gamma\epsilon,
\qquad
W = \frac{\gamma}{\epsilon},
\end{align}
where $\gamma$ is a characteristic surface–tension scale and
$\epsilon$ is a nominal interface thickness parameter. With this
choice we find $\sqrt{\sigma_{12}/(4W)} = \epsilon$, so that the interface thickness becomes:
\begin{align}
\varepsilon_{12}(a,b)
= ~ \epsilon\frac{1}{\sqrt{2}}
\log\frac{b^2}{(1-b)^2}.
\end{align}

\item \emph{Matching the nominal thickness.}
We now define $a$ and $b$ by requiring that the interface
width coincides with the nominal thickness, i.e.
\begin{align}
\varepsilon_{12}(a,b) = \epsilon.
\end{align}
This leads to
\begin{align}
  b=  \frac{1}{2}\left(1+\tanh\left(\frac{1}{2\sqrt{2}}\right)\right)
\approx 0.670,
\end{align}
and thus $a=1-b\approx 0.330$. This fixes a concrete, symmetric
bracket $(a,b)$ that reproduces the standard interface thickness
$\epsilon$ for the double–well potential.
In Fig.~\ref{fig:two-fluid-Psi0-interface-width} we plot the bulk
potential $\Psi_0(\phi)=4W\phi^2(1-\phi)^2$. The red vertical lines and the shaded band
indicate the interval $[\phi=a,\phi=b]$, which is the
portion of composition space that we count as interface width
$\varepsilon_{12}(a,b)$.
In Fig.~\ref{fig:two-fluid-phi-profile-interface-width} we show the
corresponding equilibrium profile $\phi(z)$ for the same potential and
capillary coefficient. 
\end{enumerate} 

\begin{figure}
\captionsetup[subfigure]{justification=centering}
\begin{subfigure}{0.48\textwidth}
\includegraphics[scale=0.41]{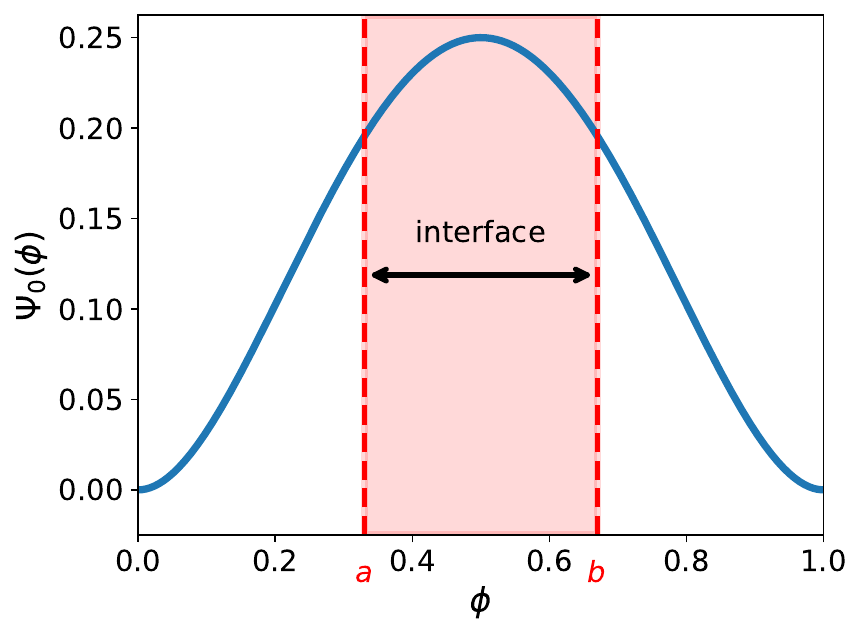}
  \caption{Double–well bulk potential
    $\Psi_0(\phi) = 4\phi^2(1-\phi)^2$.}
  \label{fig:two-fluid-Psi0-interface-width}
\end{subfigure}
\begin{subfigure}{0.49\textwidth}
\includegraphics[scale=0.41]{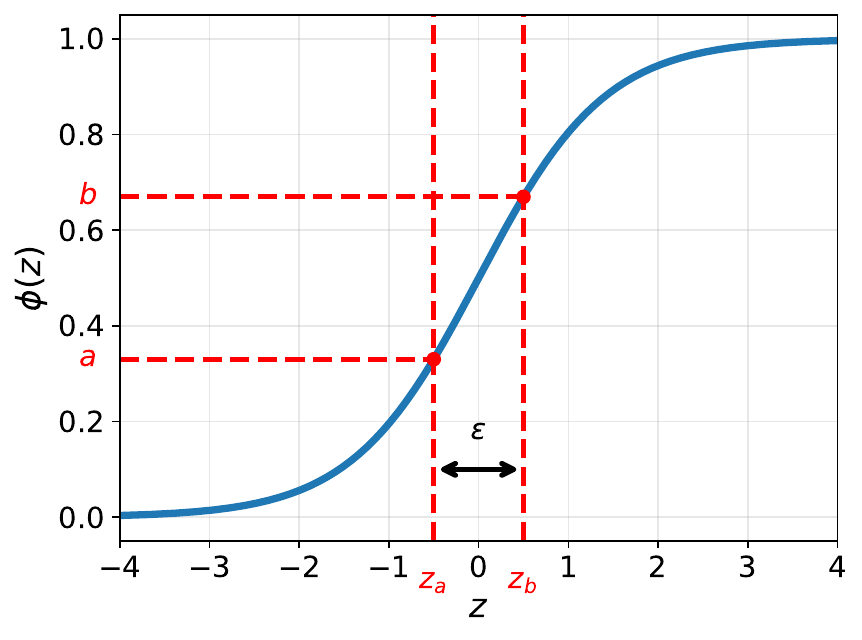}
\caption{Equilibrium one–dimensional interface profile $\phi(z)$.}
  \label{fig:two-fluid-phi-profile-interface-width}
\end{subfigure}
\caption{Visualization of the interface width for the standard double well potential.}
\label{fig: interface width}
\end{figure}

We thus set the interface width as $\varepsilon_{\mA\mB}=\varepsilon_{\mA\mB}(a,b)$, where $a$ and $b$ are defined above.

\subsection{Mesh-resolvable interface width}\label{sec:width-rescaling}

To expose explicitly the diffuse-interface thickness scale, we introduce a
parameter $\varepsilon_0>0$ and choose
\begin{align}
    W_\alpha=\frac{\bar W}{\varepsilon_0},
    \qquad
    \kappa_{\mA\mB}
    =
    \varepsilon_0\bar\kappa_{\mA\mB},
\end{align}
where $\bar W>0$ and $\bar\kappa_{\mA\mB}$ are independent of
$\varepsilon_0$. Substitution into the interface width formula \eqref{eq:width_energy} gives
\begin{align}\label{eq:epsilon_AB_scaling}
    \varepsilon_{\mA\mB} =&~ \varepsilon_0\bar\varepsilon_{\mA\mB}
\end{align}
where the reference pairwise interface width is
\begin{align}\label{eq:bar_epsilon_AB}
  \bar\varepsilon_{\mA\mB}
  :=
  \int_a^b
  \sqrt{
  \frac{
  \mathbf p'(s)^{\!\top}
  \bar{\boldsymbol{\kappa}}
  \mathbf p'(s)}
  {2\,\bar\Omega_{\mathbf{b}^{(\mA)}}(\mathbf p^{(\mA\mB)}(s))}
  }\,ds ,
\end{align}
where $\bar\Omega_{\mathbf{b}^{(\mA)}}= \varepsilon_0 \Omega_{\mathbf{b}^{(\mA)}}$ denotes the corresponding excess energy contribution obtained with $\bar W$.
Since an $N$-phase mixture generally contains several pairwise interfaces, the
values $\bar\varepsilon_{\mA\mB}$ are in general not identical. We therefore define the
representative interface width as the thinnest pairwise width:
\begin{align}\label{eq:representative_width}
    \varepsilon
    :=
    \min_{\mA<\mB}\varepsilon_{\mA\mB}
    =
    \varepsilon_0
    \min_{\mA<\mB}\bar\varepsilon_{\mA\mB}.
\end{align}
Consequently, prescribing the representative width $\varepsilon$ determines $    \varepsilon_0$. With this choice, the thinnest pairwise diffuse interface has width
$\varepsilon$, while all other pairwise interfaces are at least as wide. In the
numerical simulations we choose $\varepsilon$ in relation to the mesh size. The surface tensions are unaffected by this choice of $\varepsilon_0$.

\section{Numerical experiments for surface tension calibration}\label{sec:num_exp}
In the following we will employ our algorithm to compute capillary matrices and the associated interface width from a given bulk potential and surface tensions. In \cref{subsec:numerics_two_phase}, we first consider two-phase examples, including the Ginzburg--Landau potential for which an exact reference formula is available. In \cref{subsec:numerics_three_phase}, we study a three-phase system and compare the behavior of the Boyer and
regularized Flory--Huggins free energies. Finally, in
\cref{subsec:numerics_four_phase}, we consider a four-phase example. In Appendix \ref{sec:coinciding-coexistence-points} we provide numerical experiments for a reduction-consistent case ($N=3$ to $N=2$).

The algorithm is implemented in Python using NGSolve \cite{schoberl2014c++} to solve for the equilibrium profiles and to compute the surface tension, by solving the Euler-Lagrange equations via Newton’s method. The tolerance for the Newton update is $10^{-10}$ and the surface tension error (Euclidean norm) is $10^{-8}$. We will consider the result of the straight-line approximation ($\boldsymbol{\kappa}_0$) and the converged capillary matrix ($\boldsymbol{\kappa}_{\rm eq}$). Furthermore we also compute the approximate interface widths for the equilibrium solutions related to $\boldsymbol{\kappa}_{\rm eq}$.

\subsection{Two-phase cases}
\label{subsec:numerics_two_phase}
To verify our results we first apply our scheme for the Ginzburg--Landau potential $\Psi_0(\phi_1,\phi_2)=4\phi_1^2\phi_2^2$ and choose $\gamma_{12}=0.007$. The algorithm results in
\begin{align*}
    \boldsymbol{\kappa}_0 =10^{-5}\begin{pmatrix}
    2.7562 & - 2.7562 \\
  - 2.7562 &  2.7562
\end{pmatrix}, \qquad 
    \boldsymbol{\kappa}_{\rm eq} =10^{-5} \cdot\begin{pmatrix}
   5.5028  & -5.5028 \\
  -5.5028  &  5.5028
\end{pmatrix}.
\end{align*}
For this potential, the exact surface tension is
\begin{align}\label{eq:GL_ST_formula}
    \gamma_{12}(\boldsymbol{\kappa})=\frac{\sqrt{2}}{3}\sqrt{\kappa_{11}-2\kappa_{12}+\kappa_{22}}.
\end{align}
Using \eqref{eq:GL_ST_formula} for the initial and converged matrices gives
\begin{align}
    \gamma_{12}(\boldsymbol{\kappa}_0) \approx 0.00495, \qquad \gamma_{12}(\boldsymbol{\kappa}_{\rm eq}) \approx 0.00699.
\end{align}
Thus the straight-line initialization underestimates the prescribed surface
tension, whereas the reconstructed matrix recovers the target value up to the
specified tolerance.

In the remaining experiments we use the polynomial regularization of the mixture-aware potential described in \cite{mixtureaware2026}. We start again with a simple two-phase system and choose again $\gamma_{12}=0.007$. The visualization of the interface is given in Figure \ref{fig:n2_profiles} with interface width $\varepsilon=0.0071$ and the capillary matrices
\begin{align*}
    \boldsymbol{\kappa}_0 =10^{-5}\begin{pmatrix}
   2.4517 & -2.4517 \\
  -2.4517 & 2.4517
\end{pmatrix}, \qquad 
    \boldsymbol{\kappa}_{\rm eq} =10^{-5} \cdot\begin{pmatrix}
   4.8955  & -4.8955 \\
  -4.8955  &  4.8955
\end{pmatrix}.
\end{align*}
Using \eqref{eq:GL_ST_formula} we find
\begin{align}
    \gamma_{12}(\boldsymbol{\kappa}_0) = 0.004669, \qquad \gamma_{12}(\boldsymbol{\kappa}_{\rm eq}) = 0.006598.
\end{align}
Although this formula is not exact for the regularized Flory-Huggins potential,
it provides a useful indication of the size of the correction from
$\boldsymbol{\kappa}_0$ to $\boldsymbol{\kappa}_{\rm eq}$.

\begin{figure}
    \centering
    \includegraphics[width=0.5\linewidth]{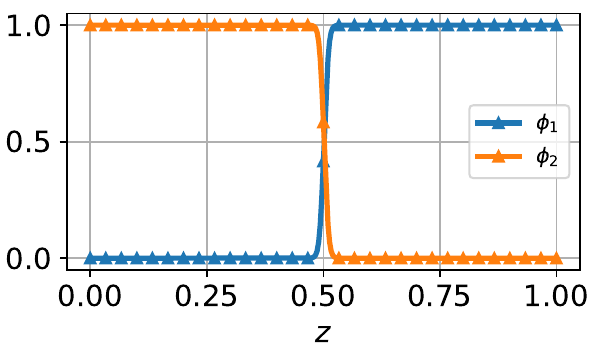}
    \caption{Equilibrium profile for the two-phase regularized Flory--Huggins example.}
    \label{fig:n2_profiles}
\end{figure}

\subsection{Three-phase cases}
\label{subsec:numerics_three_phase}
Next we consider three phases, i.e. $N=3$, and choose as surface tensions
\begin{align}
    \gamma_{12} = 0.007, \qquad \gamma_{13} = 0.005, \qquad \gamma_{23}=0.006.
\end{align}
For the Boyer potential \eqref{eq:Boyerpot} we verified that the algorithm yields, up to discretization and round-off errors, the exact capillary matrix, cf. \cref{thm:boyer-face-confinement}. In the case of the regularized mixture-aware potential the resulting capillary matrices are:
\begin{align}\label{eq:cap3three}
    \boldsymbol{\kappa}_0&=10^{-5}\begin{pmatrix}
  2.4906 & -2.0458 & -0.4447 \\
  -2.0458 & 3.2244 & -1.1785 \\
  -0.4447 & -1.1785 & 1.6233
\end{pmatrix}, \quad \boldsymbol{\kappa}_{\rm eq}=10^{-5}\begin{pmatrix}
  5.9377 & -4.9313 & -1.0064 \\
  -4.9313 & 7.7978 & -2.8664 \\
  -1.0064 & -2.8664 & 3.87287
\end{pmatrix}.
\end{align}
The corresponding interface widths are
\begin{align}
    \varepsilon_{12}=0.0084, \qquad \varepsilon_{13}=0.0060, \qquad \varepsilon_{23}=0.0071.
\end{align}
The equilibrium profiles and their projection onto the Gibbs simplex are shown in \cref{fig:n3_profiles}. The markers in the equilibrium profile plot correspond to the same markers in the Gibbs triangle. 

\begin{figure}
 \centering
 \begin{tabular}{c@{}c}
  \includegraphics[width=0.49\linewidth]{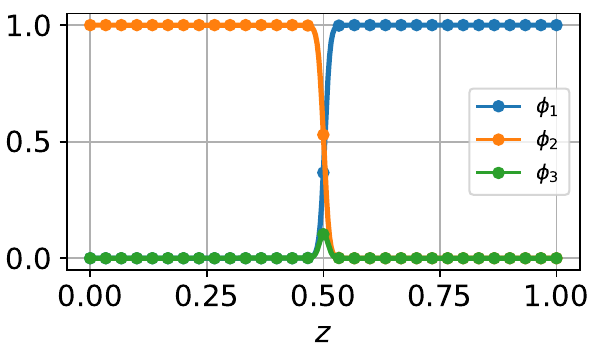}
  &
  \includegraphics[width=0.49\linewidth]{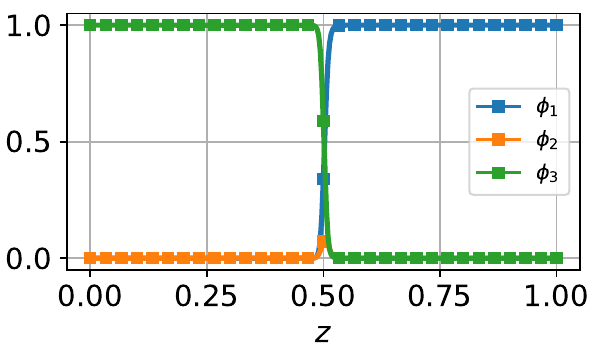} \\
  \includegraphics[width=0.49\linewidth]{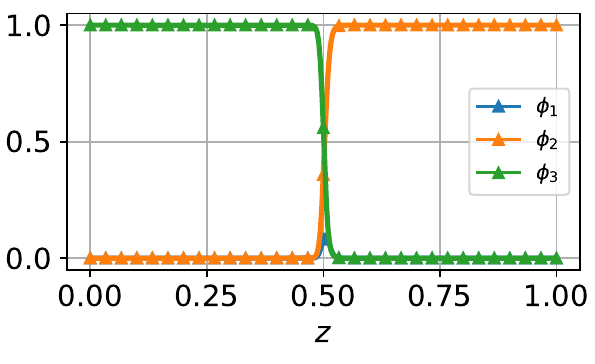}
  &
  \includegraphics[width=0.49\linewidth]{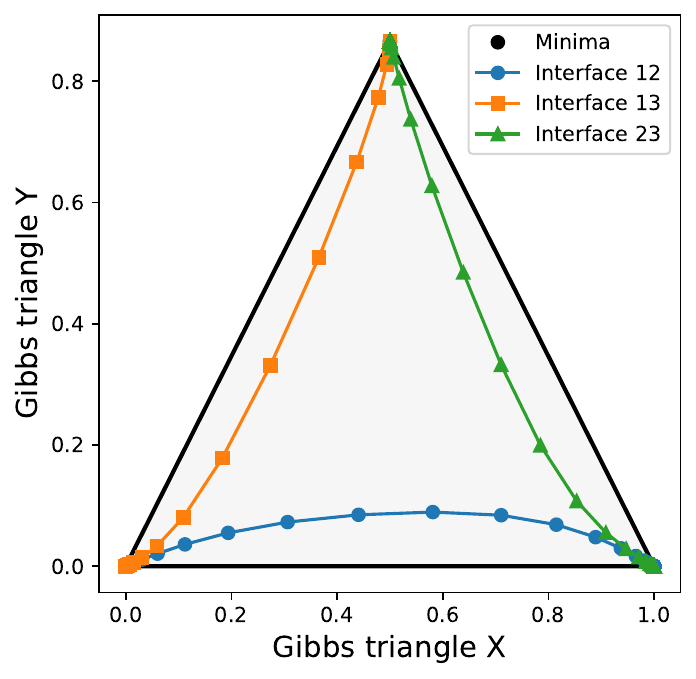} 
 \end{tabular}
 \caption{Equilibrium profiles for the three pairwise interfaces in the  $N=3$ regularized mixture-aware potential, together with their projection onto
 the Gibbs simplex.}
 \label{fig:n3_profiles}
\end{figure}

\subsection{Four-phase cases}
\label{subsec:numerics_four_phase}

Finally we consider a four-phase system. This is the smallest $N$ such that there are more interfaces than phases. We choose
\begin{align}
  \gamma_{12}=0.007,\qquad \gamma_{13}=0.01,\qquad \gamma_{23}=0.009, \nn\\
  \gamma_{14}=0.006,\qquad \gamma_{24}=0.005,\qquad \gamma_{34}=0.009.
\end{align}
The computed interface widths are
\begin{align}
  \varepsilon_{12}=0.0104,\qquad \varepsilon_{13}=0.0150,\qquad \varepsilon_{23}=0.0131, \nn\\
  \varepsilon_{14}=0.0085,\qquad \varepsilon_{24}=0.0072,\qquad \varepsilon_{34}=0.0130.
\end{align}
The corresponding capillary matrices are
\begin{align}
    \boldsymbol{\kappa}_0&=10^{-5}\begin{pmatrix}
         4.6031 & -1.0507 & -3.4774&  -0.0750\\
        -1.0507 &  3.1021 & -2.3266 &  0.2751\\
        -3.4774 & -2.3266 &  8.4558 & -2.6518\\
        -0.0750 &  0.2751 & -2.6518 &  2.4517
    \end{pmatrix}, \\
    \boldsymbol{\kappa}_{\rm eq}&=10^{-4}\begin{pmatrix}
  1.3682 & -0.3156 & -1.0588 & 0.0062 \\
  -0.3156 & 0.9151 & -0.7243 & 0.1248 \\
  -1.05881 & -0.7243 & 2.6552 & -0.8721 \\
  0.0062 & 0.1248 & -0.8721 & 0.7410
\end{pmatrix}.
\end{align}
The resulting interfaces are depicted in Figure \ref{fig:n4_profiles}, and their projections into the Gibbs simplex in Figure \ref{fig:gibbs_N4}. As in the
three-phase case, the markers in the equilibrium profile plot correspond to the same markers in the Gibbs simplex. 

\begin{figure}
 \centering
 \begin{tabular}{c@{}c}
  \includegraphics[width=0.49\linewidth]{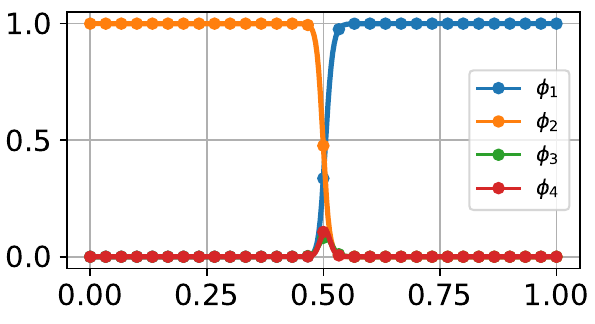}
  &
  \includegraphics[width=0.49\linewidth]{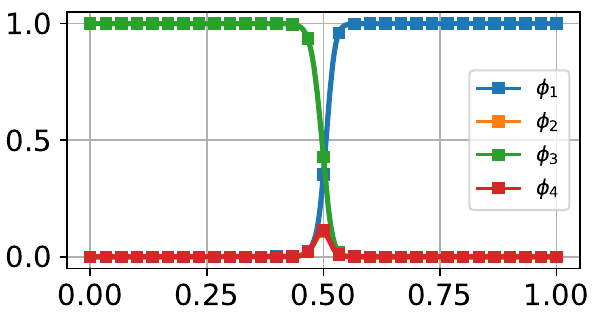} \\
  \includegraphics[width=0.49\linewidth]{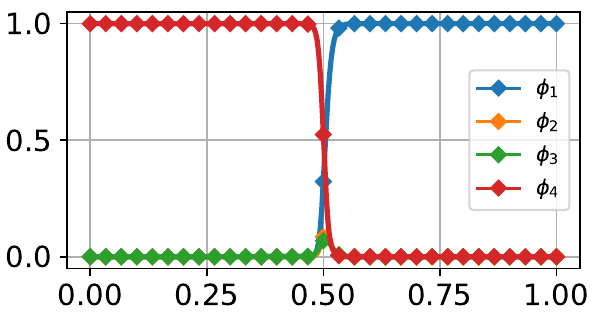}
  &
  \includegraphics[width=0.49\linewidth]{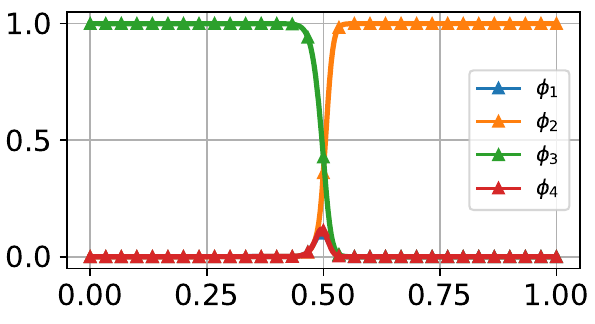} \\
  \includegraphics[width=0.49\linewidth]{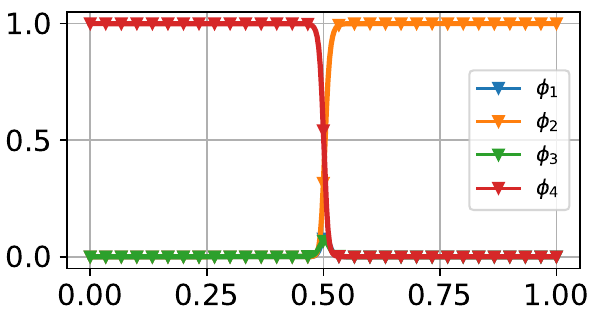}
  &
  \includegraphics[width=0.49\linewidth]{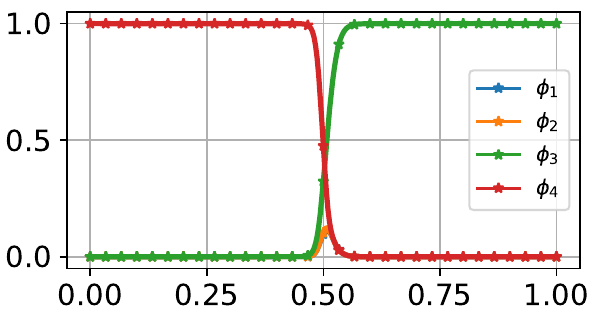} \\
 \end{tabular}
 \caption{Equilibrium profiles for the six pairwise interfaces in the $N=4$ regularized mixture-aware potential.}
 \label{fig:n4_profiles}
\end{figure}

\begin{figure}
    \centering
    \includegraphics[width=0.7\linewidth]{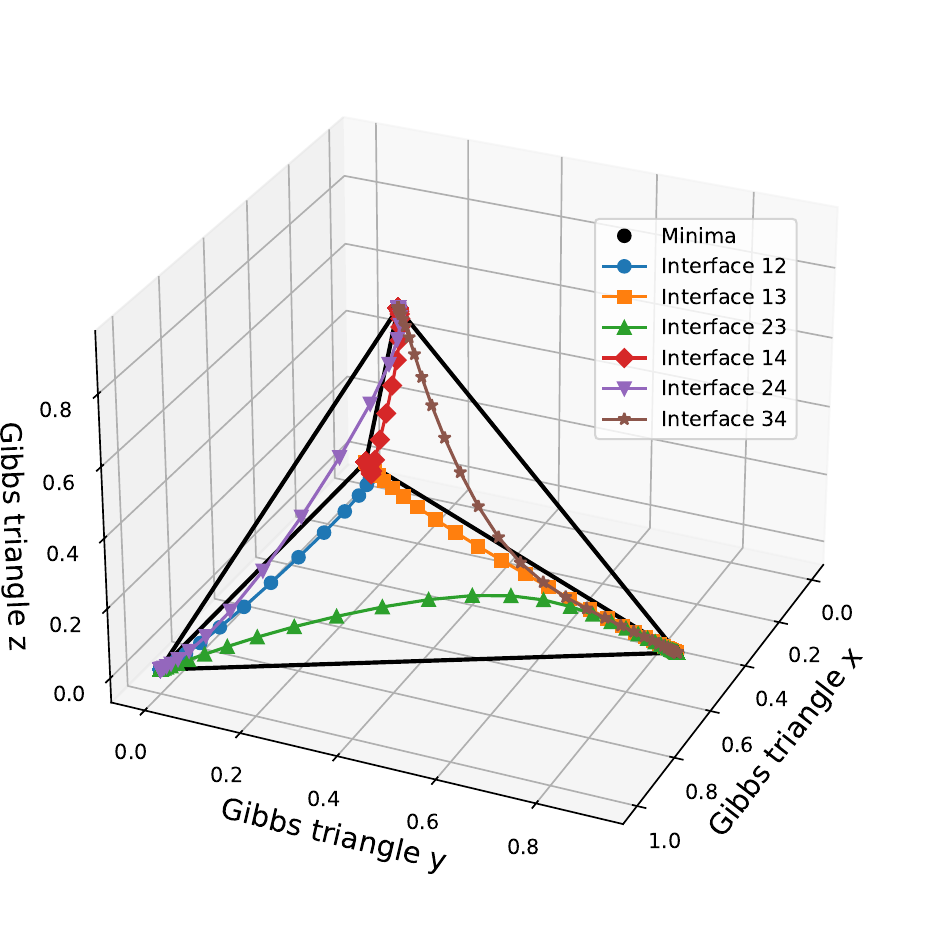}
    \caption{Projection of the $N=4$ equilibrium profiles into the Gibbs simplex.}
    \label{fig:gibbs_N4}
\end{figure}

\subsection{Positive definiteness and convergence behavior.}

The matrices reported above are not positive definite on $\mathbb{R}^{N}$. However, they are positive definite on the restricted space $\bm{1}^\perp$, which is the relevant space due to the saturation constraint. If one insists on a positive definite matrix on the full space, one can use $\boldsymbol{\kappa}+\alpha \bm{1}\otimes\bm{1}$ for $\alpha$ large enough.

In all three cases we computed the error of the approximated surface tension in the standard Euclidean norm as well as the error in the $L^2(-1,1)$-norm of every interface configuration between two successive iterates ($j+1$ and $j$), i.e. $\norm{\boldsymbol{\phi}^{(j+1)}_{\mA\mB}-\boldsymbol{\phi}^{(j)}_{\mA\mB}}_{L^2}$ for all interfaces $\mA-\mB$ (see Figure \ref{fig:errors}). 

\begin{figure}
    \centering
    \begin{subfigure}{0.49\linewidth}
        \includegraphics[trim={0cm 14cm 0cm 0cm}, clip, width=\linewidth]
        {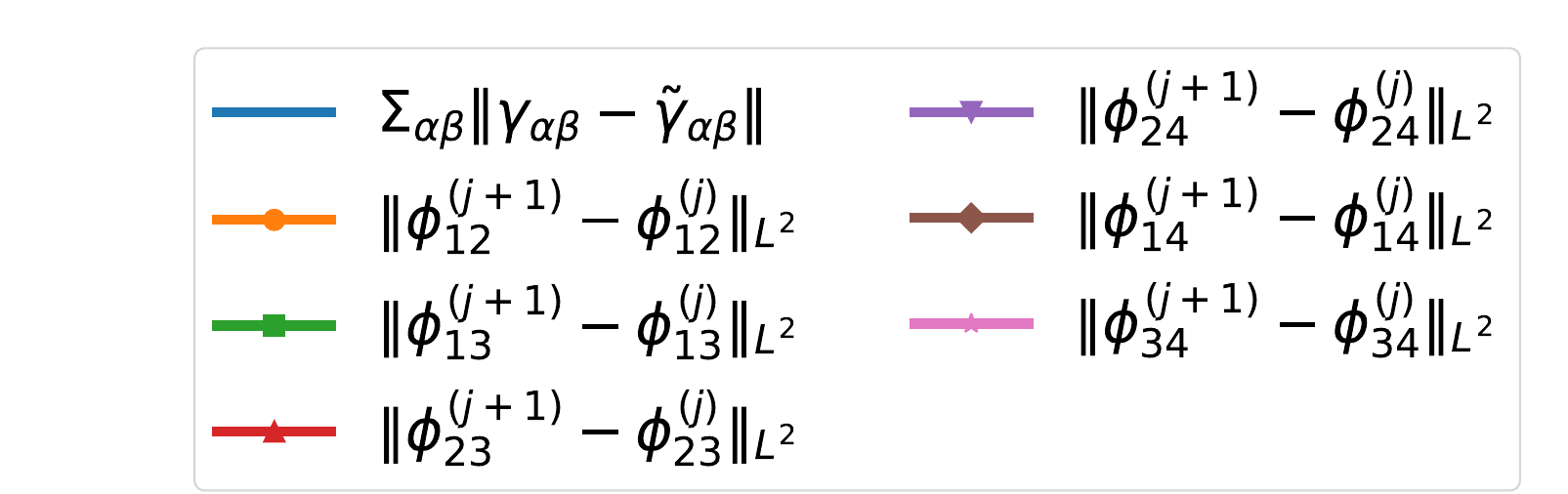}
        \vspace{0.5cm} 
        \caption{Legend for all three plots}
    \end{subfigure}
    \hfill
    \begin{subfigure}{0.49\linewidth}
        \includegraphics[trim={0cm 0cm 0cm 3cm}, clip, width=\linewidth]
        {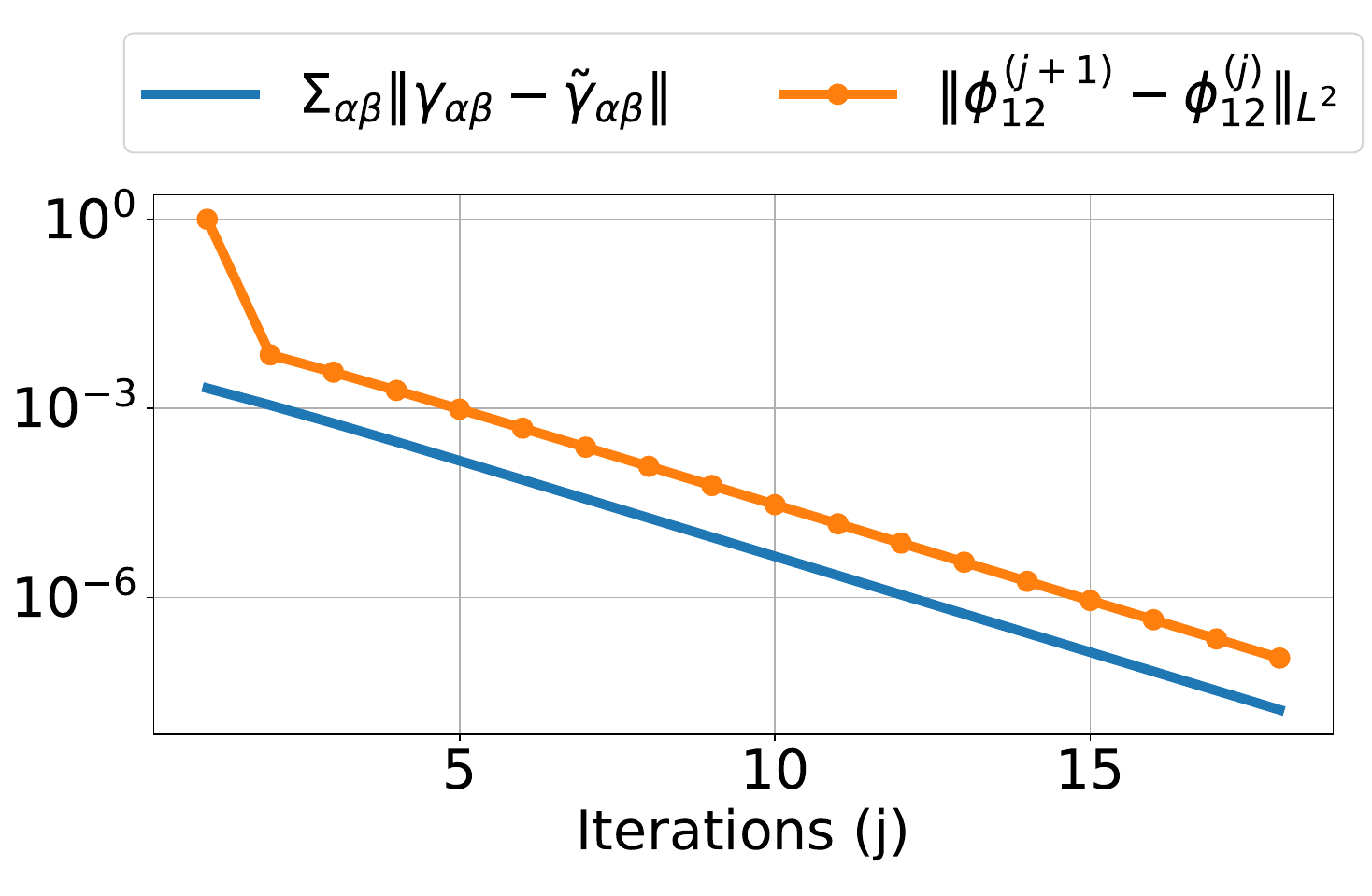}
        \caption{Error in surface tensions and iterates for $N=2$.}
    \end{subfigure}

    \medskip

    \begin{subfigure}{0.49\linewidth}
        \includegraphics[trim={0cm 0cm 0cm 5cm}, clip, width=\linewidth]
        {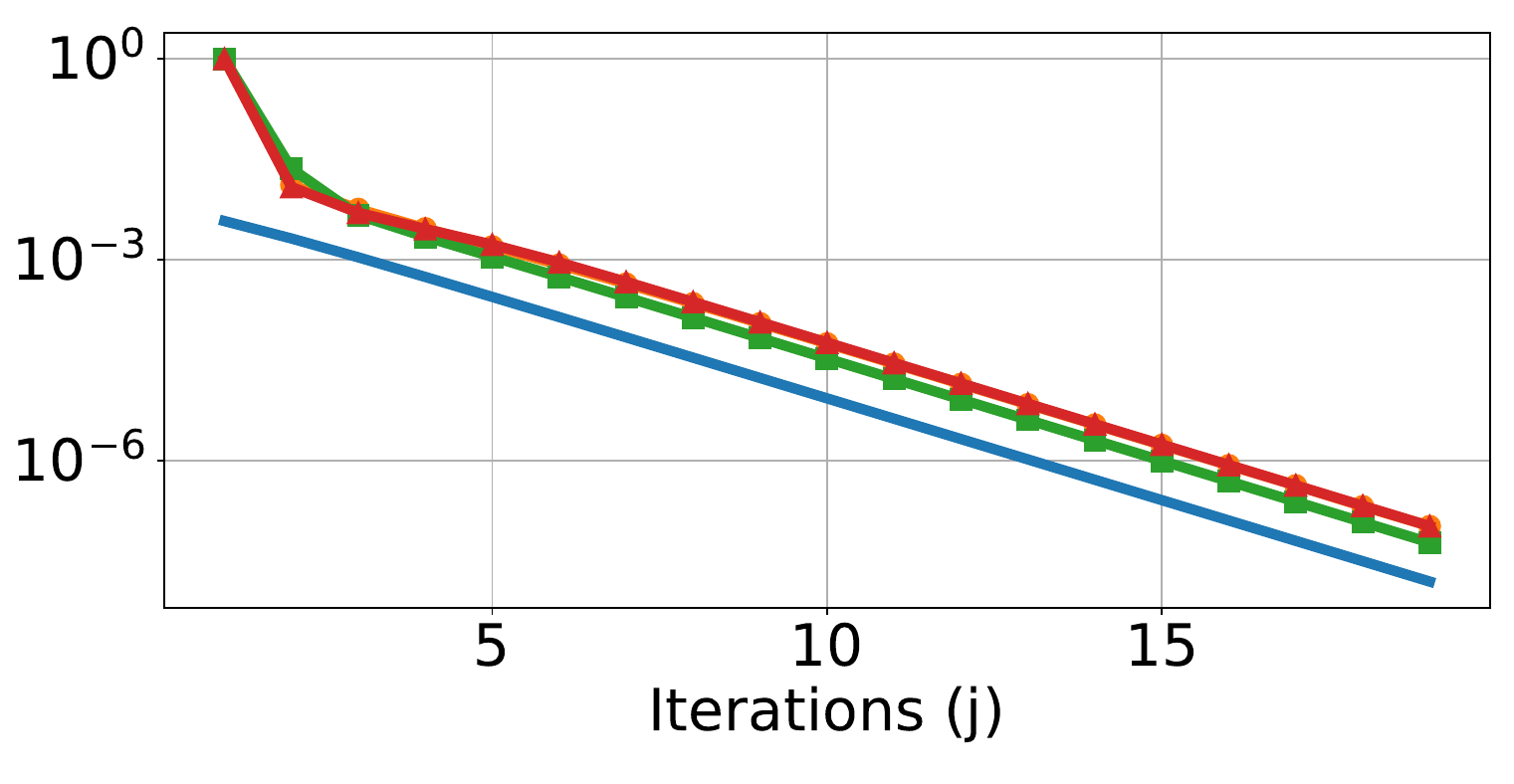}
        \caption{Error in surface tensions and iterates for $N=3$.}
    \end{subfigure}
    \hfill
    \begin{subfigure}{0.49\linewidth}
        \includegraphics[trim={0cm 0cm 0cm 9cm}, clip, width=\linewidth]
        {figures/surface_tension/N4/errors.pdf}
        \caption{Error in surface tensions and iterates for $N=4$.}
    \end{subfigure}
    \caption{Convergence behavior of the surface-tension reconstruction algorithm.
The plots show the surface-tension error and the $L^2(-1,1)$-norm of the change
in the pairwise equilibrium profiles between successive reconstruction steps.}
    \label{fig:errors}
\end{figure}

\section{Applications}\label{sec:Applications}
In this section we apply the calibrated free-energy closures in two model problems. First, in \cref{subsec:lens-test}, we consider a three-phase lens test. This test isolates the effect of the calibrated pairwise surface tensions by comparing the equilibrium triple-junction angles with the sharp-interface Young's law. Second, in \cref{sec:results:rising_bubble}, we consider an axisymmetric
Navier--Stokes--Cahn--Hilliard simulation of a gas bubble rising through two stratified liquid layers. This example tests the calibrated closure in a density-contrast flow with three interacting interfaces. The precise definition of the NSCH model is provided in \cref{appendix: NSCH model} (and its associated CH model is provided in \cref{appendix: CH model}).

\subsection{Lens test}\label{subsec:lens-test}
The lens test is a standard equilibrium test for three-phase surface-tension calibration. Starting from a nonequilibrium configuration, the phase fields relax toward a lens-shaped equilibrium in which the three diffuse interfaces meet in a triple-junction region. In the corresponding sharp-interface setting, the contact angles are determined solely by the three pairwise surface tensions
through Young's law, \cite{Toth2016}. The purpose of the test is therefore to verify that the calibrated capillary matrix produces the prescribed surface-tension ratios at equilibrium.

We consider the final time $T=1$ and the domain $\Omega=[0.3,0.7]\times[0.1,0.4]$ and periodic boundary conditions in $x$-direction and closed boundary conditions in $y$-direction, i.e. zero-flux for $\sum_\mB M_{\mA\mB}\nabla g_\mB$ and homogeneous Dirichlet boundary conditions for $\vv$. We consider well-known lens test initial data, cf. \cite{Yang21}, given by
\begin{align*}
    z_1 &= 0.5\tanh\left(\frac{0.09 - \sqrt{(x-0.5)^2 + (y-0.25)^2}}{0.5\hat{\varepsilon}}\right) + 0.5, \\
    z_2 &= (1-z_1)\left(0.5+0.5\tanh\left(\frac{4(y-0.25)}{\hat{\varepsilon}} \right)\right), \\
    \boldsymbol{\phi}_0 &= (z_2,1-z_1-z_2,z_1)^\top, \quad \vv_0=(0,0)^\top,
\end{align*}
with initial scaling parameter $\hat{\varepsilon} = 10^{-2}$.
To accelerate the convergence to equilibrium we take the constant mobility matrix $M_{\mA\mB}=\delta_{\mA\mB}-\frac{1}{3}$ (instead of the degenerate mobility matrix), the constant viscosity $\nu=0.1$, and equal partial densities $\rho_1=\rho_2=\rho_3=1$. We use the mixture-aware free energy, and set the surface tensions to $\gamma_{12}=0.07, \gamma_{23}=0.05, \gamma_{13}=0.06$ which results in the capillary matrix \eqref{eq:cap3three}. We choose the time-step $\tau=10^{-3}$ and mesh resolution $h\approx 0.0033$.

For a sharp interface Young's relations, \cite{garcke1999}, implies that the contact angles $\theta_1,\theta_2,\theta_3$ are determined by the surface tensions by
\begin{align*}
 \frac{\gamma_{23}}{\sin(\theta_1)} = \frac{\gamma_{13}}{\sin(\theta_2)}= \frac{\gamma_{12}}{\sin(\theta_3)}.   
\end{align*}
The angles can be computed by constructing a triangle with side lengths $\gamma_{\mA\mB}$. In this case the $\theta_\mA$, in $\mathrm{rad}$, correspond to the external angles which yields
\begin{align*}
    \theta_1&= \pi-\arccos\left( \frac{\gamma_{12}^2 + \gamma_{13}^2 - \gamma_{23}^2}{2\gamma_{12}\gamma_{13}}\right), \\
    \theta_2&= \pi-\arccos\left( \frac{\gamma_{12}^2 + \gamma_{23}^2 - \gamma_{13}^2}{2\gamma_{12}\gamma_{23}}\right), \\
    \theta_3&= \pi-\arccos\left( \frac{\gamma_{13}^2 + \gamma_{23}^2 - \gamma_{12}^2}{2\gamma_{13}\gamma_{23}}\right). 
\end{align*}

Converting to degrees is obtained by multiplying with $180/\pi.$ In the above situation the exact angles can be computed and are given by $\theta_1 \approx 135.584^\circ, \theta_2 \approx 122.878^\circ, \theta_3 \approx 101.536^\circ.$

The temporal evolution of the volume fraction is visualized in Figure \ref{fig:lens} and shows the evolution from the spherical initial configuration of $\phi_1$ towards the expected lens shape. In the diffuse-interface representation the relation holds only approximately and the triple point corresponds to a whole region. We located one triple point at $(0.633, 0.248)$ and obtained the following errors in the angles:
\begin{align}
   |\theta_1-\theta_{1,\text{num}}|&=1.317^\circ, \quad
|\theta_2-\theta_{2,\text{num}}|=1.006^\circ, \quad
|\theta_3-\theta_{3,\text{num}}|=0.310^\circ, 
\end{align}
and in Young's law:
\begin{subequations}
\begin{align}
\frac{\gamma_{23}}{\sin(\theta_{1,\text{num}})}-\frac{\gamma_{13}}{\sin(\theta_{2,\text{num}})}=&~-0.0024, \\
\frac{\gamma_{13}}{\sin(\theta_{2,\text{num}})}-\frac{\gamma_{12}}{\sin(\theta_{3,\text{num}})}=&~~0.0007,\\
\frac{\gamma_{23}}{\sin(\theta_{1,\text{num}})}-\frac{\gamma_{12}}{\sin(\theta_{3,\text{num}})}=&~-0.00164.
\end{align}
\end{subequations}

\begin{figure}
 \centering
 \begin{tabular}{c@{}c@{}c@{}c@{}}
 \multicolumn{4}{c}{\includegraphics[trim={27cm 65cm 27cm 3cm},clip,width=0.7\linewidth]{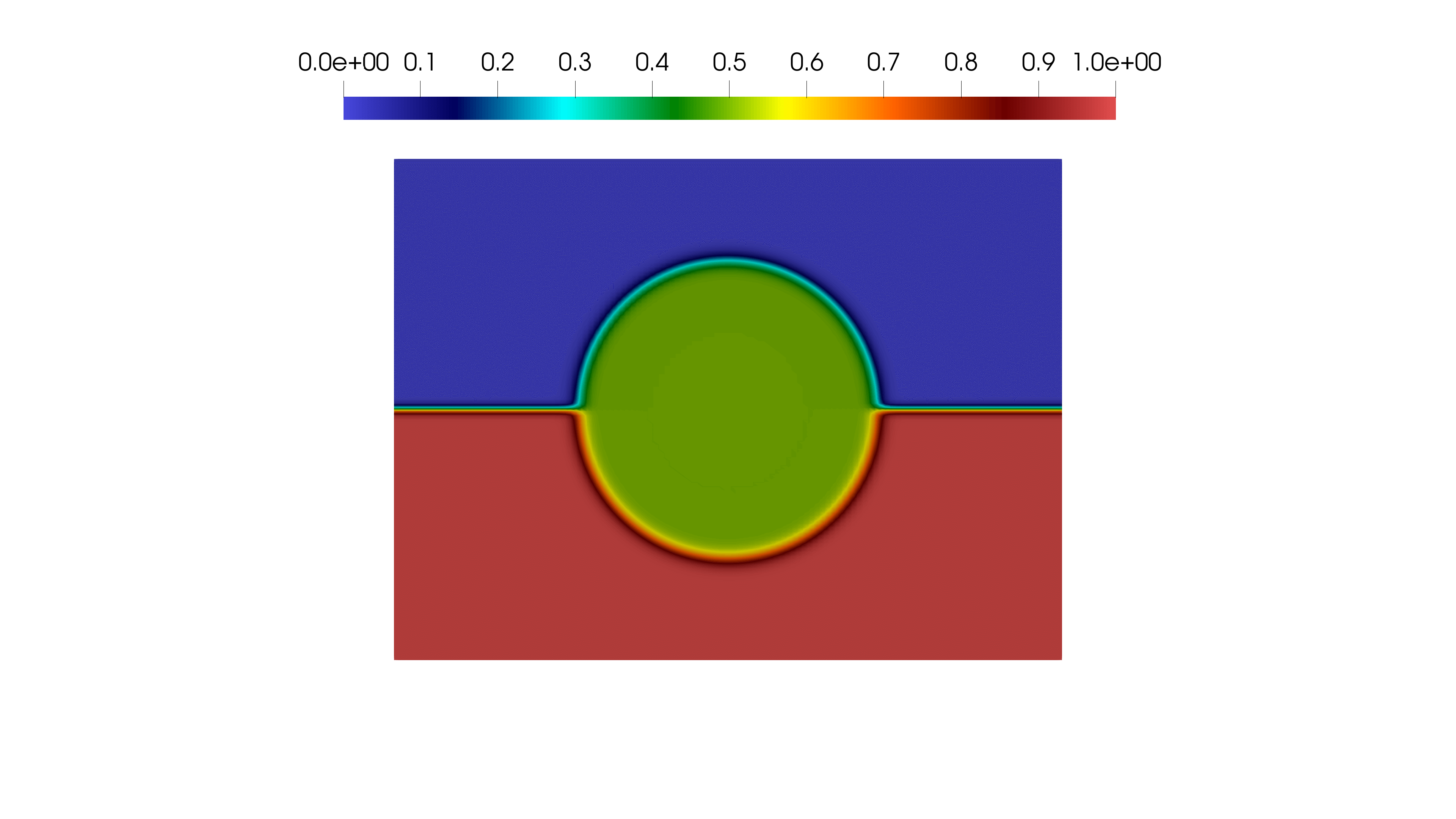}} \\
  \includegraphics[trim={35cm 15cm 35cm 15cm},clip,width=0.248\linewidth]{figures/surface_tension/lens/phi_combo.0000.png}
  &
  \includegraphics[trim={35cm 15cm 35cm 15cm},clip,width=0.248\linewidth]{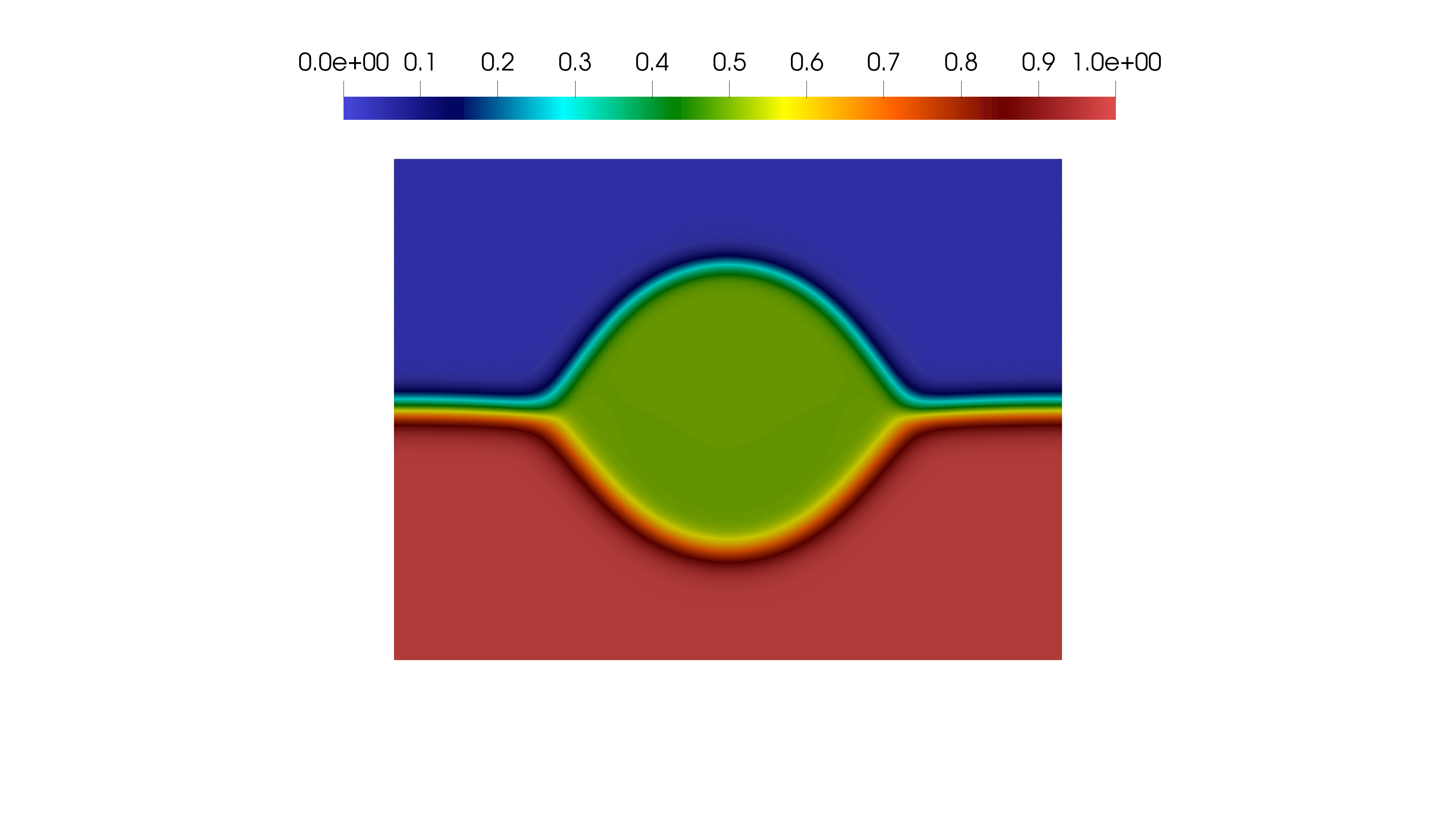} 
  &
  \includegraphics[trim={35cm 15cm 35cm 15cm},clip,width=0.248\linewidth]{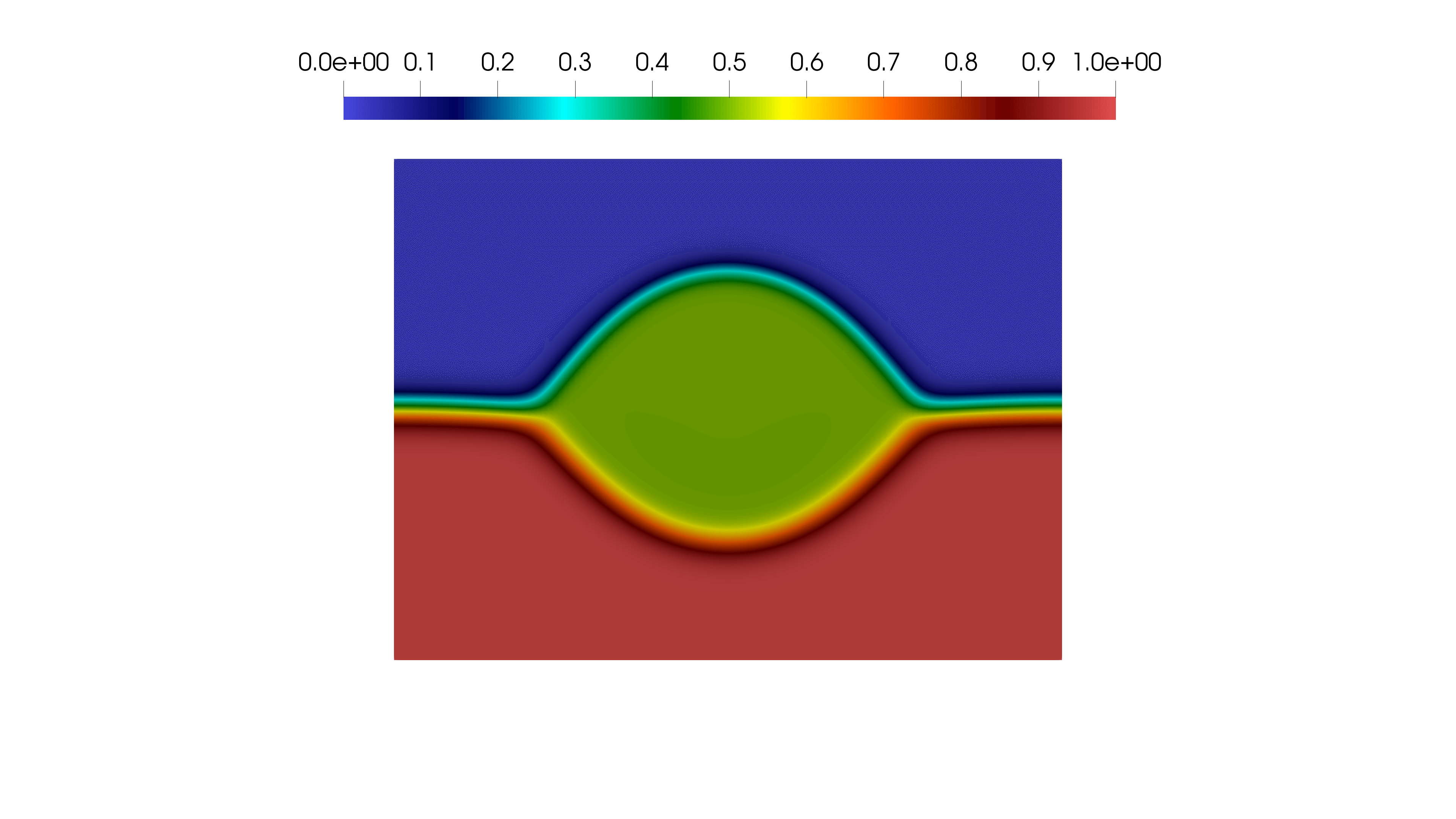} 
  &
  \includegraphics[trim={35cm 15cm 35cm 15cm},clip,width=0.248\linewidth]{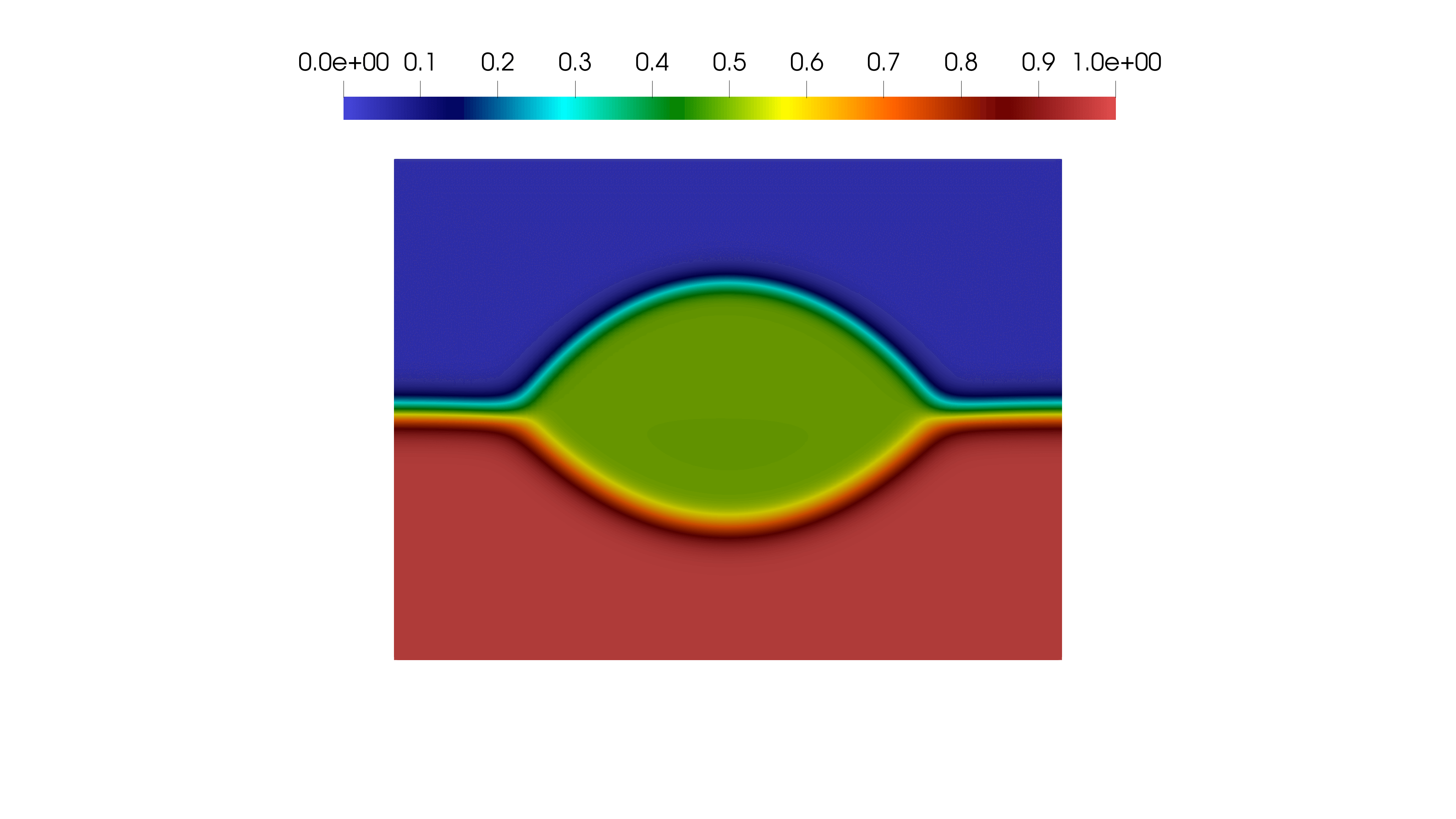}  \\
  \includegraphics[trim={35cm 15cm 35cm 15cm},clip,width=0.248\linewidth]{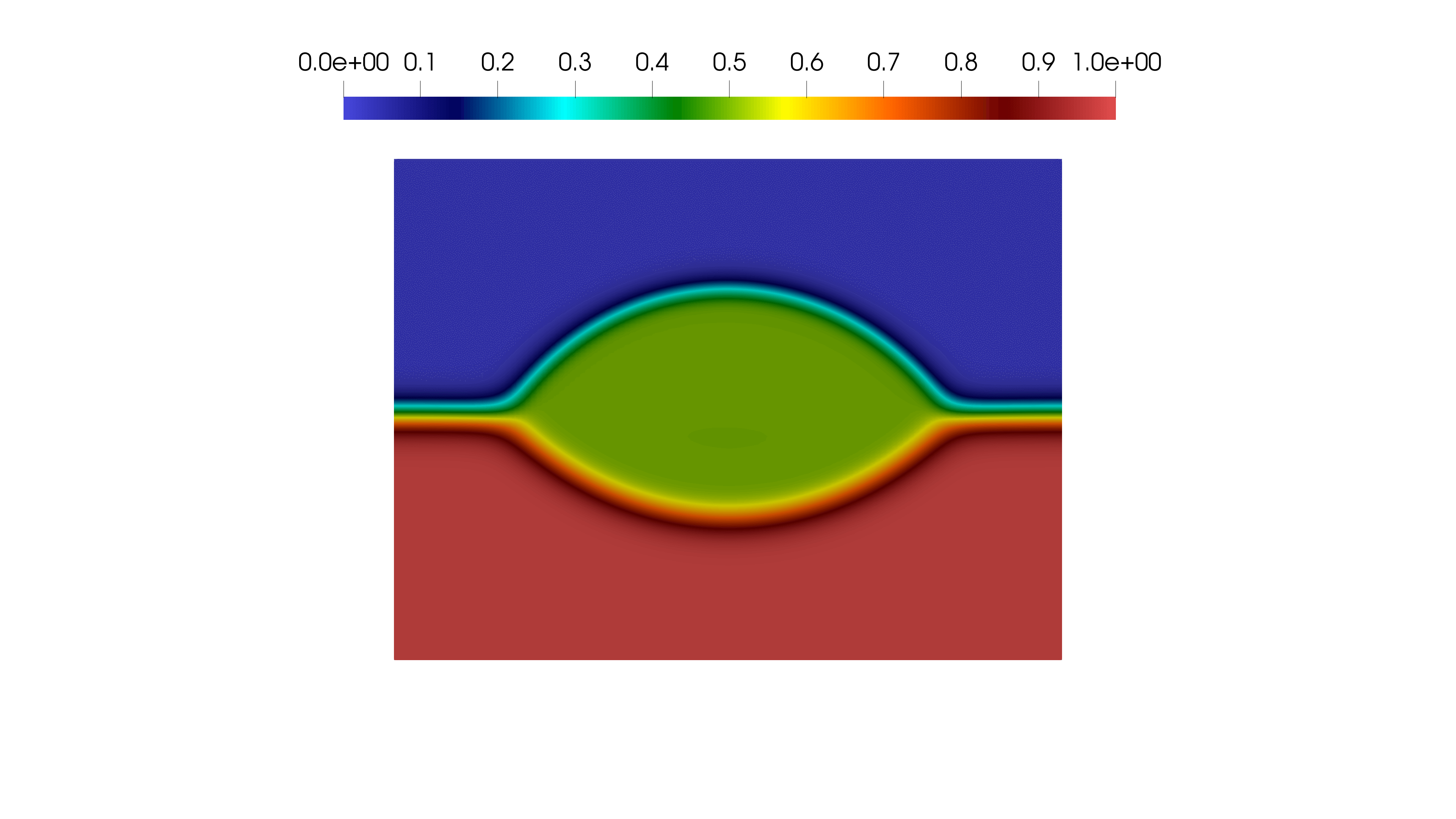}
  &
  \includegraphics[trim={35cm 15cm 35cm 15cm},clip,width=0.248\linewidth]{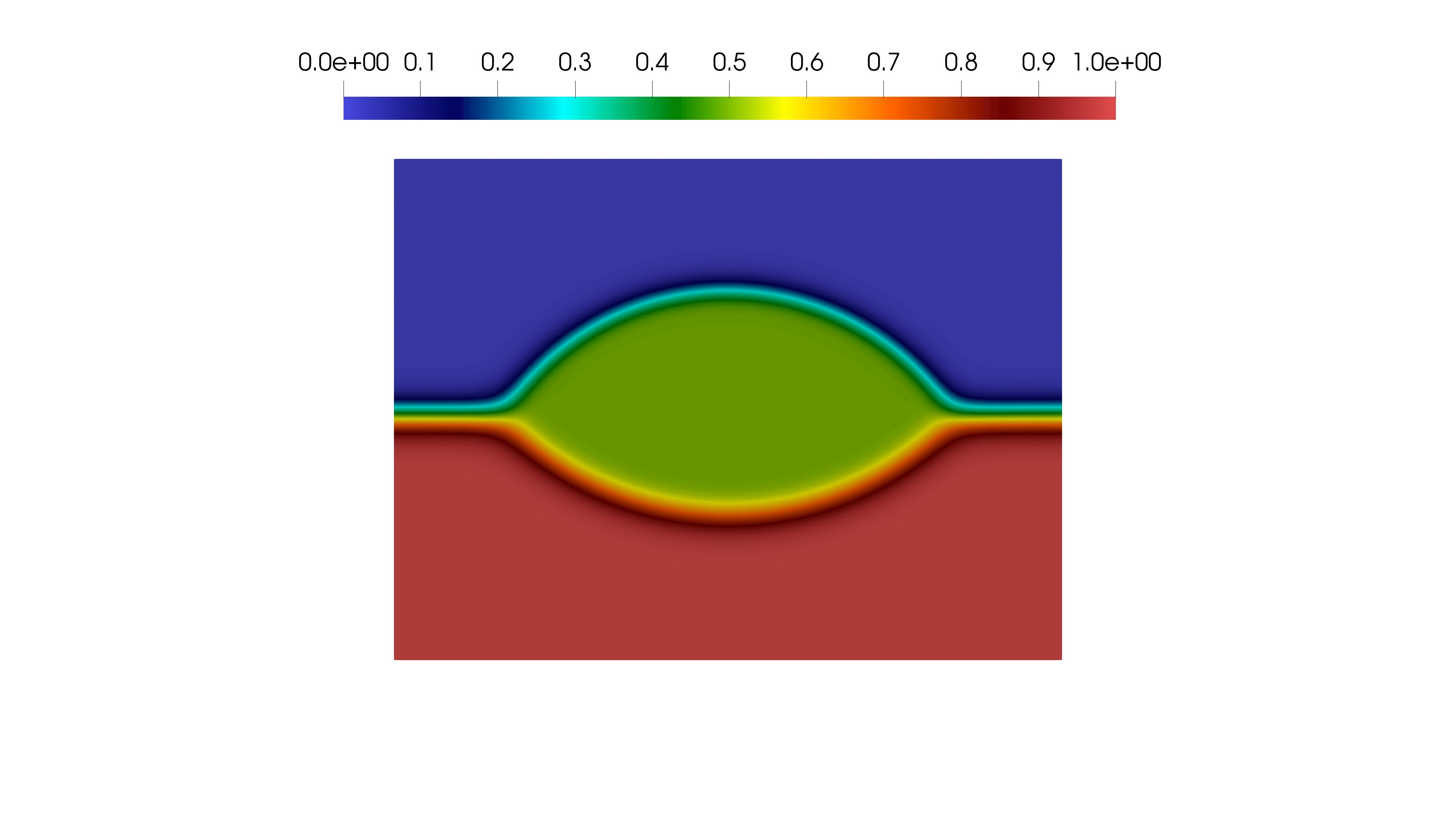} 
  &
  \includegraphics[trim={35cm 15cm 35cm 15cm},clip,width=0.248\linewidth]{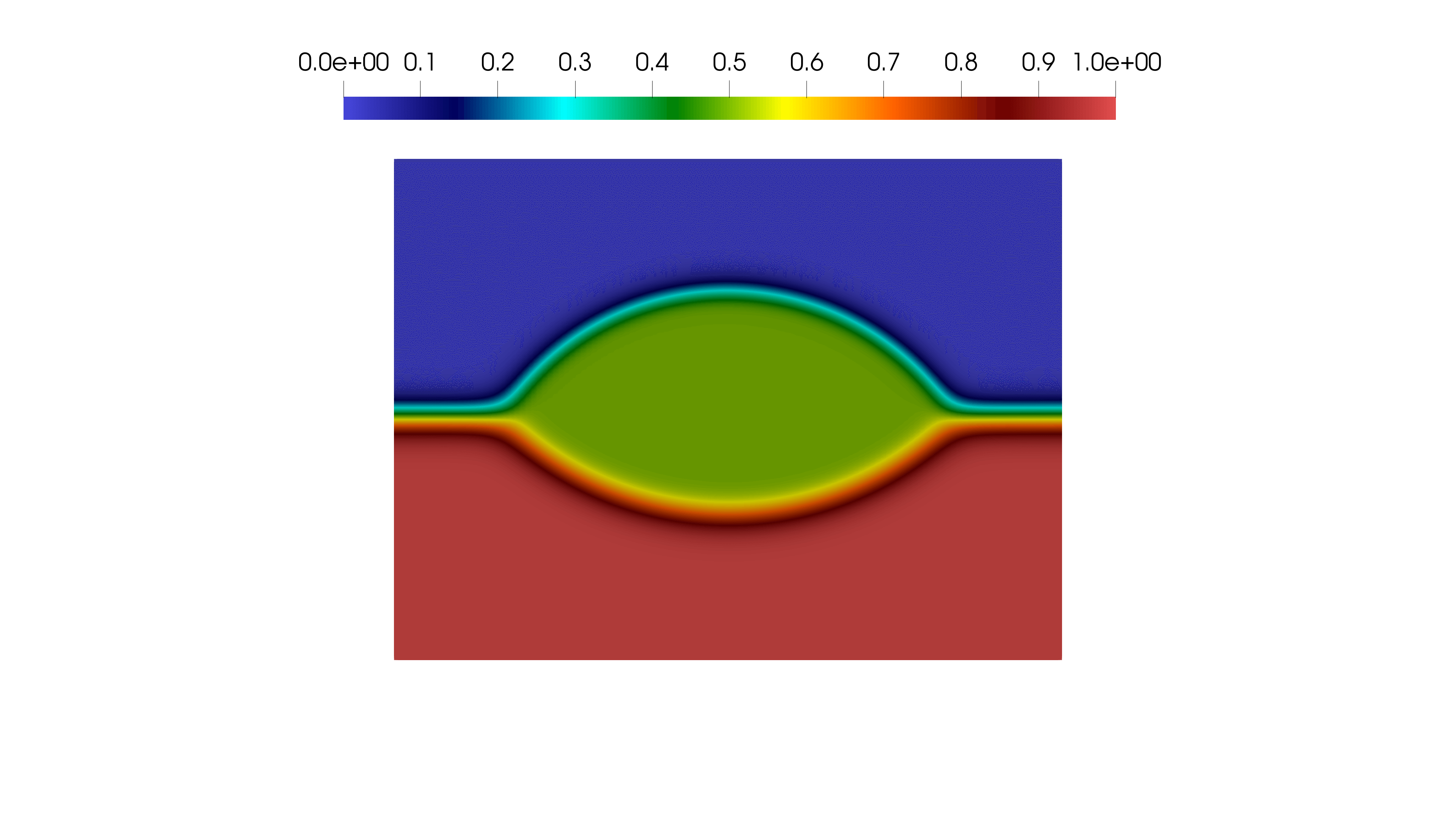} 
  &
  \includegraphics[trim={35cm 15cm 35cm 15cm},clip,width=0.248\linewidth]{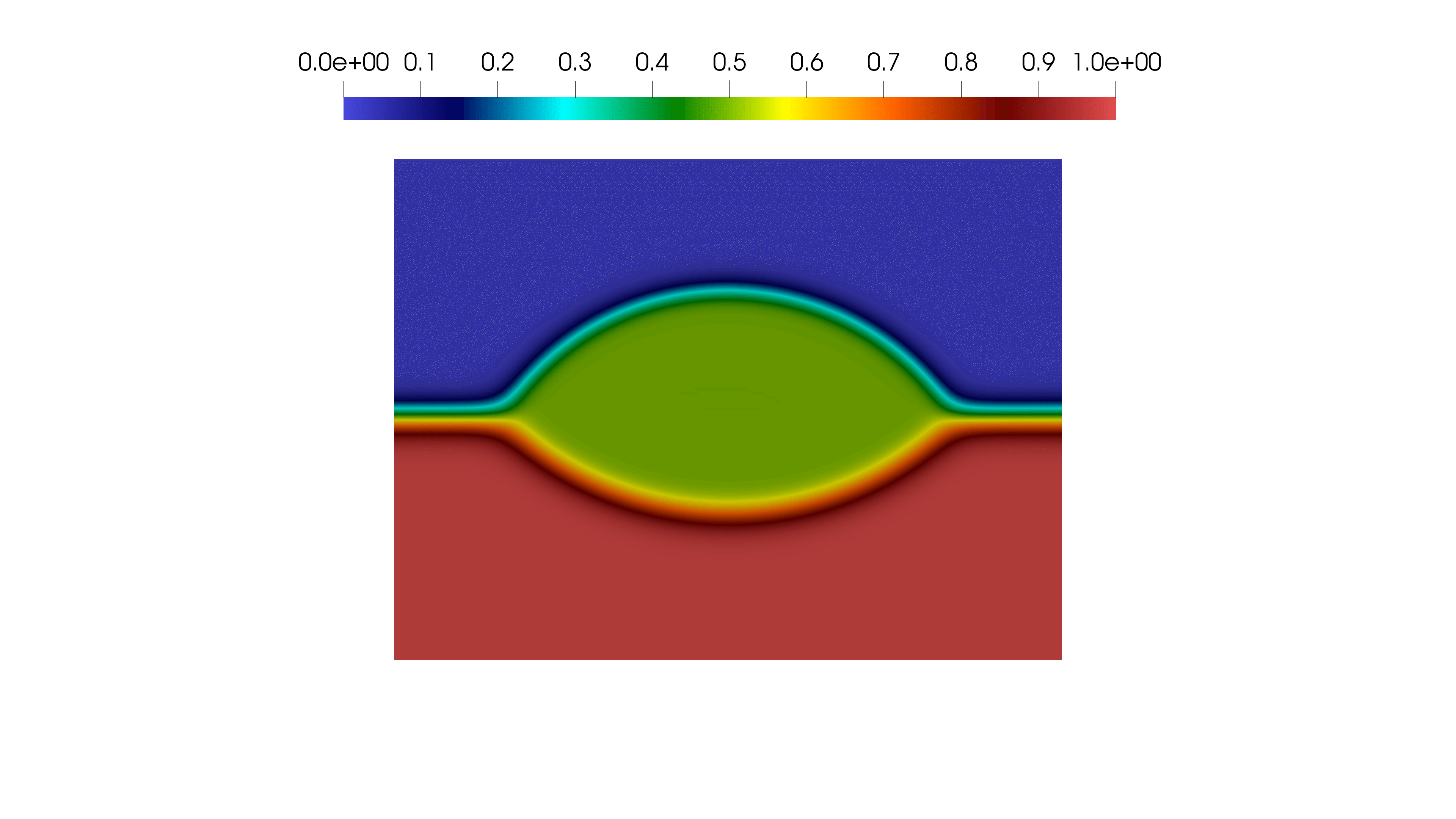} 
 \end{tabular}
 \caption{Temporal evolution of $\phi_2+\frac{1}{2}\phi_1$ for the times $t\in\{0,0.01,0.02,0.05,0.1,0.15,0.2,1\}$}
 \label{fig:lens}
\end{figure}

\subsection{Three-fluid rising bubble problem}\label{sec:results:rising_bubble}

We next consider a three-fluid benchmark in which a gas bubble rises through a vertically stratified configuration of two immiscible liquid layers. The test involves three distinct pairwise surface tensions, large density and viscosity contrasts, and a strong interaction between the bubble interface and the liquid--liquid interface.

Initially, the heavy liquid, denoted by phase $2$, occupies the lower part of the domain, while the lighter liquid, denoted by phase $3$, occupies the upper part. A spherical gas bubble, denoted by phase $1$, is placed in the lower liquid below the initially planar liquid--liquid interface.

The computation is carried out in an axisymmetric setting using an
axisymmetric version of the structure-preserving finite element method proposed in \cite{structureNphase2026}. The computational domain has size $2r\times 15r$ and is discretized by $64\times480$ finite elements. We choose the time-step size $\Delta t=0.5 h$ and the interface width $\varepsilon=h$, where $h$ denotes the mesh width. The initial configuration consists of a diffuse, nearly planar interface between the two liquid layers, phases $2$ and $3$, and a diffuse spherical bubble of radius $r$, phase $1$, placed in the lower liquid
phase below this interface. 

The target surface tensions are
\begin{align}
  \gamma_{12}=\gamma_{13}=0.07~\mathrm{N\,m^{-1}},\qquad
  \gamma_{23}=0.05~\mathrm{N\,m^{-1}}.
\end{align}
We use the algorithm provided in \cref{subsec:kappa_reconstruction} to calibrate the capillary matrix and the interface widths. This provides the target surface tensions and interface widths:
\begin{subequations}
    \begin{align}    \bar{\boldsymbol{\kappa}} =&~10^{-3}\begin{pmatrix}
4.1361833328 & -6.5504823549  & -1.7218843107 \\
-6.5504823549  & 6.5504823549
 & -6.5504823549 \\
-1.7218843107  &  -6.5504823549  &  4.1361833328 
\end{pmatrix}\\
\bar{\varepsilon}_{12} = &~0.0819004278, \quad \bar{\varepsilon}_{13} = 0.0581144728, \quad \bar{\varepsilon}_{23} = 0.0819004278, \\
\varepsilon_0 =&~  2.15092719 \times 10^{-3}.
\end{align}
\end{subequations}
The densities and viscosities are
\begin{center}
\begin{tabular}{lcc}
\hline
phase & density $\rho$ [$\mathrm{kg\,m^{-3}}$] & viscosity $\eta$ [$\mathrm{Pa\,s}$] \\
\hline
bubble ($\phi_1$) & $1$ & $10^{-4}$ \\
heavy liquid ($\phi_2$) & $1200$ & $0.15$ \\
light liquid ($\phi_3$) & $1000$ & $0.10$ \\
\hline
\end{tabular}
\end{center}
Gravity acts in the vertical direction, $\mathbf g=-g\mathbf e_y$, with
$g=9.81~\mathrm{m\,s^{-2}}$. No-penetration boundary conditions are imposed on the velocity at the outer boundary, and homogeneous Neumann conditions are used for the phase fields and chemical potentials.

The qualitative outcome of this benchmark depends on the bubble size. Bubbles below a critical size may remain trapped at the liquid--liquid interface, whereas larger bubbles are expected to cross the capillary barrier and enter the upper liquid layer. Following \cite{greene1991bubble,greene1988onset}, this transition is estimated by the critical bubble volume
\begin{align}\label{eq:Vp_criterion}
  V_p =
  \left(
    2\pi\left(\frac{3}{4\pi}\right)^{1/3}
    \frac{\gamma_{23}}{(\rho_3-\rho_1)g}
  \right)^{3/2}
  \approx 8.87\times 10^{-8}\ \mathrm{m^3}.
\end{align}
The corresponding critical spherical radius is
\begin{align}
    r_p
    =
    \left(\frac{3V_p}{4\pi}\right)^{1/3}
    \approx 2.76\times 10^{-3}~\mathrm{m}.
\end{align}
In the present computation we consider a bubble of radius $r = 4.0\times 10^{-3}~\mathrm{m}$, so that $r>r_p$. The criterion therefore predicts penetration of the liquid--liquid interface. Additional bubble radii for the same benchmark were considered in \cite{boyer2010cahn,mixtureaware2026}.

\begin{remark}[Scaling estimate for the penetration criterion]
\label{rem:penetration-criterion}
The penetration criterion follows from balancing buoyancy against the capillary
force resisting deformation of the liquid--liquid interface. In the
axisymmetric setting, the estimate is three-dimensional. For a bubble of volume
$V=\frac43\pi R^3$, the buoyancy force is $ F_b=(\rho_3-\rho_1)gV$, while the capillary force is estimated by $F_\gamma=2\pi R\gamma_{23}$.
The critical volume is obtained from $F_b=F_\gamma$,
which gives
\begin{align}
    V_p =
    \left(
    2\pi\left(\frac{3}{4\pi}\right)^{1/3}
    \frac{\gamma_{23}}{(\rho_3-\rho_1)g}
    \right)^{3/2},
    \qquad
    r_p
    =
    \left(\frac{3V_p}{4\pi}\right)^{1/3}
    =
    \sqrt{\frac{3\gamma_{23}}{2(\rho_3-\rho_1)g}} .
\end{align}
\end{remark}

The computed evolution is shown in \cref{fig:penetration-case}. Consistent with the estimate \eqref{eq:Vp_criterion}, the bubble radius is above the critical value and the capillary barrier at the liquid--liquid interface is overcome. The bubble enters the upper liquid layer and entrains a thin column of the lower liquid during the penetration process. This qualitative behavior agrees with the observations reported in \cite{boyer2010cahn} and with the computations in
\cite{mixtureaware2026}. 
\begin{figure}[h]
\captionsetup[subfigure]{justification=centering}
\begin{subfigure}{0.245\textwidth}
\centering
\includegraphics[scale=0.17]{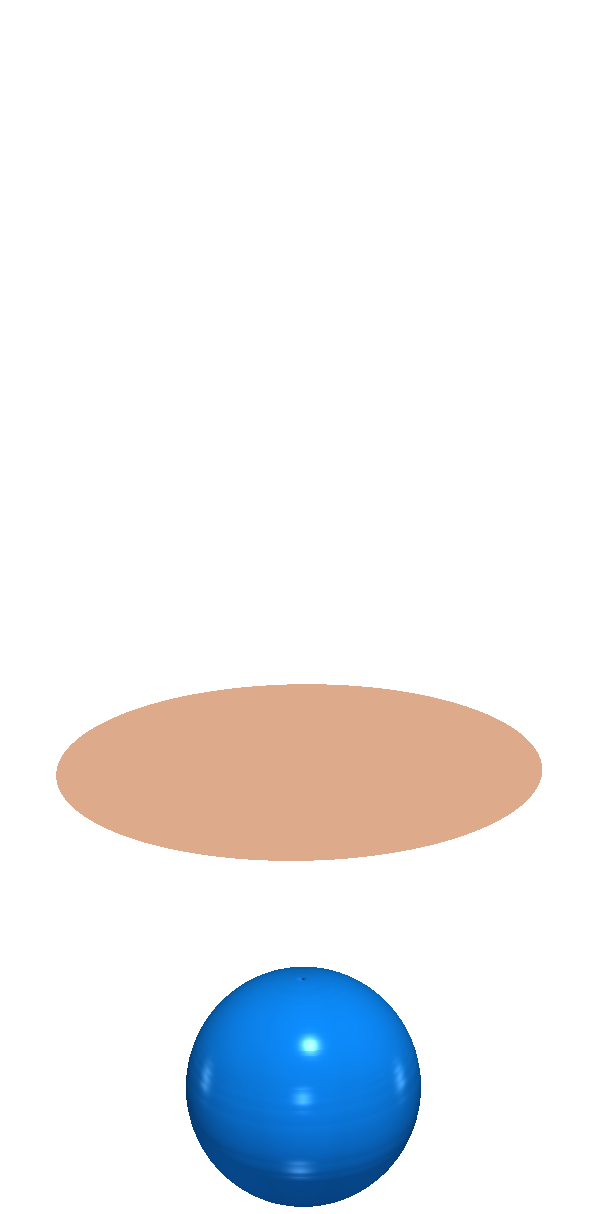}
\caption{$0.00~{\rm s}$}
\end{subfigure}
\begin{subfigure}{0.245\textwidth}
\centering
\includegraphics[scale=0.17]{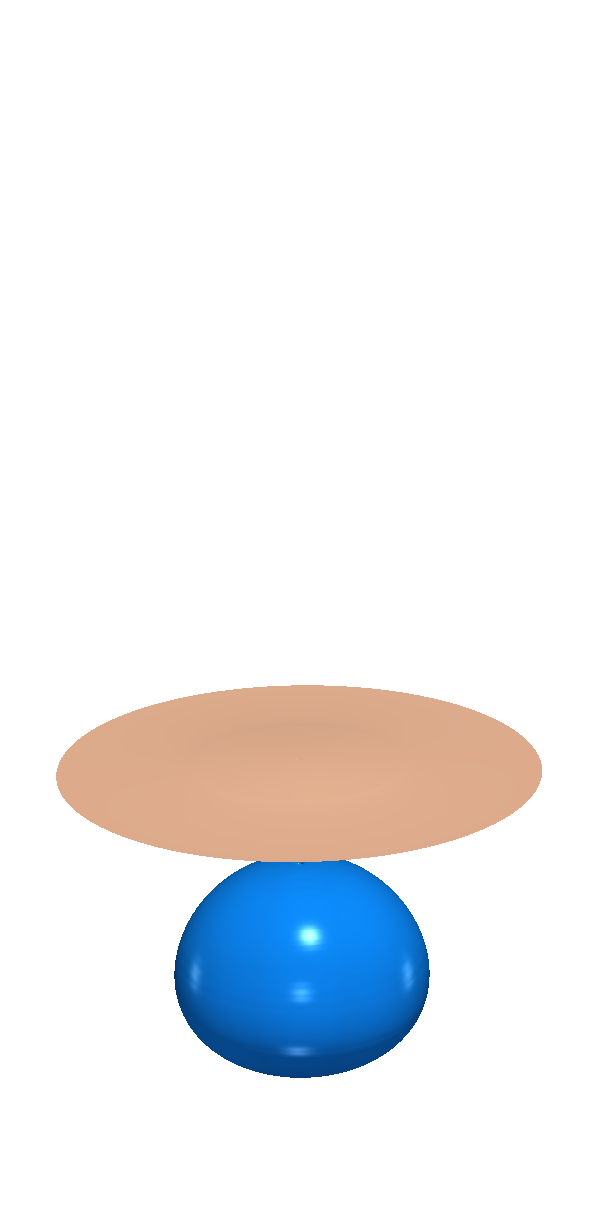}
\caption{$0.05~{\rm s}$}
\end{subfigure}
\begin{subfigure}{0.245\textwidth}
\centering
\includegraphics[scale=0.17]{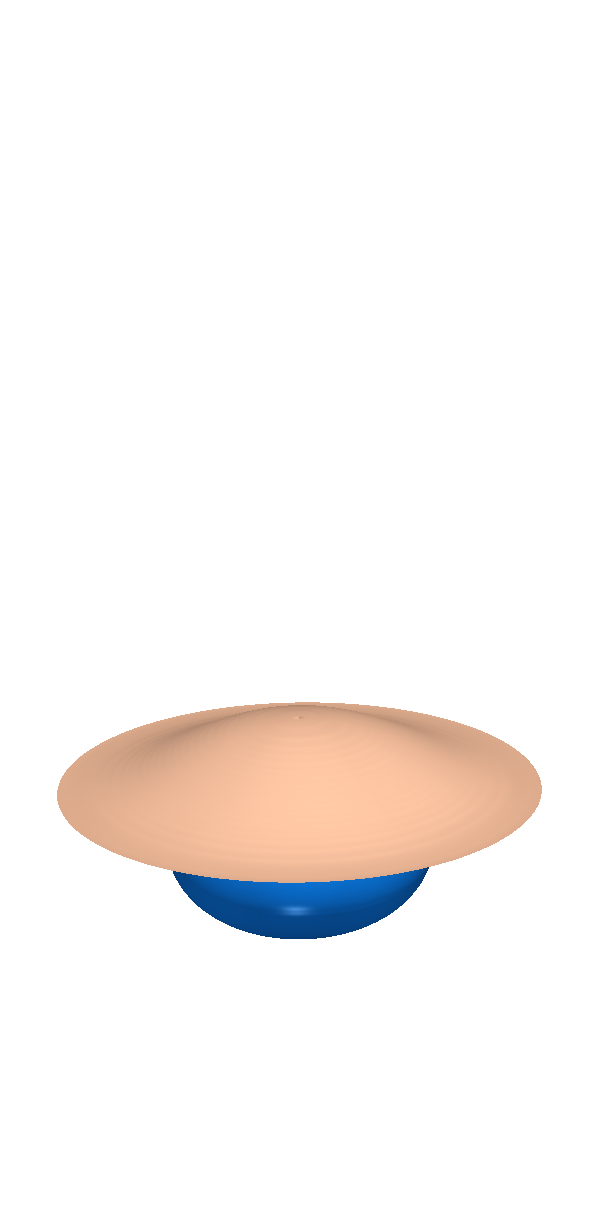}
\caption{$0.10~{\rm s}$}
\end{subfigure}
\begin{subfigure}{0.245\textwidth}
\centering
\includegraphics[scale=0.17]{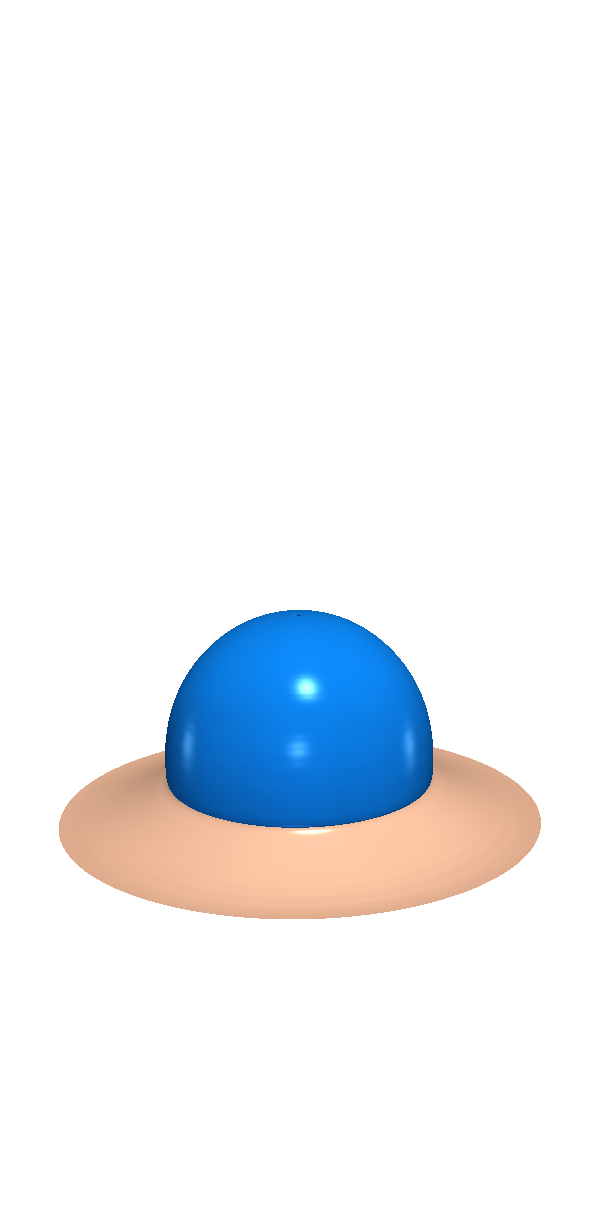}
\caption{$0.15~{\rm s}$}
\end{subfigure}
\begin{subfigure}{0.245\textwidth}
\centering
\includegraphics[scale=0.17]{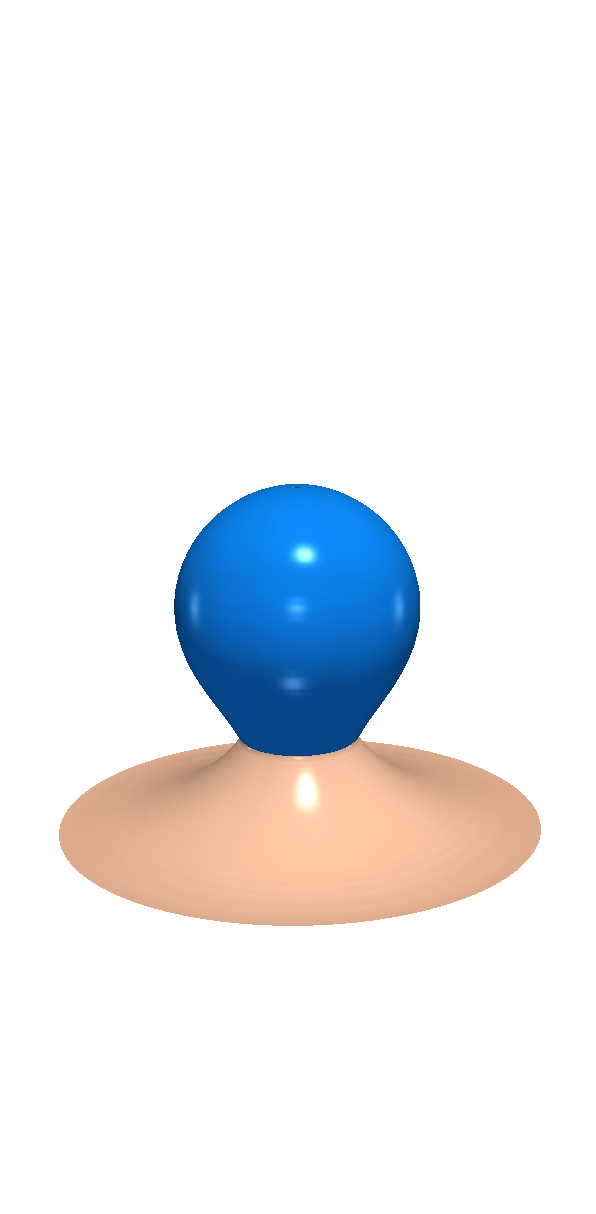}
\caption{$0.20~{\rm s}$}
\end{subfigure}
\begin{subfigure}{0.245\textwidth}
\centering
\includegraphics[scale=0.17]{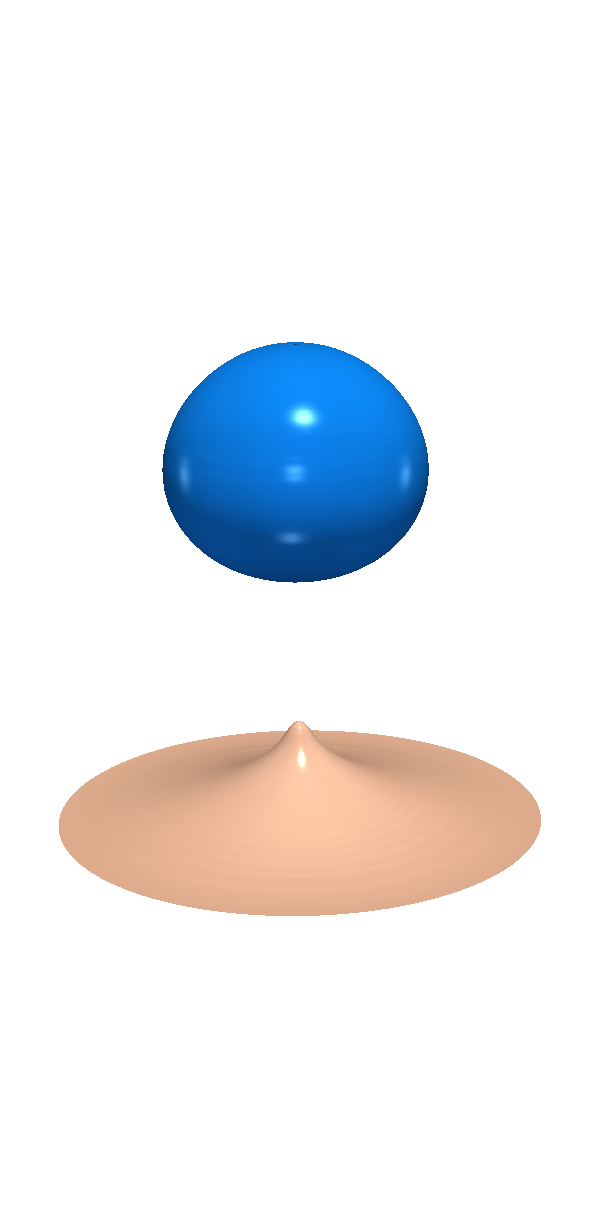}
\caption{$0.25~{\rm s}$}
\end{subfigure}
\begin{subfigure}{0.245\textwidth}
\centering
\includegraphics[scale=0.17]{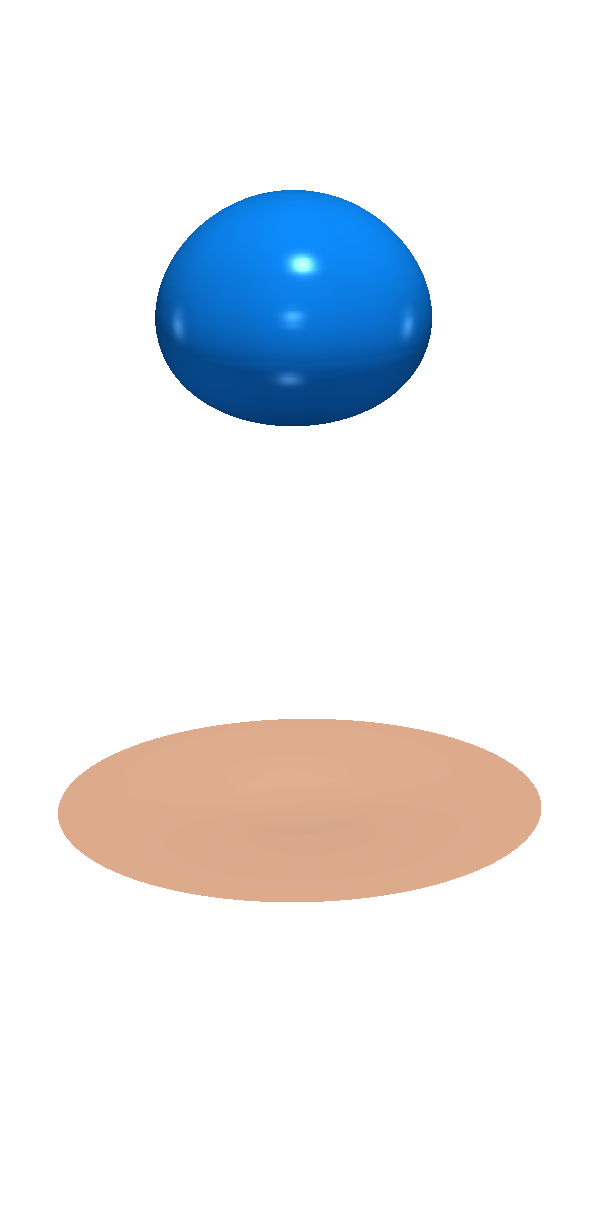}
\caption{$0.30~{\rm s}$}
\end{subfigure}
\begin{subfigure}{0.245\textwidth}
\centering
\includegraphics[scale=0.17]{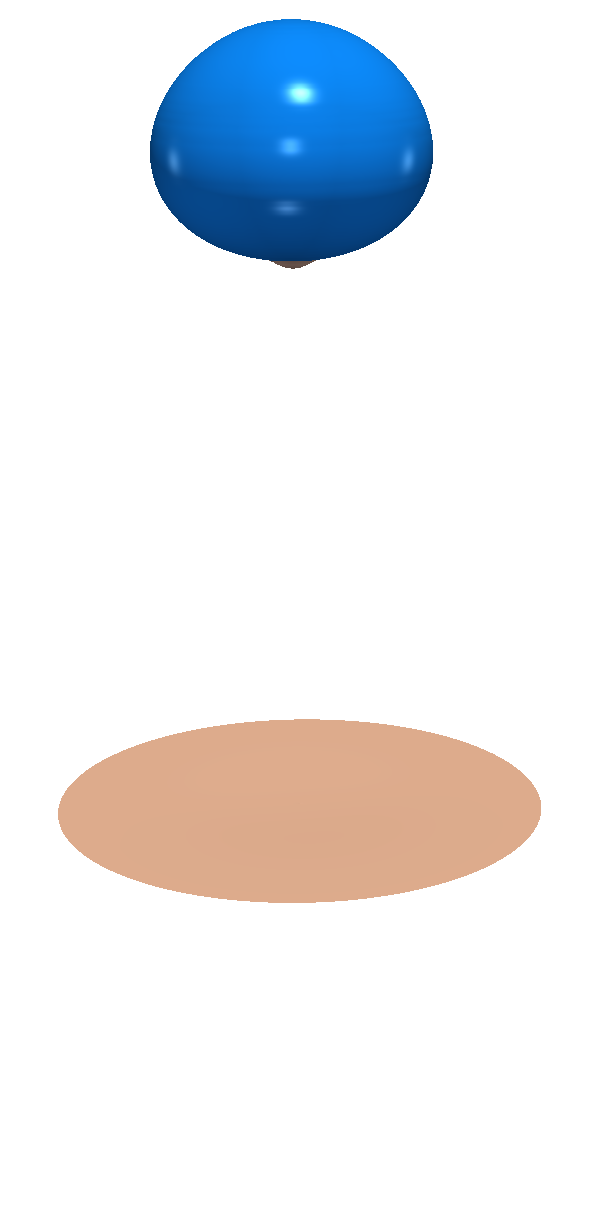}
\caption{$0.35~{\rm s}$}
\end{subfigure}
\caption{Axisymmetric rising-bubble benchmark in two stratified liquid layers. Iso-contours of the phase distribution. The gas bubble rises from the lower liquid, crosses the liquid--liquid interface, and
entrains a thin column of the lower liquid during penetration.}
\label{fig:penetration-case}
\end{figure}

\section{Conclusions}\label{sec: Conclusions}

We have developed a surface-tension calibration procedure for $N$-phase diffuse-interface free energies. The method determines a symmetric capillary matrix from prescribed pairwise surface tensions by using one-dimensional
equilibrium profiles in the Gibbs simplex. A straight-line approximation provides an initial estimate, after which the capillary matrix is iteratively updated using equilibrium profiles computed from the Euler--Lagrange equations.
We also established precise face-confinement properties for the two bulk free energies considered in this work. For the regularized mixture-aware free energy used here, which was proposed in \cite{mixtureaware2026}, minimizing profiles are
not exactly confined to Gibbs-simplex faces. In contrast, for the Boyer free energy we proved exact face confinement of the minimizing profiles. These results put these two free-energy choices on different geometric footings and show that calibration for general multiphase free energies cannot rely only on face-restricted binary profiles. The numerical calibration tests for two-, three-, and four-phase systems show that this reconstruction accurately recovers the prescribed surface tensions, whereas the straight-line approximation alone is generally insufficient.

After calibrating the surface tensions, we introduced a definition of pairwise interface width based on the normal-coordinate equilibrium profiles. This makes
it possible to compare the widths of different pairwise interfaces and to define a representative numerical interface width. By introducing a global thickness
parameter, the diffuse interfaces can be made resolvable on a chosen mesh while leaving the calibrated surface tensions unchanged.

The calibrated free-energy closures were then used in application tests for three-phase Navier--Stokes--Cahn--Hilliard models. The lens test demonstrates that the calibrated surface tensions lead to the expected equilibrium interfacial configuration in a multiphase Cahn--Hilliard setting. The rising-bubble simulations show that the same calibrated closure can be used in density-contrast NSCH flows, where capillary effects, hydrodynamics, and multiple interfaces interact.

Overall, the proposed procedure provides a practical route from prescribed pairwise surface tensions and mesh-resolution requirements to a calibrated multiphase free-energy closure suitable for computation. Future work includes the treatment of total spreading regimes, where a prescribed pairwise surface tension may be larger than the sum of two others and a smooth direct minimizing profile may fail to exist. Extending the calibration to such cases would require accounting for composite interfaces and possible wetting layers. A second direction for future research concerns the determination of coexistence points, which seems to be a nontrivial problem in general.

\appendix

\section{Reduction for coinciding coexistence points}\label{sec:coinciding-coexistence-points}
In this section, we sketch for $N=3$ what happens when two coexistence points coincide and how to recover the capillary matrix for given surface tensions $(\gamma_{12},\gamma_{13},\gamma_{23})$. We assume that the coexistence points have the form
\begin{align}
    \mathbf b^{(1)}=\mathbf b^{(2)}=(p,p,q),
    \qquad
    \mathbf b^{(3)}=(r,r,s),
    \qquad
    p\neq r,\quad q\neq s,
\end{align}
with $2p+q=1$ and $2r+s=1$. The first two phases have identical coexistence states and cannot be separated by an independent diffuse interface. To exploit this degeneracy, we introduce the variables
\begin{align}\label{eq:omega}
    \omega_1=\phi_1-\phi_2,
    \qquad
    \omega_2=\frac{\phi_1+\phi_2}{2},
    \qquad
    \omega_3=\phi_3 .
\end{align}
In these variables the coexistence points become
\begin{align}
    \widetilde{\mathbf b}^{(1)}=\widetilde{\mathbf b}^{(2)}=(0,p,q),
    \qquad
    \widetilde{\mathbf b}^{(3)}=(0,r,s).
\end{align}
The reduced problem is therefore obtained by setting $\omega_1=0$, equivalently
$\phi_1=\phi_2$, and considering profiles in the two variables $\omega=(\omega_2,\omega_3)$. These profiles connect the reduced coexistence points
\begin{align}
    \widehat{\mathbf b}^{(12)}=(p,q),
    \qquad
    \widehat{\mathbf b}^{(3)}=(r,s).
\end{align}
Since $\mathbf b^{(1)}=\mathbf b^{(2)}$, the $1$--$2$ interface is not represented as an
independent material interface. Hence $\gamma_{12}$ does not enter the reduced
calibration. The effective surface tension between the merged phase pair
$(1,2)$ and phase $3$ is taken as
\begin{align}
    \widehat{\gamma}_{(12),3}
    =
    \frac{1}{2}\left(\gamma_{13}+\gamma_{23}\right).
\end{align}
The reduced two-phase calibration gives a matrix
$\boldsymbol{\kappa}_{\rm red}\in\mathbb R^{2\times 2}$ which, in the present case, has the form
\begin{align}\label{eq:kappa_red}
    \boldsymbol{\kappa}_{\rm red}
    =
    \begin{pmatrix}
         a    & -a \\
        -a    &  a
    \end{pmatrix}.
\end{align}
We then reconstruct a full $3\times 3$ capillary matrix by requiring
\begin{align}
    \begin{pmatrix}
    \nabla\omega_2 \\
    \nabla\omega_3
    \end{pmatrix}^{\!\top}
    \boldsymbol{\kappa}_{\rm red}
    \begin{pmatrix}
    \nabla\omega_2 \\
    \nabla\omega_3
    \end{pmatrix}
    =
    \begin{pmatrix}
    \nabla\phi_1 \\
    \nabla\phi_2 \\
    \nabla\phi_3
    \end{pmatrix}^{\!\top}
    \boldsymbol{\kappa}_{\rm full}
    \begin{pmatrix}
    \nabla\phi_1 \\
    \nabla\phi_2 \\
    \nabla\phi_3
    \end{pmatrix}.
\end{align}
Substitution gives
\begin{align}
    \boldsymbol{\kappa}_{\rm full}
    =
    \begin{pmatrix}
        \frac{a}{4}      & \frac{a}{4} & -\frac{a}{2} \\
        \frac{a}{4}      & \frac{a}{4} & -\frac{a}{2} \\
        -\frac{a}{2}     & -\frac{a}{2}& a
    \end{pmatrix}.
\end{align}

We illustrate the reduction on a three-phase test case in which the first two coexistence points coincide, i.e. taking $\chi_{\mA\mB}=2.0753225$ if $\mA\mB\in\{11,12,21,22,33\}$ and zero else. The reduced two-phase calibration is first carried
out for the transition between the merged pair $(1,2)$ and phase $3$ using $\gamma_{13}=0.05, \gamma_{23}=0.06$. The resulting reduced capillary matrix is then lifted to the original three-variable representation as described above. This produces the straight-line initialization and equilibrium-profile reconstruction:
\begin{align*}
\boldsymbol{\kappa}_0 = 10^{-6}\begin{pmatrix}
  3.783 & 3.783 & -7.567 \\
  3.783 & 3.783 & -7.567 \\
  -7.567 & -7.567 & 15.135
\end{pmatrix}
,\quad  \boldsymbol{\kappa}_{\text{eq}}=10^{-6}\begin{pmatrix}
  7.656  & 7.656  & -15.312 \\
  7.656  & 7.656  & -15.312 \\
  -15.312 & -15.312 & 30.625
\end{pmatrix}.
\end{align*}

The corresponding equilibrium profile and its path in the Gibbs triangle are shown in \cref{fig:n3red_profiles}.
\begin{figure}
 \centering

 \begin{minipage}{0.49\linewidth}
   \centering
   \includegraphics[width=\linewidth]{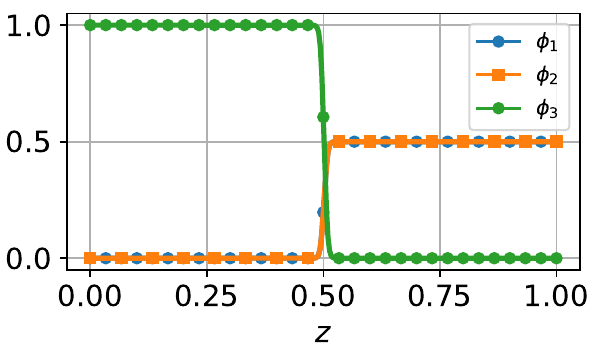}

   \vspace{0.02\linewidth}

   \includegraphics[width=\linewidth]{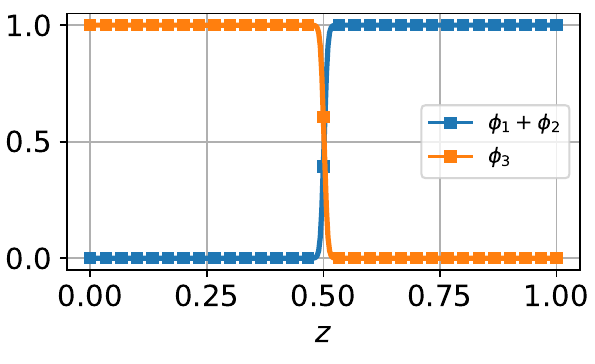}
 \end{minipage}
 \hfill
 \begin{minipage}{0.48\linewidth}
   \centering
   \includegraphics[width=\linewidth]{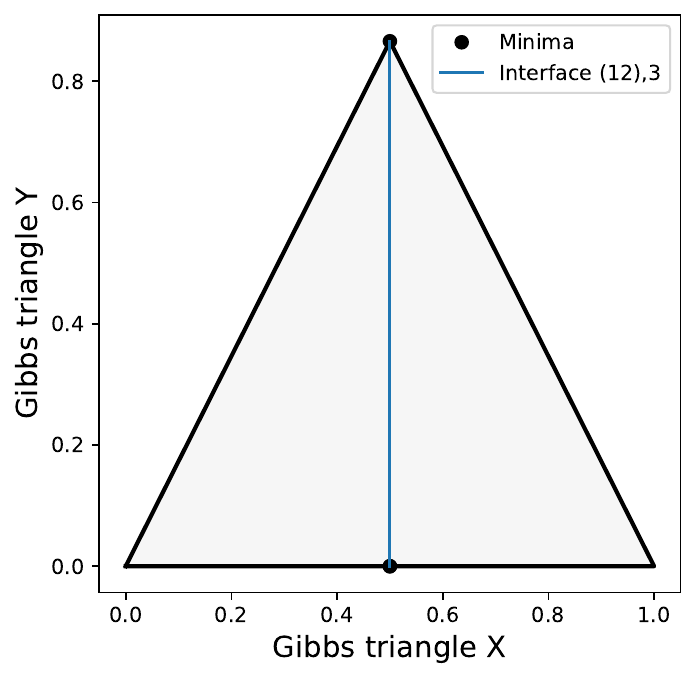}
 \end{minipage}
 \caption{Equilibrium profile for coinciding coexistence points and the associated visualization in the Gibbs triangle.}
\label{fig:n3red_profiles}
\end{figure}
The proposed construction can directly be extended to $N$ phases with two coinciding coexistence points.

\section{Computation of equilibrium profiles}\label{app:equicomp}
For each interface $(\mA,\mB)$ we compute a 1D equilibrium composition profile
$\boldsymbol{\phi}(z)=(\phi_1(z),\dots,\phi_N(z))$ on $\Omega=(0,L_x)$
by enforcing stationarity of the constrained phase--field energy.
The unknowns are $(\boldsymbol{\phi},\lambda)$, where $\lambda(z)$ is a Lagrange multiplier
for the simplex constraint $\sum_{\mC=1}^N \phi_\mC = 1$.

With a symmetric capillary matrix $\boldsymbol{\kappa}\in\mathbb{R}^{N\times N}$ and bulk potential
$\Psi_0(\boldsymbol{\phi})$, the functional is
\begin{align}
\mathcal{E}(\boldsymbol{\phi},\lambda)
= \int_0^{L_x}
\Bigg[
\frac{1}{2}\sum_{m=1}^N\sum_{n=1}^N \kappa_{mn}\,\phi_m'(z)\,\phi_n'(z)
+ \Psi_0(\boldsymbol{\phi}(z))
+ \lambda(z)\Big(\sum_{k=1}^N \phi_k(z)-1\Big)
\Bigg]\,dz.
\label{eq:energy}
\end{align}
Stationarity of~\eqref{eq:energy} yields the coupled system of Euler--Lagrange equations:
\begin{subequations}\label{eq: EL5}
\begin{align}
-\frac{d}{dz}\!\left(\sum_{n=1}^N \kappa_{mn}\,\phi_n'(z)\right)
+ \frac{\partial \Psi_0}{\partial \phi_m}\big(\boldsymbol{\phi}(z)\big)
+ \lambda(z) &= 0, \qquad m=1,\dots,N, \label{eq:EL_phi_1d}\\
\sum_{k=1}^N \phi_k(z)-1 &= 0. \label{eq:EL_constraint_1d}
\end{align}
\end{subequations}
In the code, the equilibrium profile is obtained by solving the first-order optimality conditions of the constrained 1D phase-field energy for the unknowns $U=(\boldsymbol{\phi}_h,\lambda_h)\in (H^1(\Omega))^N\times H^1(\Omega)$ on $\Omega=(0,L_x)$, where $\lambda$ enforces the simplex constraint $\sum_{k=1}^N\phi_k=1$. With a symmetric capillary matrix $\boldsymbol{\kappa}$ and bulk potential $\Psi(\boldsymbol{\phi})$, the weak residual is defined by
\begin{align}
\mathcal{R}(U;W)
:= \int_0^{L_x}\!\Bigg[&
\sum_{m,n=1}^N \kappa_{mn}\,\phi'_{n,h}\,v'_{m,h}
+ \sum_{m=1}^N \frac{\partial \Psi_0}{\partial \phi_m}(\boldsymbol{\phi}_h)\,v_{m,h}\nn\\
& + \lambda_h \sum_{m=1}^N v_{m,h}
+ \mu_h\Big(\sum_{k=1}^N \phi_{k,h} - 1\Big)
\Bigg]\,dz,
\end{align}
for all test functions $W=(\boldsymbol{v}_h,\mu_h)$. The equilibrium profile is the solution of $\mathcal{R}(U;W)=0$ for all $W$, i.e. the Euler--Lagrange equations
\eqref{eq: EL5}. Since no Dirichlet boundary conditions are imposed, the variational formulation implies the natural conditions $\sum_n K_{mn}\phi_n'(0)=\sum_n \kappa_{mn}\phi_n'(L_x)=0$ for all $m$. Newton's method is applied to this nonlinear system: given $U^{(k)}$, the increment $\delta U^{(k)}$ is computed from the linearization $D_U\mathcal{R}(U^{(k)};W)[\delta U^{(k)}]=-\mathcal{R}(U^{(k)};W)$ for all $W$, and the iterate is updated by $U^{(k+1)}=U^{(k)}+\alpha^{(k)}\delta U^{(k)}$ with a line-search factor $\alpha^{(k)}\in(0,1]$ until the tolerance is reached.

\section{N-phase Cahn--Hilliard model}\label{appendix: CH model}

We consider the $N$-phase Cahn--Hilliard system in spatial domain $\Omega \subset \mathbb{R}^d$ (of dimension $d=2,3$), boundary $\partial\Omega$ and unit outward normal $\mathbf{n}$. The state variables are the volume fraction variables $\phi_\mA:\Omega \rightarrow \mathbb{R}$, $\mA=1,\dots,N$, subject to the initial conditions
\begin{align}
\phi_\mA(\mathbf{x},0)=\phi_{\mA,0}(\mathbf{x}).
\end{align}
The Helmholtz free energy $\Psi$ associated with the system belongs to the constitutive class
\begin{align}
    \Psi=\Psi\left(\left\{\phi_\mA\right\},\left\{\nabla\phi_\mA\right\}\right),
\end{align}
and the associated chemical potentials are
\begin{align}\label{eq:CH chem pot}
    \mu_\mA=\dfrac{\partial \Psi}{\partial \phi_\mA}
    -{\rm div}\dfrac{\partial \Psi}{\partial \nabla \phi_\mA}.
\end{align}
The model consists of $N$ phase balance equations of Cahn--Hilliard type:
\begin{align}\label{eq:CH mass}
    \partial_t \phi_\mA + {\rm div}\,\bJ_\mA = 0,
    \qquad \mA=1,\dots,N,
\end{align}
where the diffusive fluxes are taken as
\begin{align}\label{eq:CH flux}
\bJ_\mA = -\sum_\mB \mathbf{M}_{\mA\mB}\nabla \mu_\mB,
\end{align}
with $\mathbf{M}_{\mA\mB}=\mathbf{M}_{\mB\mA}$ a symmetric degenerate mobility matrix. 

Next, the absence of void spaces implies that the volume fractions are subject to the saturation constraint
\begin{align}\label{eq:CH saturation constraint}
    \sum_\mA \phi_\mA = 1.
\end{align}
The saturation constraint is enforced through the mobility structure, namely by requiring that $   \sum_\mA \mathbf{M}_{\mA\mB} =  \sum_\mB \mathbf{M}_{\mA\mB} = 0$, so that $
\sum_\mA \bJ_\mA = 0$. Incorporating the chemical potential as a separate state variable, we now arrive at a system that contains $2N$ equations for $2N$ state variables $
\left(\left\{\phi_\mA\right\}_{\mA=1,\dots,N},\left\{\mu_\mA\right\}_{\mA=1,\dots,N}\right)$:
\begin{subequations}\label{eq:CH sys}
\begin{align}
    \partial_t \phi_\mA + {\rm div}\bJ_\mA &= 0,
    \label{eq:CH sys mass}\\
    \mu_\mA - \dfrac{\partial \Psi}{\partial \phi_\mA}
    + {\rm div}\dfrac{\partial \Psi}{\partial \nabla \phi_\mA} &= 0,
    \label{eq:CH sys chem}
\end{align}
\end{subequations}
for $\mA=1,\dots,N$.

Under appropriate boundary conditions, the system conserves the phase volume fractions, preserves the saturation constraint pointwise, and dissipates the total free energy:
\begin{subequations}
\begin{align}
    \ddt \int_\Omega \phi_\mA\,{\rm d}\Omega &= 0,
    \qquad \mA=1,\dots,N,\\
    \left(\sum_\mA \phi_\mA(\mathbf{x},0)=1\right)
    &\qquad \Rightarrow \qquad
    \left(\sum_\mA \phi_\mA(\mathbf{x},t)=1,\quad t>0\right),\\
    \ddt \mathcal{E}\left(\left\{\phi_\mA\right\}\right)
    &= -\mathcal{D}\left(\left\{\mu_\mA\right\}\right)\le 0,
    \label{eq:CH global energy evolution}
\end{align}
\end{subequations}
where the energy $\mathcal{E}$ and dissipation rate $\mathcal{D}$ are given by
\begin{subequations}
\begin{align}
    \mathcal{E}\left(\left\{\phi_\mA\right\}\right)
    &:= \int_\Omega
    \Psi\left(\left\{\phi_\mA\right\},\left\{\nabla\phi_\mA\right\}\right)\,{\rm d}\Omega,
    \label{eq:CH defEnergy}\\
    \mathcal{D}\left(\left\{\mu_\mA\right\}\right)
    &:= \int_\Omega
    \sum_{\mA,\mB}
    (\nabla \mu_\mA)^T \mathbf{M}_{\mA\mB}\nabla \mu_\mB
    \,{\rm d}\Omega
    \ge 0.
    \label{eq:CH defDissipation}
\end{align}
\end{subequations}

\section{N-phase Navier--Stokes--Cahn--Hilliard model}\label{appendix: NSCH model}

We consider the Navier--Stokes--Cahn--Hilliard system for $N$ phases with non-matching densities introduced in \cite{ten2025unified}. 

\subsection{Governing equations}
Let $\Omega\subset\mathbb{R}^d$, with $d=2,3$, denote the spatial domain. The unknowns are the mass-averaged velocity $\vv:\Omega\to\mathbb{R}^d$, the volume fractions $\phi_\mA:\Omega\to\mathbb{R}$ for $\mA=1,\dots,N$, and the Lagrange multiplier pressure $\lambda:\Omega\to\mathbb{R}$. The governing equations read
\begin{subequations}\label{eq:sys1}
\begin{align}
\partial_t (\rho \bv) + {\rm div} \left( \rho \bv\otimes \bv \right)
+ \sum_\mB \phi_\mB \nabla \mu_\mB + \nabla \lambda
- {\rm div}\boldsymbol{\tau} - \rho \mathbf{b} &= 0,
\label{eq: intro mass: mom}\\
\partial_t \phi_\mA + {\rm div}(\phi_\mA \bv) + \rho_\mA^{-1}{\rm div}\,\bJ_\mA &= 0,
\label{eq: intro mass: mass}
\end{align}
\end{subequations}
supplemented with initial data
\begin{align}
\vv(\mathbf{x},0)=\vv_0(\mathbf{x}),
\qquad
\phi_\mA(\mathbf{x},0)=\phi_{\mA,0}(\mathbf{x}).
\end{align}

The constant specific densities of the phases are denoted by $\rho_\mA$, and the corresponding partial mass densities are $\trho_\mA=\rho_\mA\phi_\mA$. The total mixture density is then given by $\rho=\sum_\mA \trho_\mA$. The body force is taken as $\mathbf{b}=\mathbf{g}=-g\mathbf{j}$, where $g$ is the gravitational constant, $\mathbf{j}=\nabla y$ is the unit vector in vertical direction, and $y$ denotes the vertical coordinate. The viscous stress tensor is assumed to be
\begin{align}
\boldsymbol{\tau}
=
\nu\Bigl(2\nabla^s\bv+\bar{\lambda}({\rm div}\bv)\mathbf{I}\Bigr),
\end{align}
where $\nu$ is the dynamic viscosity, $\bar{\lambda}=-2/d$, and $\nabla^s\bv=(\nabla\bv+(\nabla\bv)^T)/2$
denotes the symmetric part of the velocity gradient.

The Helmholtz free energy density is assumed to belong to the constitutive class
\begin{align}
\Psi=\Psi\left(\left\{\phi_\mA\right\},\left\{\nabla\phi_\mA\right\}\right),
\end{align}
and the corresponding chemical potentials are defined by
\begin{align}\label{eq: chem pot}
\mu_\mA=
\dfrac{\partial\Psi}{\partial\phi_\mA}
-
{\rm div}\dfrac{\partial\Psi}{\partial\nabla\phi_\mA}.
\end{align}
Accordingly, \eqref{eq: intro mass: mom} represents the balance of mixture momentum, while \eqref{eq: intro mass: mass} gives the phase mass balances for $\mA=1,\dots,N$.

\begin{remark}[Mixture velocity]
The system is written in terms of the mass-averaged velocity $\vv$. Alternative formulations based on other mixture velocities are possible; for instance, a volume-averaged formulation is discussed in \cite{ten2025unified}.
\end{remark}

Since voids are excluded, the volume fractions satisfy the saturation constraint
\begin{align}\label{eq: saturation constraint}
\sum_\mA \phi_\mA = 1.
\end{align}
Summing the phase balances \eqref{eq: intro mass: mass} and invoking \eqref{eq: saturation constraint} yields
\begin{align}\label{eq: LM 1}
{\rm div}\bv+\sum_\mA \rho_\mA^{-1}{\rm div}\,\bJ_\mA=0.
\end{align}
This relation is therefore used as the constraint equation associated with saturation. The diffusive fluxes are modeled as
\begin{align}
\bJ_\mA=-\sum_\mB \mathbf{M}_{\mA\mB}\nabla g_\mB,
\end{align}
where $\mathbf{M}_{\mA\mB}$ is the degenerate mobility matrix and $
g_\mA=\rho_\mA^{-1}(\mu_\mA+\lambda)$ is the effective chemical potential.

Introducing the chemical potentials as additional unknowns, the model can be cast as a system of $2N+2$ equations:
\begin{subequations}\label{eq:sys2}
\begin{align}
\partial_t (\rho \bv) + {\rm div} \left( \rho \bv\otimes \bv \right)
+ \sum_\mB \phi_\mB \nabla \mu_\mB + \nabla \lambda
- {\rm div}\boldsymbol{\tau} - \rho\mathbf{b} &= 0,
\label{eq:sys2: mom}\\
\partial_t \phi_\mA + {\rm div}(\phi_\mA \bv) + \rho_\mA^{-1}{\rm div}(\bJ_\mA) &= 0,
\label{eq:sys2: mass}\\
{\rm div}\bv+\sum_\mA \rho_\mA^{-1}{\rm div}\,\bJ_\mA &= 0,
\label{eq:sys2: div}\\
\mu_\mA-\dfrac{\partial\Psi}{\partial\phi_\mA}
+{\rm div}\dfrac{\partial\Psi}{\partial\nabla\phi_\mA} &= 0,
\label{eq:sys2: chem}
\end{align}
\end{subequations}
for $\mA=1,\dots,N$ with $2N+2$ independent variables $
\bv,\left\{\phi_\mA\right\}_{\mA=1,\dots,N},\lambda,\left\{\mu_\mA\right\}_{\mA=1,\dots,N}$. In this formulation, $\lambda$ acts as the Lagrange multiplier associated with the constraint \eqref{eq:sys2: div}.

\subsection{Conservation and dissipation properties}
The system satisfies several fundamental conservation and dissipation properties. Under suitable boundary conditions, it conserves each phase mass, each integrated volume fraction, and the total mass; it also preserves the saturation constraint pointwise and admits a non-increasing total energy:
\begin{subequations}
\begin{align}
\ddt\int_\Omega \trho_\mA\,{\rm d}\Omega &= 0,
\qquad
\ddt\int_\Omega \phi_\mA\,{\rm d}\Omega = 0,
\qquad
\mA=1,\dots,N,
\\
\ddt\int_\Omega \rho\,{\rm d}\Omega &= 0,
\\
\left(\sum_\mA \phi_\mA(\mathbf{x},0)=1\right)
&\qquad\Rightarrow\qquad
\left(\sum_\mA \phi_\mA(\mathbf{x},t)=1,\quad t>0\right),
\\
\ddt \mathcal{E}\left(\vv,\left\{\phi_\mA\right\}\right)
&=
-\mathcal{D}\left(\vv,\left\{g_\mA\right\}\right)
\leq 0.
\label{eq: global energy evolution}
\end{align}
\end{subequations}
Here the total energy and dissipation are given by
\begin{subequations}
\begin{align}
\mathcal{E}\left(\vv,\left\{\phi_\mA\right\}\right)
&:=
\int_\Omega
K\left(\vv,\left\{\phi_\mA\right\}\right)
+
G\left(\left\{\phi_\mA\right\}\right)
+
\Psi\left(\left\{\phi_\mA\right\},\left\{\nabla\phi_\mA\right\}\right)
\,{\rm d}\Omega,
\label{eq:defEnergy}\\
K\left(\vv,\left\{\phi_\mA\right\}\right)
&=
\frac{1}{2}\rho\left(\left\{\phi_\mA\right\}\right)\snorm{\vv}^2,
\\
G\left(\left\{\phi_\mA\right\}\right)
&=
\rho\left(\left\{\phi_\mA\right\}\right)gy,
\\
\mathcal{D}\left(\vv,\left\{g_\mA\right\}\right)
&:=
\int_\Omega
2\nu
\left(
\nabla^s\bv-\frac{1}{d}({\rm div}\bv)\mathbf{I}
\right)
:
\left(
\nabla^s\bv-\frac{1}{d}({\rm div}\bv)\mathbf{I}
\right)
+
\nu\left(\bar{\lambda}+\frac{2}{d}\right)\left({\rm div}\bv\right)^2
\nn\\
&\qquad\qquad
+
\sum_{\mA,\mB}
(\nabla g_\mA)^T \mathbf{M}_{\mA\mB}\nabla g_\mB
\,{\rm d}\Omega
\geq 0.
\label{eq:defDissipation}
\end{align}
\end{subequations}
The energy density in \eqref{eq:defEnergy} thus consists of the kinetic contribution $K$, the gravitational contribution $G$, and the interfacial free energy density $\Psi$.

\subsection{Equilibrium properties}

Equilibrium conditions of the NSCH system are discussed in \cite{structureNphase2026} and yields
\begin{equation}
    \vv=0, \qquad g_\mA=0, \qquad \mathbf{J}_\mA = 0 \qquad \text{ for all } \mA.
\end{equation}
The volume fractions are characterized by
\begin{align}
 &\min_{\phi\in H^1(\Omega)^N} \int_\Omega  \Psi(\{\phi_\beta\},\{\nabla \phi_\beta\})   - \sum_\mA \hat{c}_\mA (\phi_\alpha-m_\alpha) + \lambda(\sum_\mA \phi_\mA - 1) \nn\\
 = &\min_{\substack{\boldsymbol{\phi}\in H^1(\Omega)^N,\\ \la \boldsymbol{\phi},\mathbf{e}_\mA \ra=m_\mA}} \int_\Omega  \Psi(\{\phi_\beta\},\{\nabla \phi_\beta\})   + \lambda(\sum_\mA \phi_\mA - 1) .\label{eq:1Dmin2}
\end{align}

\section{Energy minimizers}\label{appendix sec: minizers}

Here we provide the proofs of the Theorem and Corollaries from \cref{subsec:mass_constraints} and \cref{subsec:face_confinement}.
\begin{theorem}[Equipartition of the energy]\label{thm:equi_form_CHNS}
Let $\boldsymbol{\phi}$ be a one-dimensional equilibrium profile connecting
$\mathbf b^{(\mA)}$ and $\mathbf b^{(\mB)}$, with
$\partial_z\boldsymbol{\phi}\to0$ at both endpoints. Then we have the equipartition:
\begin{align}
    \widetilde{\Psi}_0(\boldsymbol{\phi})
    -
    \widetilde{\Psi}_0(\mathbf{b}^{(\mA)})
    =
    \sum_{\mC,\mD}
    \frac{\kappa_{\mC\mD}}{2}
    \partial_z\phi_\mC \partial_z\phi_\mD,
\end{align}
where
\begin{align}
    \widetilde{\Psi}_0(\boldsymbol{\phi})
    =
    \Psi_0(\boldsymbol{\phi})
    -
    \sum_\mC \hat c_\mC\phi_\mC
\end{align}
whose endpoint values satisfy
$\widetilde{\Psi}_0(\mathbf{b}^{(\mA)})
 =
 \widetilde{\Psi}_0(\mathbf{b}^{(\mB)})$.
\end{theorem}
\begin{proof}
Since $\boldsymbol{\phi}$ is a constrained equilibrium profile, it satisfies the
Euler--Lagrange equations with mass multipliers $\hat c_\mC$ and a saturation
multiplier $\lambda$:
\begin{align}
    \frac{\partial\Psi_0}{\partial\phi_\mC}(\boldsymbol{\phi})
    -
    \sum_\mD \kappa_{\mC\mD}\partial_z^2\phi_\mD
    -
    \hat c_\mC
    +
    \lambda
    =
    0,
    \qquad
    \mC=1,\ldots,N.
\end{align}
Multiplying the equation for phase $\mC$ by $\partial_z\phi_\mC$ and summing
over $\mC$ gives
\begin{align}
    \sum_\mC
    \frac{\partial\Psi_0}{\partial\phi_\mC}
    \partial_z\phi_\mC
    -
    \sum_{\mC,\mD}
    \kappa_{\mC\mD}
    \partial_z^2\phi_\mD\,\partial_z\phi_\mC
    -
    \sum_\mC \hat c_\mC\partial_z\phi_\mC
    =
    0,
\end{align}
where the term with $\lambda$ drops out due to the saturation constraint. This expression can be written as
\begin{align}
    \partial_z
    \left[
    \tilde{\Psi}_0(\boldsymbol{\phi})
    -
    \sum_{\mC,\mD}
    \frac{\kappa_{\mC\mD}}{2}
    \partial_z\phi_\mC\,\partial_z\phi_\mD
    \right]
    =
    0,
\end{align}
where we have used the symmetry $\kappa_{\mC\mD}=\kappa_{\mD\mC}$. Thus the quantity in brackets is constant:
\begin{align}
    \widetilde{\Psi}_0(\boldsymbol{\phi})
    -
    \sum_{\mC,\mD}
    \frac{\kappa_{\mC\mD}}{2}
    \partial_z\phi_\mC\,\partial_z\phi_\mD
    =
    C,
    \label{eq:tilted-first-integral}
\end{align}
Taking the limits $z\to-\infty$ and $z\to+\infty$, and using
$\partial_z\boldsymbol{\phi}\to0$ at both endpoints, gives
\begin{align}
    C
    =
    \widetilde{\Psi}_0(\mathbf{b}^{(\mA)})
    =
    \widetilde{\Psi}_0(\mathbf{b}^{(\mB)}).
\end{align}
\end{proof}

\begin{corollary}[Equal-depth wells]
\label{cor:appendix:equal-depth-equipartitionproof}
Assume the setting of \cref{thm:equi_form_CHNS}. If the two endpoint
values of the bulk potential are equal, i.e. $\Psi_0(\mathbf{b}^{(\mA)})=\Psi_0(\mathbf{b}^{(\mB)})$, then the mass multiplier satisfies
\begin{align}
    \hat{\mathbf c}\cdot(\mathbf{b}^{(\mB)}-\mathbf{b}^{(\mA)})=0.
\end{align}
Moreover, if
\begin{align}
    \hat{\mathbf c}\cdot(\boldsymbol{\phi}(z)-\mathbf{b}^{(\mA)})=0
    \qquad \text{for all }z,
\end{align}
then
\begin{align}
    \Psi_0(\boldsymbol{\phi})
    -
    \Psi_0(\mathbf{b}^{(\mA)})
    =
    \sum_{\mC,\mD}
    \frac{\kappa_{\mC\mD}}{2}
    \partial_z\phi_\mC\,\partial_z\phi_\mD .
\end{align}
If, in addition, the common endpoint value is normalized to zero, then
\begin{align}
    \Psi_0(\boldsymbol{\phi})
    =
    \sum_{\mC,\mD}
    \frac{\kappa_{\mC\mD}}{2}
    \partial_z\phi_\mC\,\partial_z\phi_\mD .
\end{align}
\end{corollary}
\begin{proof}
By \cref{thm:equi_form_CHNS},
\begin{align}
    \Psi_0(\boldsymbol{\phi})
    -
    \Psi_0(\mathbf{b}^{(\mA)})
    -
    \hat{\mathbf c}\cdot(\boldsymbol{\phi}-\mathbf{b}^{(\mA)})
    =
    \sum_{\mC,\mD}
    \frac{\kappa_{\mC\mD}}{2}
    \partial_z\phi_\mC\,\partial_z\phi_\mD .
\end{align}
Since
$\Psi_0(\mathbf{b}^{(\mA)})=\Psi_0(\mathbf{b}^{(\mB)})$ and
$\widetilde{\Psi}_0(\mathbf b_\mA)=\widetilde{\Psi}_0(\mathbf b_\mB)$,
we obtain
\begin{align}
    \hat{\mathbf c}\cdot(\mathbf{b}^{(\mB)}-\mathbf{b}^{(\mA)})=0 .
\end{align}
If, in addition,
$\hat{\mathbf c}\cdot(\boldsymbol{\phi}-\mathbf{b}^{(\mA)})=0$ along the profile,
then the tilted term drops out, giving the stated identity. The normalized
case follows immediately from $\Psi_0(\mathbf{b}^{(\mA)})=0$.
\end{proof}

\begin{corollary}[Equipartition for unequal well depths]
\label{cor:appendix:unequal-depth-tiltproof}
Assume the setting of \cref{thm:equi_form_CHNS}. If the two endpoint
values of the bulk potential are unequal, i.e. 
$\Psi_0(\mathbf{b}^{(\mA)})\neq \Psi_0(\mathbf{b}^{(\mB)})$ then 
\begin{align}
    \Psi_0(\boldsymbol{\phi}) - \Psi_0(\mathbf{b}^{(\mA)})
    - \hat{\mathbf c}\cdot \left(\boldsymbol{\phi}-\mathbf{b}^{(\mA)}\right)  = \sum_{\mC,\mD}
    \frac{\kappa_{\mC\mD}}{2} \partial_z\phi_\mC\,\partial_z\phi_\mD,
\end{align}
where one compatible choice of $\hat{\mathbf c}$ is
\begin{align}
    \hat{\mathbf c} = \frac{\Psi_0(\mathbf{b}^{(\mB)})-\Psi_0(\mathbf{b}^{(\mA)})}{\|\mathbf{b}^{(\mB)}-\mathbf{b}^{(\mA)}\|^2}(\mathbf{b}^{(\mB)}-\mathbf{b}^{(\mA)}).
\end{align}
\end{corollary}

\begin{proof}
With the above choice of $\hat{\mathbf c}$, the tilted potential
\[
    \widetilde{\Psi}_0(\boldsymbol{\phi})
    =
    \Psi_0(\boldsymbol{\phi})
    -
    \hat{\mathbf c}\cdot\boldsymbol{\phi}
\]
has equal endpoint values, i.e. $\widetilde{\Psi}_0(\mathbf b_\mB)=\widetilde{\Psi}_0(\mathbf b_\mA)$. Applying \cref{thm:equi_form_CHNS} to $\widetilde{\Psi}_0$, substituting the definition of $\widetilde{\Psi}_0$ and the explicit expression
for $\hat{\mathbf c}$ gives the claimed identity.
\end{proof}

We show that for the mixture-aware free energy, an energy-minimizing interfacial profile cannot lie exactly on a Gibbs-simplex face.
\begin{theorem}[No face confinement for the mixture-aware free energy]
Let $N\geq 3$. A profile that is confined to a Gibbs-simplex face on a
nontrivial interfacial region cannot be a local minimizer of the total
mixture-aware energy with $W_\mA=W>0$ and finite coefficients
$\chi_{\alpha\beta}$ and $\kappa_{\alpha\beta}$.
\end{theorem}

\begin{proof}
Let $N=3$. Assume that $\boldsymbol{\phi}$ is confined to the $\mA$--$\mB$ face on a nontrivial part of the interface. Since the profile is nontrivial, there is an interval on which $\phi_\mC=0,  0<\phi_\mA,\phi_\mB<1, \phi_\mA+\phi_\mB=1$. Let $\eta$ be a smooth nonnegative perturbation, not identically zero, supported
inside this interval. For sufficiently small $\varepsilon>0$, define
\begin{align}
    \phi_\mA^\varepsilon
    &=
    (1-\varepsilon\eta)\phi_\mA,
    \quad
    \phi_\mB^\varepsilon
    =
    (1-\varepsilon\eta)\phi_\mB,
    \quad
    \phi_\mC^\varepsilon
    =
    \varepsilon\eta, \quad 
    \phi_\mD^\varepsilon
    =
    \phi_\mD
    \quad \text{for } \mD\neq \mA,\mB,\mC .
\end{align}
This perturbation preserves the simplex constraint, since $\sum_\alpha \phi_\alpha^\varepsilon=1$, and for sufficiently small $\varepsilon$ it also preserves nonnegativity. Since $\eta$ is
supported inside the interfacial region, the far-field boundary conditions are
unchanged.

We now compare the energies. The entropy contribution involving
phases $\mA, \mB$ and $\mC$ changes by
\begin{align}
W\int
\Big[
&\phi_\mA^\varepsilon\log\phi_\mA^\varepsilon
-\phi_\mA\log\phi_\mA
+
\phi_\mB^\varepsilon\log\phi_\mB^\varepsilon
-\phi_\mB\log\phi_\mB +
(\varepsilon\eta)\log(\varepsilon\eta)
\Big] {\rm d}z.
\end{align}
Because $\phi_\mA>0$ on the support of $\eta$, the first four terms are $O(\varepsilon)$, while
\begin{align}
(\varepsilon\eta)\log(\varepsilon\eta)
=
\varepsilon\eta\log\varepsilon+\varepsilon\eta\log\eta.
\end{align}
Therefore the entropy change is
\begin{align}
W \varepsilon\log\varepsilon \int_\Omega \eta{\rm d}x + O(\varepsilon).
\end{align}
Since $\eta\ge 0$ and not identically zero, and $W>0$, we have $\int_\Omega \eta{\rm d}x>0$. Hence the leading term is strictly negative. By contrast, the interaction term and the quadratic gradient term vary only by $O(\varepsilon)$, because they are smooth in $\boldsymbol{\phi}$ and $\nabla\boldsymbol{\phi}$ and the perturbation is of size $\varepsilon$. Hence the energy difference is
\begin{align}
\int_\Omega \Psi(\boldsymbol{\phi}^\varepsilon,\nabla\boldsymbol{\phi}^\varepsilon){\rm d}x
-
\int_\Omega \Psi(\boldsymbol{\phi},\nabla\boldsymbol{\phi}){\rm d}x
=
W \varepsilon\log\varepsilon \int_\Omega \eta{\rm d}x + O(\varepsilon).
\end{align}
For sufficiently small $\varepsilon>0$, the negative term $\varepsilon\log\varepsilon$ dominates the $O(\varepsilon)$ remainder, so the perturbed profile has strictly lower energy. Thus $\boldsymbol{\phi}$ is not a local minimizer.\\

The same argument extends directly to $N>3$ by perturbing from any proper Gibbs-simplex face into one absent phase.
\end{proof}

Next, we show that for Boyer's energy every energy-minimizing one-dimensional equilibrium profile connecting $\mathbf{e}_1$ and $\mathbf{e}_2$ lies on the binary face $\{\phi_3=0\}$.
\begin{theorem}[Face confinement for the Boyer free energy]
Let $N=3$. For the Boyer free energy defined in \cite{boyer2006study}, with
positive pair coefficients and nonnegative ternary penalty, every energy-minimizing one-dimensional profile connecting two pure states is confined
to the corresponding Gibbs-simplex face. Moreover, the minimum energy equals the prescribed surface tension of that pair.
\end{theorem}

\begin{proof}
We give the proof for the connection from
$\mathbf e_1=(1,0,0)$ to $\mathbf e_2=(0,1,0)$. Assume $\Sigma_1>0, \Sigma_2>0, \Sigma_3>0, \Lambda\ge 0$. Consider the admissible class \begin{align}
    \mathcal{G}_{12}:=\Bigl\{
\boldsymbol{\phi}\in \mathcal{G}: 
\boldsymbol{\phi}(-\infty)=\mathbf{e}_1,\ \boldsymbol{\phi}(+\infty)=\mathbf{e}_2
\Bigr\},
\end{align}
and the one-dimensional Boyer energy
 \begin{align}
\Psi
:=&~
\int_{\mathbb R}
\left(
\Psi_0(\boldsymbol{\phi})
+
\frac{3\varepsilon}{8}\sum_{\mA=1}^3 \Sigma_\mA |\phi_\mA'|^2
\right){\rm d}z,
\end{align}
where $\Psi_0$ is given by \eqref{eq:Boyerpot}. 
We also use the relation $\gamma_{12} = (\Sigma_1+\Sigma_2)/2$.

We split the remainder of the proof into three steps.\\
\medskip
\noindent\textbf{Step 1: Lower bound for arbitrary ternary connections.} We first decompose the energy into phase-wise contributions. Recalling
\begin{align}
\Psi_0(\phi_1,\phi_2,\phi_3)
=&~\frac{6}{\varepsilon}\left(\Sigma_{1}\phi_1^2(1-\phi_1)^2+\Sigma_{2}\phi_2^2(1-\phi_2)^2 + \Sigma_{3}\phi_3^2(1-\phi_3)^2\right)\nn\\
&~+\frac{12}{\varepsilon}\Lambda \phi_1^2\phi_2^2\phi_3^2,\nn\\
\Sigma_{1}&=\gamma_{12}+\gamma_{13}-\gamma_{23},\qquad
\Sigma_{2}=\gamma_{12}+\gamma_{23}-\gamma_{13},\qquad
\Sigma_{3}=\gamma_{13}+\gamma_{23}-\gamma_{12},
\end{align}
the one-dimensional energy can be written as
\begin{align}
\Psi(\boldsymbol{\phi})
=
\sum_{\mA=1}^3 J_\mA[\phi_\mA]
+
\frac{12\Lambda}{\varepsilon}
\int_{\mathbb R}
\phi_1^2\phi_2^2\phi_3^2{\rm d}z,
\end{align}
where
\begin{align}
J_\mA[\phi_\mA]
:=
\int_{\mathbb R}
\left(
\frac{6\Sigma_\mA}{\varepsilon}\phi_\mA^2(1-\phi_\mA)^2
+
\frac{3\varepsilon\Sigma_\mA}{8}|\phi_\mA'|^2
\right){\rm d}z.
\end{align}
Since $\Sigma_\mA>0$ and $\Lambda\ge 0$, all terms in this decomposition are nonnegative.

By the elementary inequality $a^2+b^2\ge 2ab$, with
\begin{align}
a=\sqrt{\frac{6\Sigma_\mA}{\varepsilon}}\,\phi_\mA(1-\phi_\mA),
\qquad
b=\sqrt{\frac{3\varepsilon\Sigma_\mA}{8}}\,|\phi_\mA'|,
\end{align}
we obtain
\begin{align}\label{eq:J}
J_\mA[\phi_\mA]\ge 3\Sigma_\mA\int_{\mathbb R} \phi_\mA(1-\phi_\mA)|\phi_\mA'|{\rm d}z.
\end{align}
Introducing
\begin{align}
\Phi(s):=\int_0^s t(1-t)\,dt=\frac12 s^2-\frac13 s^3,
\end{align}
we have $\Phi'(s)=s(1-s)$, and hence $
\phi_\mA(1-\phi_\mA)|\phi_\mA'| = \bigl|(\Phi(\phi_\mA))'\bigr|$. Substitution into \eqref{eq:J} provides:
\begin{align}
J_\mA[\phi_\mA]
\ge
3\Sigma_\mA\int_{\mathbb R}\bigl|(\Phi(\phi_\mA))'\bigr|{\rm d}z
\ge
3\Sigma_\mA\left|\Phi(\phi_\mA(+\infty))-\Phi(\phi_\mA(-\infty))\right|.
\end{align}
For a connection from $e_1$ to $e_2$ we have
\begin{subequations}
    \begin{align}
\phi_1(-\infty)=&~1,\quad \phi_1(+\infty)=0,\\
\phi_2(-\infty)=&~0,\quad \phi_2(+\infty)=1,\\
\phi_3(-\infty)=&~0,\quad \phi_3(+\infty)=0.
\end{align}
\end{subequations}
Since $\Phi(1)-\Phi(0)=\frac16$, this yields
\begin{align}
J_1[\phi_1]\ge \frac{\Sigma_1}{2},\qquad
J_2[\phi_2]\ge \frac{\Sigma_2}{2},\qquad
J_3[\phi_3]\ge 0.
\end{align}
Hence every $\boldsymbol{\phi}\in\mathcal{G}_{12}$ satisfies $\Psi(\boldsymbol{\phi})
\ge
(\Sigma_1+\Sigma_2)/2=
\gamma_{12}$. Therefore we conclude:
\begin{align}\label{eq:global-lower-bound}
  \inf_{\boldsymbol{\phi}\in\mathcal{G}_{12}}\Psi(\boldsymbol{\phi})\ge \gamma_{12}.    
\end{align}

\medskip
\noindent\textbf{Step 2: The binary face profile has energy $\gamma_{12}$.}
Restrict now to the face class $
\mathcal{G}_{12}^{\mathrm{face}}
:=
\{\boldsymbol{\phi}\in\mathcal{G}_{12}: \phi_3\equiv 0\}$. Any element of $\mathcal{G}_{12}^{\mathrm{face}}$ has the form $\boldsymbol{\phi}=(\phi,1-\phi,0)$, with $\phi(-\infty)=1$ and $\phi(+\infty)=0$. Since $\phi_3\equiv 0$, the $\Lambda$-term vanishes and
\begin{align}
\Psi_0(\phi,1-\phi,0)
&=
\frac{6}{\varepsilon}
\left(
\Sigma_1\phi^2(1-\phi)^2
+
\Sigma_2(1-\phi)^2\phi^2
\right) \nn\\
&=
\frac{6(\Sigma_1+\Sigma_2)}{\varepsilon}
\phi^2(1-\phi)^2
=
\frac{12\gamma_{12}}{\varepsilon}
\phi^2(1-\phi)^2 .
\end{align}
Moreover,
\begin{align}
\Sigma_1|\phi'|^2+\Sigma_2|(1-\phi)'|^2
=
(\Sigma_1+\Sigma_2)|\phi'|^2
=
2\gamma_{12}|\phi'|^2.
\end{align}
Therefore the restricted energy is exactly the binary energy
\begin{align}
\Psi(\phi,1-\phi,0):=
\int_{\mathbb R}
\left(
\frac{12\gamma_{12}}{\varepsilon}\phi^2(1-\phi)^2
+
\frac{3\varepsilon\gamma_{12}}{4}|\phi'|^2
\right){\rm d}z.
\end{align}

We now show that the minimizing heteroclinic profile in the binary class has energy $\gamma_{12}$.
First, again by $a^2+b^2\ge 2ab$,
\begin{align}
\frac{12\gamma_{12}}{\varepsilon}\phi^2(1-\phi)^2
+
\frac{3\varepsilon\gamma_{12}}{4}|\phi'|^2
\ge
6\gamma_{12}\phi(1-\phi)|\phi'|,
\end{align}
hence
\begin{align}
\Psi(\phi,1-\phi,0)
\ge
6\gamma_{12}\int_{\mathbb R}\phi(1-\phi)|\phi'|{\rm d}z
\ge
6\gamma_{12}\left|\int_0^1 s(1-s)\,ds\right|
=
\gamma_{12}.
\end{align}
Thus
\begin{equation}
\inf_{\boldsymbol{\phi}\in \mathcal{G}_{12}^{\mathrm{face}}}
\Psi(\boldsymbol{\phi})
\ge
\gamma_{12}.
\label{eq:binary-lower-bound}
\end{equation}

Now consider the one-dimensional profile 
\begin{align}
\phi_\star(z)=\frac12
\left(
1-\tanh\left(\frac{2(z-z_0)}{\varepsilon}\right)
\right),
\end{align}
for $z_0\in\mathbb R$. It satisfies $\phi_\star'=-(4/\varepsilon)\phi_\star(1-\phi_\star)$,
and hence
\begin{align}
\frac{3\varepsilon\gamma_{12}}{4}|\phi_\star'|^2
=
\frac{12\gamma_{12}}{\varepsilon}\phi_\star^2(1-\phi_\star)^2.
\end{align}
Substituting and a direct computation provide:
\begin{align*}
\Psi(\phi_\star,1-\phi_\star,0)
&=\frac{24\gamma_{12}}{\varepsilon}
\int_{\mathbb R}\phi_\star^2(1-\phi_\star)^2{\rm d}z=
6\gamma_{12}\int_0^1 s(1-s)\,ds
=
\gamma_{12}.
\end{align*}
Combining this with \eqref{eq:binary-lower-bound}, we conclude that
\begin{align}
\inf_{\boldsymbol{\phi}\in\mathcal{G}_{12}^{\mathrm{face}}}
\Psi(\boldsymbol{\phi})
=
\gamma_{12}.
\end{align}
Hence, together with \eqref{eq:global-lower-bound}, this gives
\begin{align}
\inf_{\boldsymbol{\phi}\in\mathcal{G}_{12}}\Psi(\boldsymbol{\phi})=\gamma_{12}.
\end{align}

\medskip
\noindent\textbf{Step 3: Absent phase $3$.}
Finally, let $\boldsymbol{\phi}^\star$ be a minimizer. Since
$\Psi(\boldsymbol{\phi}^\star)=\gamma_{12}$, while Step 1 gives
\begin{align}
    J_1[\phi_1^\star]\ge \frac{\Sigma_1}{2},
    \qquad
    J_2[\phi_2^\star]\ge \frac{\Sigma_2}{2},
\end{align}
and all remaining terms are nonnegative, equality can hold only if $J_3[\phi_3^\star]=0$. Therefore
\begin{align}
\frac{6\Sigma_3}{\varepsilon}(\phi_3^\star)^2(1-\phi_3^\star)^2=0
\quad\text{and}\quad
\frac{3\varepsilon\Sigma_3}{8}|(\phi_3^\star)'|^2=0.
\end{align}
Since $\Sigma_3>0$, this implies
\begin{align}
(\phi_3^\star)'=0
\quad\text{and}\quad
\phi_3^\star(1-\phi_3^\star)=0,
\end{align}
and hence $\phi_3^\star \in \{0,1\}$. Because $\phi_3^\star(-\infty)=\phi_3^\star(+\infty)=0$, we conclude that $\phi_3^\star\equiv 0$, which proves the claim.
\end{proof}

\subsection*{Funding}
MtE acknowlfaces support from the German Research Foundation (Deutsche Forschungsgemeinschaft DFG), project number 566600860. A.B. acknowlfaces support by the German Research Foundation (DFG) via TRR 146, subproject C3, project number 233630050 and via SPP 2256 within the project ``Variational quantitative phase-field modeling and simulation of powder bed fusion additive manufacturing'' project number 441153493. MtE and AB acknowlfaces support by the RMU via the joined project "MULTIMIX".

\bibliography{jfm}

\end{document}

%% file: declarations.tex
\def\la{\langle}
\def\ra{\rangle}

\def\nn{\nonumber}

\usepackage{mathtools}

\DeclarePairedDelimiter{\norm}{\|}{\|}
\DeclarePairedDelimiter{\snorm}{|}{|}

\def\vv{\mathbf{v}}

\def\ddt{\frac{\mathrm{d}}{\mathrm{d}t}}

\def\softd{{\leavevmode\setbox1=\hbox{d}%
\hbox to 1.05\wd1{d\kern-0.4ex{\char039}\hss}}}

\newcommand{\mA}{\alpha}
\newcommand{\mB}{{\beta}}
\newcommand{\mC}{{\gamma}}
\newcommand{\mD}{{\delta}}

\newcommand{\bJ}{\mathbf{J}}

\newcommand{\bv}{\mathbf{v}}
\newcommand{\trho}{\tilde{\rho}}